\documentclass[11pt]{article}
\usepackage{stmaryrd}
\usepackage{amsfonts}
\usepackage{bbm}
\usepackage{float}
\usepackage{amscd}   
\usepackage{mathrsfs}
\usepackage{latexsym,amssymb,amsmath,amscd,amscd,amsthm,amsxtra,xypic}
\usepackage[dvips]{graphicx}
\usepackage[utf8]{inputenc}
\usepackage[T1]{fontenc}  
\usepackage{lmodern}
\usepackage{amssymb}
\usepackage[all]{xy}
\usepackage{nicefrac,mathtools,enumitem}
\usepackage{microtype}
\usepackage{lipsum}
\usepackage{mathtools}
\usepackage{xfrac}
\usepackage[colorinlistoftodos]{todonotes}

\usepackage{geometry}
\usepackage[hidelinks]{hyperref}

\usepackage{times}
\usepackage{amsthm}
\usepackage{amsmath}
\usepackage{amsfonts}
\usepackage{amssymb}
\usepackage{mathdots}
\usepackage{color}
\usepackage{mathrsfs}
\usepackage{stmaryrd}
\SetSymbolFont{stmry}{bold}{U}{stmry}{m}{n}

\usepackage{bm}

\usepackage{tikz-cd}
\usepackage{tikz}
\usepackage{graphicx}
\usepackage{subfigure}
\usepackage[toc,page]{appendix}
\usepackage[usestackEOL]{stackengine}

\newcommand{\be }{\begin{equation}}
\newcommand{\ee }{\end{equation}}

\newcommand{\Real}{\mathbb R}
\newcommand{\Comp}{\mathbb C}

\newcommand{\abs}[1]{\left\vert#1\right\vert}

\newcommand{\h}{\mathfrak h}

\newcommand{\frkt}{\mathfrak t}

\newcommand{\U}{{\rm U}}

\newcommand{\half}{\frac{1}{2}}

\newcommand{\br}[1]{   [ \cdot,    \cdot  ]   }

\newcommand{\g}{\mathfrak g}

\newcommand{\Hom}{\mathrm{Hom}}

\newcommand{\tr}{\mathrm{tr}}

\newcommand{\CG}{{\mathcal{G}}}

\newcommand{\diag}{\mathrm{diag}}

\newcommand {\eps}{\varepsilon}

\newcommand{\Herm}{{\rm Herm}}

\usepackage{thmtools}

\newcommand{\I}{{\mathrm{i}}}

\newcommand{\fgl}{{\mathfrak {gl}}}
\newcommand{\fpgl}{{\mathfrak {pgl}}}
\newcommand{\fsl}{{\mathfrak {sl}}}

\newcommand{\PGL}{\mathrm{PGL}}
\newcommand{\GL}{\mathrm{GL}}
\newcommand{\ab}{\mathrm{ab}}

\newcommand{\cE}{{\mathcal E}}
\newcommand{\cP}{{\mathcal P}}
\newcommand{\cM}{{\mathcal M}}

\newcommand{\Sol}{{\rm Sol}}

\newcommand{\bbZ}{{\mathbb Z}}
\newcommand{\bbQ}{{\mathbb Q}}

\newcommand{\cW}{{\mathcal W}}

\newcommand{\fh}{\mathfrak{h}}
\newcommand{\reg}{{\mathrm{reg}}}
\newcommand{\rme}{\mathrm{e}}
\newcommand{\fX}{{\mathfrak{X}}}

\newcommand{\IP}[1]{\langle#1\rangle}
\DeclareMathOperator{\re}{Re}
\DeclareMathOperator{\im}{Im}
\DeclareMathOperator{\Path}{Path}
\DeclareMathOperator{\Lift}{Lift}
\DeclareMathOperator{\Nab}{Nab}

\declaretheoremstyle[spaceabove=0.25cm,spacebelow=0.25cm,notefont=\normalfont\bfseries, notebraces={(}{)}]{theorem}
\declaretheoremstyle[spaceabove=0.25cm,spacebelow=0.25cm,bodyfont=\normalfont,notefont=\normalfont\bfseries, notebraces={(}{)}]{noital}
\declaretheoremstyle[spaceabove=0.25cm,spacebelow=0.25cm,bodyfont=\normalfont\color{darkgreen},notefont=\normalfont\bfseries, notebraces={(}{)}]{green}
\declaretheoremstyle[spaceabove=0.25cm,spacebelow=0.25cm,bodyfont=\normalfont,notefont=\normalfont\bfseries,qed=$\qedsymbol$,notebraces={(}{)}]{proofstyle}

\declaretheorem[name=Theorem,numberwithin=section,style=theorem]{thm}
\declaretheorem[name=Conjecture,sibling=thm,style=theorem]{conj}
\declaretheorem[name=Proposition,sibling=thm,style=theorem]{pro}

\declaretheorem[name=Corollary,sibling=thm,style=theorem]{cor}
\declaretheorem[name=Lemma,sibling=thm,style=theorem]{lem}
\declaretheorem[name=Definition,sibling=thm,style=noital]{defi}
\declaretheorem[name=Example,sibling=thm,style=noital]{ex}
\declaretheorem[name=Remark,sibling=thm,style=theorem]{rmk}

\usetikzlibrary{matrix,arrows,decorations.pathmorphing,decorations.markings,arrows.meta,calc,backgrounds,plotmarks,intersections}

\definecolor{darkgreen}{rgb}{0.0, 0.5, 0.0}
\definecolor{color12lines}{rgb}{0.5, 0.0, 0.0}
\definecolor{color13lines}{rgb}{0.0, 0.5, 0.0}
\definecolor{color23lines}{rgb}{0.0, 0.0, 0.6}

\tikzset{
    branchpoint/.style={orange, thick, mark=x, mark options={orange, line width=1.25pt}}
}

\tikzset{
    singularpoint/.style={blue, mark=*, mark options={blue, mark size=1.25pt}}
}

\tikzset{
    stokeslabel/.style={gray!95,font=\tiny}
}

\tikzset{
    cutlabel/.style={orange,font=\tiny}
}

\tikzset{
    sheetlabel/.style={font=\small}
}

\tikzset{
    branchcut/.style={orange,dashed,semithick}
}

\tikzset{
    disccolor/.style={gray!12}
}

\tikzset{
    wall/.style={black,thick}
}

\tikzset{
    path/.style={thick,rounded corners}
}

\tikzset{
    witharrow/.style={
        decoration={markings, mark=at position #1 with {\arrow[sloped]{Latex}}},
        postaction={decorate}
    }
}

\tikzset{
    antistokesmark/.style={semithick, black, radius=1.5pt, fill=yellow}
}

\tikzset{
    withbackgroundrectangle/.style={show background rectangle, background rectangle/.style={fill=gray!7}}
}

\numberwithin{equation}{subsection}

\title{WKB asymptotics of Stokes matrices, spectral curves and rhombus inequalities}
\author{Anton Alekseev, Andrew Neitzke, Xiaomeng Xu and Yan Zhou}
\date{}

\newcommand{\Addresses}{{
  \bigskip
  \footnotesize
\noindent \textsc{Department of Mathematics, Universit\'e de Geneve, 7-9 rue du
Conseil G\'en\'eral, Case postale 64, 1211 Geneve 4, Switzerland}\par\nopagebreak
  \textit{E-mail address}: \texttt{anton.alekseev@unige.ch}\\
  \textsc{Department of Mathematics, Yale University, PO Box 208283, New Haven, CT 06511, United States of America}\par\nopagebreak
  \textit{E-mail address}: \texttt{andrew.neitzke@yale.edu}\\
  \textsc{School of Mathematical Sciences \& Beijing International Center
for Mathematical Research, Peking University, Beijing 100871, China}\par\nopagebreak
  \textit{E-mail address}: \texttt{xxu@bicmr.pku.edu.cn}\\
  \textsc{Department of Mathematics, Northeastern University, 360 Huntington Avenue, MA 02115, United States of America}\par\nopagebreak
  \textit{E-mail address}: \texttt{y.zhou@northeastern.edu}
}}
\begin{document}

\maketitle
\begin{abstract}
We consider an $n\times n$ system of ODEs on $\mathbb{P}^1$ with a simple pole $A$ at $z=0$ and a double pole $u={\rm diag}(u_1, \dots, u_n)$ at $z=\infty$. 
This is the simplest situation in which the monodromy data of the system are described by upper and lower triangular Stokes matrices $S_\pm$, and we impose reality conditions which imply $S_-=S_+^\dagger$.
We study leading WKB exponents of Stokes matrices in parametrizations given by generalized minors and by spectral coordinates, and 
we show that for $u$ on the caterpillar line (which corresponds to the limit $(u_{j+1}-u_j)/(u_j- u_{j-1}) \to \infty$ for $j=2, \cdots, n-1$), the real parts of these exponents are
given by periods of certain cycles on the degenerate spectral curve $\Gamma(u_{\rm cat}(t), A)$.

These cycles admit unique deformations for $u$ near the caterpillar line. Using the spectral network theory, we give for $n=2$, and $n=3$ exact WKB predictions for asymptotics of generalized minors in terms of periods of these cycles. Boalch's theorem from Poisson geometry implies that real parts of leading WKB exponents satisfy the rhombus (or interlacing) inequalities. We show that these inequalities are in correspondence with finite webs of the canonical foliation 
on the root curve $\Gamma^r(u, A)$, and that they follow from the positivity of the corresponding periods. We conjecture that a similar mechanism applies for $n>3$.

We also outline the relation of the spectral coordinates
with the cluster structures considered by Goncharov-Shen, and with ${\mathcal N}=2$ supersymmetric quantum field theories in dimension four
associated to some simple quivers.

\end{abstract}

\tableofcontents

\section{Introduction}

\subsection{Meromorphic ODEs, Stokes matrices, and WKB expansion}

In this section we briefly introduce the background material of our study: Stokes matrices for meromorphic ODEs, their WKB asymptotics, and their relation to Poisson geometry and spectral curves.

\subsubsection{Meromorphic ODEs}

The main object of study in this paper is the following system of ODEs on $\mathbb{P}^1\backslash \{ 0, \infty\}$:
\begin{eqnarray}\label{eq:intro-main-ode}
\eps \frac{dF}{dz} = \left(\I  u-\frac{1}{2\pi\I}\frac{A}{z}\right)F \, .
\end{eqnarray}
Here $\eps$ is a parameter, $u \in \h_{\rm reg}(\mathbb{C})$ is an $n\times n$ diagonal matrix with distinct diagonal entries and $A \in \mathfrak{gl}_n(\mathbb{C})$ is an $n \times n$ matrix.
Alternatively, we can view \eqref{eq:intro-main-ode} as coming from the meromorphic flat connection 
$$
\nabla_{(u, A, \eps)} =  d - \frac{1}{\eps} \left(\I  u-\frac{1}{2\pi\I}\frac{A}{z}\right)dz
$$
on the trivial rank $n$ vector bundle over $\mathbb{P}^1$.
This class of Poincar\'e rank 1 ODEs naturally appears in various contexts in geometry and representation theory, \emph{e.g.} the study of Frobenius manifolds \cite{Dubrovin} and in particular the quantum cohomology of Fano manifolds \cite{GGI}, linearization in Poisson geometry \cite{Boalch}, quantum Weyl group actions on Poisson groups \cite{BoalchG}, stability conditions \cite{Bridgeland, BTL}, Yang-Baxter equations \cite{Xu2} \emph{etc}. These ODEs also provide simple local models for the wild non-abelian Hodge correspondence on curves \cite{BB}.

In most of the paper, we impose the following \emph{reality conditions}:  
\begin{equation} \label{eq:intro-reality}
\eps \in \mathbb{R}_+, \hskip 0.3cm
u \in \h_{\rm reg}(\mathbb{R}), \hskip 0.3cm
A \in \Herm(n),
\end{equation}
where $\Herm(n)$ is the space of $n\times n$ Hermitian matrices. We will see that these conditions are natural from the point of view of tropicalization of dual Poisson-Lie groups \cite{ABHL, ALL}.

\subsubsection{Stokes matrices and WKB}

The ODE \eqref{eq:intro-main-ode} has a first-order pole at $z=0$ and a second-order pole at $z=\infty$. Hence, its monodromy data is encoded in \emph{Stokes matrices}. Under reality conditions \eqref{eq:intro-reality}, there are two Stokes supersectors with canonical solutions $F_\pm(z,\eps)$ specified by a prescribed asymptotic behavior at $z \to \infty$. The Stokes matrices $S_{\pm}(u,A,\eps)$ are transition matrices between these two canonical solutions. Up to permutations of indices, $S_+$ and $S_-$ are upper and lower triangular matrices, respectively. The reality conditions imply that they are Hermitian conjugates of each other:
\begin{equation}
S_-(u, A, \eps) = S_+(u, A, \eps)^\dagger.
\end{equation}

We will be interested in the asymptotic behavior of $S_\pm$ in the limit of $\eps \to 0$. In order to describe these asymptotics, we use several parametrizations of Stokes matrices. One such parametrization is given by the \emph{generalized minors} (see e.g. \cite{ABHL, ALL}):
for $1\le i\le k\le n$, define $\Delta^{(k)}_i(S_+)$ to be the determinant of the $i\times i$ submatrix formed by taking the first $i$ rows and $(k-i + 1)^{th}$ to $k^{th}$ columns of the upper-triangular matrix $S_+$.
Other possible parametrizations include matrix entries $(S_+)_{i,j}$ and spectral coordinates $X_\gamma$ to be defined later. For any of these parametrizations, 
if $(u,A)$ are generic, one expects an asymptotic expansion as $\eps \to 0$:
\begin{equation} \label{eq:intro-WKB-series}
\log X(u, A, \eps) \sim \eps^{-1} x(u, A) + x_0(u,A) + \eps x_1(u,A) + \cdots
\end{equation}
Here $X(u, A, \eps)$ is one of the parameters (a generalized minor, a spectral coordinate, etc),
$x(u, A)$ is the leading WKB exponent, and $x_i(u, A)$ are higher-order coefficients in the WKB expansion. 

In this paper, we will be mostly concerned with the leading WKB exponents, so
we often truncate \eqref{eq:intro-WKB-series} to the leading term
$$
\log X(u, A, \eps) \sim \eps^{-1} x(u, A) \, .
$$
In the case of a generalized minor $X = \Delta^{(k)}_i$ we denote the leading WKB exponent by $\delta^{(k)}_i$, 
and its real part by 
\begin{equation}
l^{(k)}_i = {\rm Re} \, \delta^{(k)}_i \, .    
\end{equation}

\subsubsection{Poisson structures and inequalities}

Next let us discuss some constraints on WKB asymptotics, which arise from Poisson geometry. 

Equip the set of data $(u, A)$ with the linear Kirillov-Kostant-Souriau (KKS) Poisson structure on $A$ and the vanishing Poisson structure on $u$. Then, by a theorem of Boalch \cite{Boalch}, the Stokes matrices $S_\pm(u,A,\eps)$ satisfy the Poisson brackets of the dual Poisson-Lie group $U(n)^*$, rescaled by $\eps^{-1}$. 
These brackets are polynomial in the matrix entries, and Laurent polynomial in the generalized minors.

These facts alone already imply an interesting corollary for WKB asymptotics, as follows.
Suppose we have two parameters $X_1$, $X_2$, either matrix entries or generalized minors, which 
have the WKB behavior $\log X_1\sim \eps^{-1} x_1$,
$\log X_2 \sim \eps^{-1} x_2$. Then, we have
\begin{equation} \label{Poisson}
e^{\frac{1}{\eps}(x_1+x_2) + \cdots} \left(\frac{1}{\eps^2} \, \{ x_1, x_2\} + \cdots\right) =
\{ e^{\frac{1}{\eps} x_1 + \dots}, e^{\frac{1}{\eps} x_2 + \dots}\} =\{ X_1, X_2\} = \frac{1}{\eps} \,  \sum_a e^{\frac{1}{\eps} x_a + \dots},
\end{equation}
where the $x_a$ are the leading WKB exponents of the terms $X_a$ appearing in $\{ X_1, X_2 \}$.
In \eqref{Poisson}, the right-hand side stands for the Laurent polynomial expression for the Poisson bracket $\{ X_1, X_2\}$, and the left-hand side stands for the KKS bracket of the WKB expansions of $X_1$ and $X_2$.
We observe that the left and right sides of \eqref{Poisson} can match only if 
\begin{equation} \label{Poisson>}
{\rm Re}\, (x_1 + x_2) \geq {\rm Re} \, x_a \hskip 0.3cm {\rm for} \hskip 0.2cm {\rm all} \hskip 0.3cm a.
\end{equation}
The inequalities \eqref{Poisson>} play an important role in this paper. For the generalized minors, they read
(see \cite{AD} for details)
\begin{equation}       \label{rhombus}
l^{(k+1)}_i + l^{(k)}_{i-1} \geq l^{(k+1)}_{i+1} + l^{(k)}_{i}, \hskip 0.3cm
l^{(k+1)}_{i} + l^{(k)}_{i} \geq l^{(k+1)}_{i+1} + l^{(k)}_{i-1},
\end{equation}
and we refer to them as \emph{rhombus inequalities}. 

It is sometimes convenient to introduce a linear change of variables, defining $\lambda^{(k)}_j$ for $1 \le j \le k \le n$ 
by $\lambda^{(k)}_j = l^{(k)}_{k-j+1} - l^{(k)}_{k-j}$ (where we let $l^{(k)}_0 =0$). Then
$l^{(k)}_i = \sum_{j=k-i+1}^k \lambda^{(k)}_j$.
After this change of variables, the rhombus inequalities become
the \emph{interlacing inequalities}:
\begin{equation}         \label{interlacing}
\lambda^{(k+1)}_j \le \lambda^{(k)}_j \le \lambda^{(k+1)}_{j+1}.
\end{equation}

\subsubsection{Spectral curves}

The general philosophy of the WKB method predicts that the asymptotic series \eqref{eq:intro-WKB-series} 
should be determined by computations on the \emph{spectral curve}.

\begin{defi}
The spectral curve $\Gamma(u, A)$
is given by the characteristic polynomial 
\begin{eqnarray} \label{eq:sp-intro}
\mathrm{det} \left[ \mu \cdot I_n - \left(\I u-\frac{1}{2\pi\I}\frac{A}{z}\right) \right]=0. 
\end{eqnarray}
\end{defi}

The spectral curve $\Gamma(u,A)$ is an $n$-to-$1$ cover of $\mathbb{P}^{1}\setminus \left\{ 0,\infty\right\}$ with punctures over $0$ and $\infty$. 
The \emph{period} of the canonical 1-form $\omega = \mu(z) \, dz$ on a 1-cycle $\gamma$ is
\begin{equation} \label{eq:def-period}
    Z(\gamma) = \oint_\gamma \omega \, .
\end{equation}

\subsection{The main conjecture}

The leading-order WKB prediction takes the following form:
\begin{conj}\label{mainconj}
For generic $u$ and $A$, there exist cycles $\left\{L^{(k)}_i \right\}_{1\le i\le k\le n }$ on $\Gamma(u,A)$ such that the asymptotic
expansion \eqref{eq:intro-WKB-series} holds, with leading 
coefficient
\begin{equation}     \label{eq:mainconj}
l^{(k)}_i = -\frac12 Z \left({L^{(k)}_i} \right).
\end{equation}
Moreover, the quantities $- \frac12 Z \left({L^{(k)}_i} \right)$ satisfy
the rhombus inequalities \eqref{rhombus}.
\end{conj}

In \eqref{eq:mainconj}, the minus sign is a consequence of our conventions, with no essential
significance. The factor $\frac12$ arises because we consider $l^{(k)}_i = \frac12 \left(\delta_i^{(k)} + \overline \delta_i^{(k)}\right)$
instead of $\delta_i^{(k)}$.

To address \autoref{mainconj}, we need a method of identifying the cycles $L_i^{(k)}$.
We approach this by looking at a region in parameter space
where we have a detailed understanding of the spectral curves $\Gamma(u,A)$. Namely, 
we consider the following limit in the space of parameters $u$:
\begin{equation} \label{caterpillar}
u_2 - u_1 = t \in \mathbb{R}_+, \hskip 0.3cm \frac{u_{j+1}-u_{j}}{u_{j}-u_{j-1}}\rightarrow +\infty \,\, {\rm for} \,\, j=2,\dots,n-1.
\end{equation}
This limit makes sense in the \emph{De Concini-Procesi space} \cite{dCP} $\widehat{\h_{\rm reg}}(\mathbb{R})$.
We refer to configurations \eqref{caterpillar} as lying on the \emph{caterpillar line} and denote them $u_{\rm cat}(t)$.\footnote{This terminology is inspired by the term \emph{caterpillar point} in the existing literature (e.g. \cite{HKRW, Sp}) which refers to the point $t=0$;
when we refer to $u_{\rm cat}(t)$ in this paper, though, we always mean $t > 0$.} We refer readers to \autoref{subsect-dCP} for a more detailed description of the De Concini-Procesi space and the stratifications of its boundary.

\subsection{The proof on the caterpillar line}

We first investigate the WKB behavior of Stokes matrices on the caterpillar line.
Fortunately, we have simple explicit formulas for the behavior of $S_+(u,A,\eps)$ as $u \to u_{\rm cat}(t)$
(see \autoref{sec:caterpillar_section2} and \autoref{app:Stokes_3} for more details). These formulas imply in particular (see \autoref{thm:conj_at_ucat} for the precise statement): 
\begin{thm}        \label{thm:intro-caterpillar-exponents}
The real parts of WKB exponents of minor coordinates, 
$l_i^{(k)}(u,A)$, have a limit as $u \to u_{\rm cat}(t)$, given by
\begin{equation}
l_i^{(k)}(u_{\rm cat}(t),A) = \sum_{j=k-i+1}^k \lambda^{(k)}_j
\end{equation}
where $\lambda^{(k)}_j$
are the eigenvalues of the $k \times k$ upper left corner submatrix $A^{(k)}$ of $A$,
ordered by $\lambda^{(k)}_1 < \lambda^{(k)}_2 < \cdots < \lambda^{(k)}_k$. 
\end{thm}

We remark that in this limit the interlacing inequalities \eqref{interlacing} 
become familiar: they are the Cauchy interlacing inequalities
obeyed by the eigenvalues of submatrices of the Hermitian matrix $A$.
From this point of view, the interlacing inequalities obeyed by
the WKB exponents at general $u$ 
are some interesting deformation of the usual Cauchy
interlacing inequalities.

Moreover, the spectral curve degenerates in this limit:
\begin{thm}  (\autoref{thm:smo-nei}, \autoref{cor:vani})     \label{thm:intro-curves}
The family of curves $\Gamma(u,A)$ can be extended to include a curve
$\Gamma(u_{\rm cat}(t), A)$ over $u_{\rm cat}(t)$, which is reducible, with components $\Gamma^{(2)}, \dots, \Gamma^{(n)}$. For $A$ in the generic locus as in \autoref{def:discr}, each component $\Gamma^{(k)}$ is isomorphic to $\mathbb{P}^1$ with $2k$ punctures. In this case, there are loops $V_{j}^{(k)}$ $(u_{\rm cat}(t), A)$ $(1 \leq j \leq k \leq n)$ around punctures of $\Gamma(u_{\rm cat}(t), A)$ such that 
\begin{equation}\label{eq:intro-per}
    Z \left({V^{(k)}_j(u_{\rm cat}(t), A)} \right) = -  \lambda^{(k)}_j \, .
\end{equation}
\end{thm} 
Equation \eqref{eq:intro-per} motivates the following: define the \emph{distinguished cycles} by
\begin{equation} \label{eq:intro-distinguished-cycles}
C^{(k)}_i(u_{\rm cat}(t), A) = \sum_{j=k-i+1}^k V^{(k)}_j(u_{\rm cat}(t), A) \, .    
\end{equation}
Then combining \autoref{thm:intro-caterpillar-exponents} and \autoref{thm:intro-curves} we obtain
a proof of \autoref{mainconj} at $u = u_{\mathrm{cat}}(t)$,
with the cycles given by $L_i^{(k)} = C_i^{(k)}$.

\subsection{Exact WKB and spectral networks}

Next, we investigate \autoref{mainconj} near the caterpillar line. Then, the spectral curve admits the following description:
\begin{thm}  (\autoref{cor:vani})    \label{intro:smooth_curve} 
Let $U_{\rm id}$ be the connected component of $\mathfrak{h}_{\rm reg} (\mathbb{R})$ labeled by the identity element in the symmetric group $S_n$ as in \eqref{eq:com-u}. Fix $A$ in the generic locus as in \autoref{def:discr}. Then, there exists a punctured open neighborhood $B(u_{\rm cat} (t), A)$ of $u_{\rm cat} (t)$ in $U_{\rm id}$ such that  for $u \in B(u_{\rm cat} (t), A)$, $\Gamma(u,A)$ is smooth and
there are vanishing cycles $V^{(k)}_j(u, A)$ $(1 \leq j \leq k \leq n)$ on $\Gamma(u,A)$ that become $V^{(k)}_j(u_{\rm cat} (t), A)$ as $u \rightarrow u_{\rm cat} (t)$.
\end{thm}
As before, we will denote
\begin{equation}       \label{eq:Lki(u)}
C^{(k)}_i(u, A)= \sum_{j=k-i+1}^k V^{(k)}_i(u, A).
\end{equation}
We study \autoref{mainconj} using the \emph{exact WKB} method. Generally, 
given a family of ODEs of the form 
\begin{equation}        \label{exactWKB}
\eps \partial_z \psi = a(\eps) \psi, \hskip 0.3cm a(\eps) = a_0 + \eps a_1 + \cdots,
\end{equation}
exact WKB predicts that the $\eps \to 0$ asymptotics of the monodromy/Stokes data will be determined
by computations that take place on the spectral curve of $a_0$.
In particular, there are distinguished \emph{spectral coordinates} $X_\gamma$ on the moduli space 
of monodromy/Stokes data, labeled by cycles and open paths $\gamma$ on the spectral curve. 
Exact WKB predicts that each spectral coordinate $X_\gamma$
has leading WKB exponent given by a corresponding period $Z(\gamma)$.
The construction of the spectral coordinates which we employ
uses the \emph{spectral networks} $\cW(u,A,\vartheta=0)$ as described in \cite{GMN2,HN}
and briefly reviewed in \autoref{app:sn-review}.

The linear system  \eqref{eq:intro-main-ode} clearly fits into the framework of \eqref{exactWKB}, with
$$
a_0 = \I u-\frac{1}{2\pi\I}\frac{A}{z}
$$
and the spectral curve $\Gamma(u, A)$. We prove the following:
\begin{thm} \label{thm:intro3}
For $n=2, 3$, and for $u$ near the caterpillar line and $A$ near diagonal, the exact WKB prediction implies 
\autoref{mainconj}, with the cycles $L_i^{(k)}$ given by
the $C_i^{(k)}$ of \eqref{eq:Lki(u)}.
\end{thm}

Our proof requires us to describe the spectral networks $\cW(u,A,\vartheta=0)$.
For $n=2$ we do this analytically, but our proof for $n=3$ involves some computer assistance.
The obstacle to dealing with $n > 3$ has to do with the combinatorics of the spectral networks:
some correction terms interfere with the translation from spectral coordinates to
minor coordinates, and one needs to show that these correction terms 
do not affect the leading asymptotics. The precise desired statement is 
formulated in \autoref{conj:combi}. 
Modulo this obstacle, the extension of \autoref{thm:intro3} to 
arbitrary rank is given in \autoref{thm:summary-sn-story}. 

We remark that the exact WKB prediction is stronger than \autoref{mainconj}: 
it predicts the full leading WKB exponent for generalized minors, both real and imaginary parts.
We verify this prediction directly in the case $n=2$, while for $n > 2$ it is a conjecture.

The exact WKB picture also gives a new geometric way of understanding the rhombus inequalities. Indeed, given
the spectral network $\cW(u,A,\vartheta=0)$ one has the notion of \emph{finite web}:
this is a certain network of trajectories in $\cW(u,A,\vartheta=0)$. Each finite web 
determines a $1$-cycle $\gamma$ on $\Gamma(u,A)$, the \emph{charge} of the web. 
Whenever $\gamma$ arises in this way, the corresponding period is negative $Z({\gamma}) < 0$. We show:
\begin{thm}
For $n=2$ and $n=3$, and for $u$ near the caterpillar line and $A$ near diagonal, 
each rhombus inequality obeyed by the WKB exponents 
corresponds to a finite web in $\cW(u,A,\vartheta)$, with a charge $\gamma$; 
the rhombus inequality is the fact that $Z(\gamma(u,A)) < 0$.
\end{thm}

The exact WKB story which we develop here connects the 
equation \eqref{eq:intro-main-ode} to various other areas of mathematics
and physics. 
For instance, there is a known connection between exact WKB and cluster algebras
as described in e.g. \cite{GMN,IN}: the spectral coordinates 
$X_\gamma$ are often identified with cluster coordinates
on the moduli space of monodromy/Stokes data. The spectral networks that appear most directly in our WKB analysis 
are of a degenerate kind (ultimately because of the reality condition we impose on $u, A$ and $\eps$), so
they don't directly give cluster coordinates; but, by perturbing $\eps$ to be complex, we can reach more
generic spectral networks, which do give cluster coordinates. In the case $n=2$ we work out the details,
and show that the cluster structure which arises is the one considered by Goncharov-Shen in \cite{GoSh}.

In another direction, there is a relation \cite{GMN,GMN2} 
between exact WKB analysis of ODEs and ${\mathcal N}=2$
supersymmetric field theories in four dimensions. 
In the particular example of \eqref{eq:intro-main-ode},
the relevant supersymmetric field theory is 
relatively simple: it is a Lagrangian gauge theory
given by a certain quiver as described in \cite{GMNY}. 
We explain this briefly in \autoref{sec:qft}.

\subsection{Future directions}

\begin{itemize}
\item In our discussion of spectral networks we left one important problem unsolved: how to prove \autoref{conj:combi}, and thus 
extend \autoref{thm:intro3} beyond the cases $n=2$ and $n=3$? It seems possible that this will
require some new insight into the structure of spectral networks for higher rank.

\item We explored the connection between WKB for complex $\eps$ and the Goncharov-Shen cluster structure 
in detail for $n=2$, and briefly for $n=3$. It would be interesting to develop that connection in more detail
in the $n \ge 3$ case.

\item The connection between the ODE \eqref{eq:intro-main-ode} 
and supersymmetric quantum field theory, 
briefly discussed in \autoref{sec:qft}, has various consequences which would be 
interesting to explore. For example, it implies that one should be able to 
give an explicit representation of a $\tau$ function governing the isomonodromy 
deformations of \eqref{eq:intro-main-ode}, building it out of Nekrasov's 
partition function for the gauge theory \cite{Nekrasov}.
This would be an analogue of the celebrated ``Kiev formula'' \cite{Gamayun,ILT}
for rank $2$ systems with regular singularities.

\item The data of periods and Donaldson-Thomas invariants which we describe in \autoref{sec:spectral-networks} determines a Riemann-Hilbert problem in the $\eps$-plane as explained in \cite{GMN,Gaiotto,Bridgeland}. The solution of this Riemann-Hilbert problem should give the exact Stokes
matrices.

\item In this paper we focus mainly on the situation where $u$ is a real diagonal matrix and $A$ is Hermitian. For fixed $u$, the map that takes $A$ to its WKB exponents (or equivalently the real parts of periods) is an example of a real integrable system, whose base is the Gelfand-Zeitlin cone cut out by the rhombus inequalities. If we also fix the conjugacy class of $A$ then we again have a real integrable system $\cM$, now with total space the flag manifold of $SL(n)$, and base a (compact) Gelfand-Zeitlin polytope.

If we instead allow $A$ to be complex, then we obtain a complex integrable system $\cM_\Comp$
containing $\cM$. We expect that $\cM_{\Comp}$ can be described as a moduli
space of framed wild Higgs bundles over ${\mathbb P}^1$.
Such moduli spaces are expected to admit hyperk\"ahler metrics at least in some open subsets,
though the only example which is well understood so far is the ``Ooguri-Vafa manifold''
\cite{OV,GMN,Tulli}.
We expect that $\cM_\Comp$ 
carries a hyperk\"ahler metric in an open neighborhood of $\cM$, and this
metric should be a deformation of
Kronheimer's hyperk\"ahler metric on the cotangent bundle to the flag manifold.
This can be confirmed directly in the case $n=2$, where the relevant metric is of Gibbons-Hawking
type, rather similar to Ooguri-Vafa; it should be interesting to explore it for $n>2$.

\end{itemize}

\subsection{Outline of the paper}

In \autoref{beginsection}, we state the relation between Poisson brackets and WKB asymptotics, and we show that \autoref{mainconj} holds on the caterpillar line. In \autoref{sec:spec-cat}, we study the structure of the spectral curve. In particular, we show that it degenerates on the caterpillar line, and we use this degeneration pattern to construct cycles $C^{(k)}_i$. In \autoref{sec:spectral-networks}, we use the spectral network theory to show that its predictions imply \autoref{mainconj} for the case of $n=2,3$. We also outline the relation to cluster coordinates of Goncharov-Shen, and to supersymmetric quantum gauge theory. In \autoref{subsect-dCP}, we collect information on the De Concini-Procesi compactification needed in the definition of the caterpillar line. In \autoref{app:Stokes_3}, we give a detailed calculation of WKB asymptotics of the Stokes matrices on the caterpillar line for $n=3$. In \autoref{app:sn-review}, we recall some basic facts and conjectures about spectral networks.

\subsection*{Acknowledgements}
\noindent The authors would like to thank Alexander Goncharov, Nikita Nikolaev, Alexander Shapiro, and Valerio Toledano Laredo for interesting and productive discussions. 
The research of A.A. was supported in part by grants 208235 and 200400 and by the National Center for Competence in Research (NCCR) SwissMAP of the Swiss National Science Foundation, and by the award of the Simons Foundation to the Hamilton Mathematics Institute of the Trinity College Dublin under the program “Targeted Grants to Institutes”.
The research of A.N. was supported by the National Science Foundation
grant 2005312 and by a Simons Fellowship in Mathematics.
Xu is supported by the National Key Research and Development Program of China (No. 2021YFA1002000) and by the National Natural Science Foundation of China (No. 12171006).

\section{Stokes matrices, Poisson brackets and WKB}\label{beginsection}

In \autoref{Canonicalsol}, we recall the definitions of Stokes matrices and canonical solutions in Stokes supersectors. In \autoref{sec:Poisson_section2}, we recall Poisson geometry properties of Stokes matrices. In \autoref{sec:WKB_behavior}, we conjecture that Stokes matrices admit a WKB expansion, and under this assumption, we prove rhombus inequalities on leading WKB exponents.
In \autoref{sec:caterpillar_section2}, we define regularized Stokes matrices on the caterpillar line, give their explicit formulas in terms of $\Gamma$-functions following \cite{Xu}, study their WKB asymptotics, and give an algebraic interpretation of the corresponding rhombus inequalities.

\subsection{Canonical solutions and Stokes matrices}\label{Canonicalsol}
The system \eqref{eq:intro-main-ode} has a second-order pole at $\infty$ and (if $A\neq 0$) a first-order pole at $0$, and it admits a unique formal power series solution of the form
\[
\hat{F}(z; u, \varepsilon)=({\rm Id}_n+F_1z^{-1}+F_2z^{-2}+\cdots)e^{\frac{\I uz}{\eps}}z^{\frac{[A]}{2\pi\I\eps}},
\]
where ${\rm Id}_n$ is the rank $n$ identity matrix, and $[A]$ is the diagonal part of $A$. 
The formal power series ${\rm Id}_n+F_1z^{-1}+F_2z^{-2}+\cdots$ is in general divergent. Thus, $\hat{F}(z, u, \eps)$ is only a formal solution. However, the Borel resummation method gives rise to actual solutions with prescribed asymptotics defined in \emph{Stokes supersectors}. In our case, there are two Stokes supersectors of \eqref{eq:intro-main-ode}, and they are given by  
$$
\begin{array}{lll}
\widehat{{\rm Sect}}_+ & = & \{z\in\mathbb{C}~|-\pi<{\rm arg}(z)<\pi\}, \\
\widehat{{\rm Sect}}_- & = & \{z\in\mathbb{C}~|-2\pi<{\rm arg}(z)<0\}.
\end{array}
$$
Choose the branch of ${\rm log}(z)$ which is real on the positive real axis and which has a branch cut along the nonnegative imaginary axis $\I \mathbb{R}_{\ge 0}$. Then, ${\rm log}(z)$ has imaginary part $-\pi$ on the negative real axis in $\widehat{\rm Sect}_-$. The following result is standard (see e.g., \cite{BJL, LR, Wasow, Xu}):

\begin{thm}\label{uniformresum}
For any $u\in\h_{\rm reg}(\mathbb{R})$, there are  unique fundamental solutions $F_\pm:\widehat{{\rm Sect}}_\pm\to {\rm GL}_n(\mathbb{C})$ of \eqref{eq:intro-main-ode} such that 
\begin{eqnarray*}
F_+(z;u, \varepsilon)\cdot e^{-\frac{\I uz}{\eps}}\cdot z^{\frac{[A]}{2\pi\I\eps}}&\sim & {\rm Id}_n, \ \   {\rm as} \ \  z\rightarrow 0 \ \  \text{ within } \ \  \widehat{{\rm Sect}}_+,
\\
F_-(z;u, \varepsilon)\cdot e^{-\frac{\I uz}{\eps}}\cdot z^{\frac{[A]}{2\pi\I\eps}}&\sim &{\rm Id}_n, \ \  {\rm as} \ \  z\rightarrow 0 \ \  \text{ within } \ \  \widehat{{\rm Sect}}_-.
\end{eqnarray*} 
The solutions $F_\pm$ are called  \emph{canonical solutions} in ${\rm Sect}_\pm$. 
\end{thm}

The set $\h_{\rm reg}(\mathbb{R})$ consists of connected components labeled by elements of the permutation group. In more detail,
\begin{equation} \label{eq:com-u}
    \h_{\rm reg}(\mathbb{R}) = \cup_{\sigma \in S_n} U_\sigma, \hskip 0.3cm
U_\sigma=\{ (u_1, \dots, u_n) \in \mathbb{R}^n; u_{\sigma(1)} < u_{\sigma(2)} < \dots < u_{\sigma(n)}\}.
\end{equation}

We denote by $P_\sigma \in {\rm GL}(n)$ the $n \times n$ matrix implementing the permutation $\sigma \in S_n$.

\begin{defi}\label{defiStokes}
For any $u\in U_\sigma$, the {\it Stokes matrices} of the system \eqref{eq:intro-main-ode} (with respect
to $\widehat{{\rm Sect}}_+$ and the branch of ${\rm log}(z)$) are the elements $S_\pm(u,A,\eps)\in {\rm GL}(n)$ determined by
\begin{eqnarray}
F_+(z; u, \eps)&=&F_-(z; u, \eps)\cdot e^{-\frac{[A]}{2\eps}}P_\sigma S_+(u,A,\eps)P_\sigma^{-1},\\
F_-(ze^{-2\pi \I}; u, \eps)&=&F_+(z; u, \eps)\cdot P_\sigma S_-(u,A,\eps)P_\sigma^{-1}e^{\frac{[A]}{2\eps}},
\end{eqnarray}
where the first (resp. second) identity is understood to hold in ${\rm Sect}_-$
(resp. ${\rm Sect}_+$) after $ F_+$ (resp. $F_-$)
has been analytically continued clockwise. 
\end{defi} 
The prescribed asymptotics of $F_\pm(z;u, \eps)$ at $z=\infty$ and the identities in \autoref{defiStokes} ensure that the Stokes matrices $S_+(u,A,\eps)$ and $S_-(u,A,\eps)$ are upper and lower triangular, respectively (see e.g. \cite[Chapter 9.1]{Balser} or \cite[Lemma 17]{Boalch}). 
The reality conditions  \eqref{eq:intro-reality} imply that if 
$F(z)$ is a solution of
\eqref{eq:intro-main-ode}, then so is the inverse of the conjugate transpose 
$(F(\bar{z})^{-1})^\dagger$ (see \cite{Boalch}). Thus, we have the following:
\begin{lem} \label{lem:stokes-reality}
$S_-(u,A,\eps)=S_+(u,A,\eps)^\dagger$.
\end{lem}
\autoref{lem:stokes-reality} implies that the product
\begin{equation}      \label{eq:S=SS}
S(u, A, \eps)=S_-(u,A,\eps)\cdot S_+(u,A,\eps)
\end{equation}
is a positive definite Hermitian matrix. Also, note that $S$ uniquely determines $S_+$ and $S_-$ (under the condition that their diagonal entries are positive reals).

\subsection{Poisson structure of Stokes matrices} \label{sec:Poisson_section2}

Poisson geometry turns out to be a very useful tool in the study of the Stokes phenomenon. In more detail, we equip the set of data $(u, A) \in \mathfrak{h}_{\rm reg}(\mathbb{R}) \times {\rm Herm}(n)$ with the product Poisson structure $\pi$, where $\pi$ vanishes on $\mathfrak{h}_{\rm reg}(\mathbb{R})$, and coincides with the linear Kirillov-Kostant-Souriau (KKS) Poisson structure on ${\rm Herm}(n)$. In particular, we have
\begin{equation}     \label{eq:KKS}
\{ A_{ij}, A_{kl} \} = \delta_{jk} A_{il} - \delta_{il} A_{kj}.
\end{equation}
Note that equation \eqref{eq:KKS} implies
$$
\{ \eps^{-1} A_{ij}, \eps^{-1} A_{kl} \} = \eps^{-1}\left(\delta_{jk} \eps^{-1}A_{il} - \delta_{il} \eps^{-1}A_{kj}\right).
$$
In other words, $(\eps^{-1}u, \eps^{-1}A) \in \mathfrak{h}_{\rm reg}(\mathbb{R}) \times {\rm Herm}(n)$ carry the Poisson bracket $\eps^{-1} \pi$.

The results of Guillemin-Sternberg \cite{GS} show that eigenvalues $\lambda^{(k)}_i(A)=\lambda_i(A^{(k)})$ of submatrices $A^{(k)}$ Poisson commute with each other:
$$
\left\{ \lambda^{(k)}_i, \lambda^{(l)}_j\right\} =0.
$$
Furthermore, they satisfy Cauchy's interlacing inequalities
\begin{equation}      \label{eq:Cauchy_interlacing}
\lambda^{(k)}_i \leq \lambda^{(k+1)}_i \leq \lambda^{(k)}_{i+1},
\end{equation}
and they serve as action variables of the {\em Gelfand-Tsetlin integrable system}. Their Hamiltonian flows generate an action of the {\em Thimm torus} on the set ${\rm Herm}_0(n)$ where all the inequalities \eqref{eq:Cauchy_interlacing} are strict. This action preserves the action variables $\lambda^{(k)}_i$.

It is natural to ask whether the Stokes matrices $S_\pm(u, A, \eps)$ carry a compatible Poisson structure which would make the map $(u, A) \mapsto S_\pm(u, A, \eps)$ into a Poisson map. To explain the answer to this question, recall that the set of upper triangular matrices with positive reals on the diagonal is a Poisson Lie group $U(n)^*$ (over $\mathbb{R}$) dual under Poisson-Lie duality to the compact Poisson-Lie group $U(n)$ with the standard Poisson structure (see \cite{LW,STS}). That is, $U(n)^*$ carries a canonical Poisson structure $\pi^*$ such that the multiplication map is a Poisson map. In terms of minor coordinates $\Delta^{(k)}_i$, the Poisson structure $\pi^*$ is Laurent polynomial (see {\em e.g.} \cite{AD}). 

The Poisson-Lie group $U(n)^*$ carries the Flaschka-Ratiu completely integrable system (see \cite{FR}). Let $S_+ \in U(n)^*$ and denote $S_- = S_+^\dagger$. Then, the matrix $S=S_- S_+ \in {\rm Herm}(n)$ is positive definite, and so are all $S^{(k)} = S_-^{(k)} S_+^{(k)} \in {\rm Herm}(k)$. The action variables of the Flaschka-Ratiu system are
$\log \, \lambda^{(k)}_i(S)$.
They satisfy the interlacing inequalities, and they generate the Thimm torus action which preserves $\log \, \lambda^{(k)}_i(S)$.

\begin{ex}
    For $n=2$, we have  minor coordinates $\Delta^{(1)}_1, \Delta^{(2)}_2 \in \mathbb{R}_{+}, \Delta^{(2)}_1, \overline{\Delta}^{(2)}_1 \in \mathbb{C}$. For simplicity, we put $\Delta^{(2)}_2=1$
    (this corresponds to considering $SU(2)^*$ instead of $U(2)^*$). Then, the Poisson bracket $\pi^*$ has the form
    $$
\left\{ \Delta^{(1)}_1, \Delta^{(2)}_1 \right\}^* = -\frac{\I}{2} \Delta^{(1)}_1 \Delta^{(2)}_1, \hskip 0.3cm
\left\{ \Delta^{(2)}_1, \overline{\Delta}^{(2)}_1\right\}^* = \I \left(\left(\Delta^{(1)}_1\right)^{-2} - \left(\Delta^{(1)}_1\right)^{2}\right).
    $$
\end{ex}

The following statement is one of the main results in \cite{Boalch}:
\begin{thm}     \label{thm:Boalch_Poisson}
    For each $u \in \mathfrak{h}_{\rm reg}(\mathbb{R})$ and $\eps \in \mathbb{R}_+$ the map
    $$
    ({\rm Herm}(n), \pi) \to (U(n)^*, \eps^{-1} \pi^*), \hskip 0.3cm A \mapsto S_+(u, A, \eps)
    $$
    is a Poisson map.
\end{thm}

\begin{proof}
    In more detail, Theorem 1 in \cite{Boalch} is the statement above for $\eps=1$. Since pairs
    $(\eps^{-1} u, \eps^{-1} A)$ carry the Poisson bracket $\eps^{-1} \pi$, the Poisson bracket $\pi^*$ also gets rescaled by $\eps^{-1}$.
\end{proof}

\subsection{WKB behavior}     \label{sec:WKB_behavior}

It is a common belief that in the limit of $\eps \to 0$ the Stokes matrices admit a WKB behavior. In more detail, we will assume the following:
\begin{conj} \label{conj:WKB_section2}
For $A$ generic, the minor coordinates $\Delta^{(k)}_i(S_+(u, A, \eps))$ admit an asymptotic expansion
\begin{equation}       \label{eq:WKB_section2}
\log \ \Delta^{(k)}_i(S_+(u, A, \eps)) \sim \eps^{-1} \delta^{(k)}_i(u, A) + \left(\delta^{(k)}_i \right)_0
+  \eps \left(\delta^{(k)}_i \right)_1 + \dots 
\end{equation}
The corresponding expansion of the eigenvalues 
$\lambda^{(k)}_i(S(u,A,\eps))$,
\begin{equation}      \label{eq:WKBmu}
 \log \, \lambda^{(k)}_i(S(u, A, \eps)) \sim \eps^{-1} \, \eta^{(k)}_i(u, A) + \dots \, ,
\end{equation}
is uniform as $u$ approaches the caterpillar line, $u \to u_{\rm cat}(t)$.
\end{conj}

We refer to the coefficients $\delta^{(k)}_i(u, A)$ in \eqref{eq:WKB_section2} as \emph{leading WKB exponents}. It is convenient to introduce a special notation for their real parts:
$$
l^{(k)}_i(u, A) = {\rm Re} \ \delta^{(k)}_i(u, A).
$$
In contrast to $\Delta^{(k)}_i(S_+)$, the eigenvalues $\lambda^{(k)}_i(S)$ are real, and so are the leading exponents $\eta^{(k)}_i(u, A)$ in \eqref{eq:WKBmu}.

\autoref{thm:Boalch_Poisson} and \autoref{conj:WKB_section2} imply the following interesting result:
\begin{thm}       \label{thm:rhombus_section2}
Assuming \autoref{conj:WKB_section2}, the real parts of leading WKB exponents $l^{(k)}_i(u, A)$ 
verify the rhombus inequalities
\begin{equation}       \label{eq:rhombus_section2}
l^{(k+1)}_i + l^{(k)}_{i-1} \geq l^{(k+1)}_{i+1} + l^{(k)}_{i}, \hskip 0.3cm
l^{(k+1)}_{i} + l^{(k)}_{i} \geq l^{(k+1)}_{i+1} + l^{(k)}_{i-1},
\end{equation}
where $l^{(k)}_0 \equiv 0$. Furthermore, all Poisson brackets between $\delta^{(k)}_i, \overline{\delta}^{(k)}_i$ vanish:
\begin{equation}      \label{eq:deltadelta_vanish}
    \left\{ \delta^{(k)}_i, \delta^{(l)}_j\right\} =0, \hskip 0.3cm
    \left\{ \delta^{(k)}_i, \overline{\delta}^{(l)}_j\right\} =0, \hskip 0.3cm
    \left\{ \overline{\delta}^{(k)}_i, \overline{\delta}^{(l)}_j\right\} =0.
\end{equation}
\end{thm}

\begin{proof}
    We first consider the case of $n=2$. We compute,
    \begin{equation}      \label{eq:DeltaDelta1}
 \ \left\{ \Delta^{(2)}_1, \overline{\Delta}^{(2)}_1\right\}^* = \eps \
\left\{ e^{\eps^{-1} \delta^{(2)}_1 + \dots}, e^{\eps^{-1} \overline{\delta}^{(2)}_1 + \dots}\right\} =
e^{2 \eps^{-1} l^{(2)}_1}\left(\eps^{-1} \{ \delta^{(2)}_1, \overline{\delta}^{(2)}_1\} + \dots\right),
    \end{equation}
where in the first equality we used \autoref{thm:Boalch_Poisson} to pass from $\pi^*$ to $\pi$.
Next, we compute 
\begin{equation} \label{eq:DeltaDelta2}
 \left\{ \Delta^{(2)}_1, \overline{\Delta}^{(2)}_1\right\}^* = \I \ \left( \left(\Delta^{(1)}_1\right)^{-2} - \left(\Delta^{(1)}_1\right)^{2}\right)=
\I \ \left( e^{ - 2\eps^{-1} l^{(1)}_1 + \dots} -  e^{2\eps^{-1} l^{(1)}_1 + \dots}\right),
\end{equation}
where we have used the fact that $\Delta^{(1)}_1 \in \mathbb{R}_+$ and $\delta^{(1)}_1 = l^{(1)}_1$.
We observe that the right hand sides of \eqref{eq:DeltaDelta1} and \eqref{eq:DeltaDelta2} can only match if
$$
l^{(2)}_1 \geq l^{(1)}_1, \hskip 0.3cm l^{(2)}_1 \geq -l^{(1)}_1,
$$
and that these are exactly the rhombus inequalities in the case of $n=2$ (under the simplifying assumption of $l^{(2)}_2=0$).  
The analysis for arbitrary $n$ is done in Theorem 5 in \cite{AD}. Note that in that paper one assumes
$$
\log \ \Delta^{(k)}_i = \eps^{-1} l^{(k)}_i + \I \varphi^{(k)}_i,
$$
and the resulting Poisson bracket on $(l^{(k)}_i, \varphi^{(k)}_i)$ depends on $\eps$. However, the $n=2$ argument above shows that only the real part of the leading WKB exponent is relevant for the proof of rhombus inequalities. Hence, the proof of \cite{AD} applies verbatim.

To establish the vanishing of Poisson brackets of the leading WKB exponents, we re-examine the equality of the right-hand sides of \eqref{eq:DeltaDelta1} and \eqref{eq:DeltaDelta2}. We observe that even if the exponential factors match, one has an extra $\eps^{-1}$ factor in front of $\{ \delta^{(2)}_1, \overline{\delta}^{(2)}_1\}$ in \eqref{eq:DeltaDelta1} which is absent in \eqref{eq:DeltaDelta2}. Hence,
$$
\{ \delta^{(2)}_1, \overline{\delta}^{(2)}_1\} =0.
$$
The same argument applies to the Poisson bracket of any two minor coordinates which in turn implies
\eqref{eq:deltadelta_vanish}.

\end{proof}

For $S_\pm$ generic, there is a simple linear relation between the leading WKB exponents $l^{(k)}_i$ and $\eta^{(k)}_i$:
\begin{pro}     \label{pro:l=sumnu}
    Assume that $S_\pm$ admits a WKB expansion and that all the rhombus inequalities for $l^{(k)}_i$ are strict.  Then, we have 
   \begin{equation}
    l^{(k)}_i = \frac{1}{2} \sum_{j=k-i+1}^k \eta^{(k)}_j.
    \end{equation}
\end{pro}

\begin{proof}
    This fact follows from Proposition 2 in \cite{APS} (see also Proposition 5.1 in \cite{ALL}). For the convenience of the reader, we illustrate the proof in the case of $n=2$. 
Consider a matrix $S_+$ which admits a WKB behavior:
    $$
S_+ = \left(
\begin{array}{cc}
e^{\frac{1}{\eps} \alpha + \dots} & e^{\frac{1}{\eps} \beta + \dots} \\
0 & e^{- \frac{1}{\eps} \alpha + \dots}
\end{array}
\right).
    $$
Then, $\Delta^{(2)}_1(S_+) = (S_+)_{12} = e^{\frac{1}{\eps} \beta + \dots}$ and $\delta^{(2)}_1 = \beta, l^{(2)}_1 = {\rm Re} \, \beta$. We also have
$$
\lambda^{(2)}_1(S) + \lambda^{(2)}_2(S) = {\rm Tr} \, S = {\rm Tr}(S_+^\dagger S_+)= e^{ \frac{2}{\eps} {\rm Re} \, \beta + \dots} + e^{ \frac{2}{\eps} \alpha + \dots} + e^{- \frac{2}{\eps} \alpha + \dots}.
$$
Hence,
$$
\eta^{(2)}_2 = 2 \, {\rm max}({\rm Re} \, \beta, \alpha, - \alpha) = 2 \, {\rm Re} \, \beta =
2\,  l^{(2)}_1,
$$
as required.
Here we have used the strict rhombus inequalities  ${\rm Re} \, \beta > \alpha, - \alpha$. 
\end{proof}

\begin{rmk}  \label{rmk:rhombus=interlacing}
    Under the linear change of variables $l^{(k)}_i = \frac{1}{2} \sum_{j=k-i+1}^k \eta^{(k)}_i$, the rhombus inequalities for parameters $l^{(k)}_i$ are equivalent to the interlacing inequalities for parameters $\eta^{(k)}_i$. In this paper, we will see several instances of such a change of variables.
\end{rmk}

\subsection{Stokes matrices on the caterpillar line} \label{sec:caterpillar_section2}

Determining Stokes matrices for $n\geq 3$ is a formidable task. However, there is a limit in which explicit formulas are available. In more detail, one considers the situation for $u \in U_{\rm id}$ (here ${\rm id} \in S_n$ is the identity element) when 
\begin{equation}  \label{eq:caterpillar_line}
u_2-u_1=t, \hskip 0.3cm \frac{u_{j+1}-u_j}{u_{j}-u_{j-1}} \to +\infty \ \ {\rm for} \ \ j=2, \dots, n-1.
\end{equation}
Such configurations no longer belong to $\h_{\rm reg}(\mathbb{R})$, but they make sense in its \emph{De Concini - Procesi} compactification (see \cite{dCP} and \autoref{subsect-dCP} for details). We say that configurations \eqref{eq:caterpillar_line} belong to the \emph{caterpillar line} and denote them by $u_{\rm cat}(t)$. The De Concini-Procesi compactification of $\h_{\rm reg}(\mathbb{R})$ can also be identified with the tautological line bundle $\widetilde{\mathcal{M}}_{0, n+1}(\mathbb{R}) \to \overline{\mathcal{M}}_{0, n+1}(\mathbb{R})$ over (the real locus of) the Deigne-Mumford compactification of the moduli space of rational curves with $n+1$ marked points.

The results described in this Section were established in \cite{Xu} by the method of isomonodromy deformation. We will need the following notation: for a matrix $A \in {\rm Mat}_n(\mathbb{C})$, we denote 
 \begin{eqnarray}\label{delta}
 \delta_k(A)_{ij}=\left\{
          \begin{array}{lr}
             A_{ij},   & {\rm if} \ \ 1\le i, j\le k \ {\rm or} \ i=j;  \\
           0, & {\rm otherwise}.
             \end{array}
\right.\end{eqnarray}
Note that if $A \in {\rm Herm}(n)$, then so is $\delta_k(A)$.
For a matrix $B\in {\rm Mat}_n(\mathbb{C})$, we denote by $\Delta^I_J(B)$ its minor formed by the rows $I=\{i_1, \dots, i_k\}$ and by the columns $J=\{ j_1, \dots, j_k\}$.

While the Stokes matrices $S_\pm(u, A, \eps)$ don't have a limit when $u \to u_{\rm cat}(t)$, 
one can separate their divergent and convergent parts in the following way. Define the unitary matrix
$$
V(u, A, \varepsilon) = 
 \overrightarrow{\prod_{k=2,\dots,n-1} }\left(\frac{u_{k}-u_{k-1}}{u_{k+1}-u_{k}}\right)^{\frac{{\rm log}\left( \delta_k(S_-)\delta_k(S_+) \right)}{2\pi\I\eps}},
$$
and denote
\begin{equation}        \label{eq:Thimm}
S^{\rm reg}(u, A,\eps) =
V(u, A, \eps)\cdot S(u,A,\eps)\cdot V(u, A, \eps)^{-1},
\end{equation}
where $S$ is given by equation \eqref{eq:S=SS}.
This expression uniquely defines the lower and upper triangular matrices $S_-^{\rm reg} =(S_+^{\rm reg})^\dagger$ with the property $S^{\rm reg}=S_-^{\rm reg}S_+^{\rm reg}$.
\begin{pro}    \label{pro:SSreg}
    The map $S_\pm \mapsto S^{\rm reg}_\pm$ is a Poisson map under the bracket $\eps^{-1} \pi^*$. It preserves the eigenvalues $\lambda^{(k)}_i(S^{\rm reg}) = \lambda^{(k)}_i(S)$.
\end{pro}

\begin{proof}
    The transformation \eqref{eq:Thimm} is a particular instance of the Thimm torus action, and hence it preserves all Gelfand-Tsetlin functions $\lambda^{(k)}_i(S)$. Furthermore, the parameters of that Thimm action depend only on action variables, and hence it is a canonical transformation: it preserves action variables, and it shifts the conjugate angle variables by a function of action variables.
\end{proof}

The following result will be of importance to us:

\begin{thm}\cite[Theorem 1.5]{Xu} \label{reglimitcat}
For fixed $\eps>0$ and for any $A\in\Herm(n)$, the expressions $S_\pm^{\rm reg}(u, A,\eps)$ have a well defined limit for $u \to u_{\rm cat}(t)$. Furthermore, at $u=u_{\rm cat}(t)$ one has
\begin{align*}\nonumber
     (S_+^{\rm reg})_{k,k+1}=& 2\pi\I \cdot \left(\frac{u_2-u_1}{\eps}\right)^{\frac{A_{k+1,k+1}-A_{kk}}{2\pi \I\eps}} 
 e^{\frac{{A_{kk}+A_{k+1,k+1}}}{4\eps}}  \cdot \\ \label{S+} &\sum_{i=1}^k\frac{\prod_{l=1,l\ne i}^{k}\Gamma\left(1+\frac{\lambda^{(k)}_l-\lambda^{(k)}_i}{2\pi \I\eps}\right)}{\prod_{l=1}^{k+1}\Gamma\left(1+\frac{\lambda^{(k+1)}_l-\lambda^{(k)}_i}{2\pi \I\eps}\right)}\frac{\prod_{l=1,l\ne i}^{k}\Gamma\left(\frac{\lambda^{(k)}_l-\lambda^{(k)}_i}{2\pi \I\eps}\right)}{\prod_{l=1}^{k-1}\Gamma\left(1+\frac{\lambda^{(k-1)}_l-\lambda^{(k)}_i}{2\pi \I\eps}\right)}\cdot \Delta^{1,\dots,k-1,k}_{1,\dots,k-1,k+1}\left(\frac{A-\lambda^{(k)}_i}{2\pi\I\eps}\right).
\end{align*}
Other entries of $S_\pm^{\rm reg}(u_{\rm cat}(t),A,\eps)$ are also given by explicit formulas.
\end{thm}

\begin{rmk}
We can further consider the limit $t=u_2-u_1\rightarrow 0$. In more detail, one can put
\begin{equation}\label{ucatpoint}
S_\pm^{\rm reg}(u_{\rm cat},A,\eps)= t^{\frac{\delta_1(A)}{2\pi\I \eps}} \cdot S_\pm^{\rm reg}(u_{\rm cat}(t),A,\eps) \cdot t^{\frac{-\delta_1(A)}{2\pi\I \eps}},
\end{equation}
where the right-hand side turns out to be independent of $t$.
\end{rmk}

\begin{ex} \label{ex:n2-stokes-matrices}
Consider the case of $n=2$:
\begin{eqnarray*}\label{2by2}
\eps\frac{dF}{dz}=\left(\left(\begin{array}{cc}
    \I u_1 & 0  \\
    0 & \I u_2
  \end{array}\right)
-\frac{1}{2\pi \I z}{\left(
  \begin{array}{cc}
    t_1 & a  \\
    \bar{a} & t_2
  \end{array}
\right)}\right)\cdot F.
\quad
\end{eqnarray*} 
Following Proposition 8 in \cite{BJL}, the Stokes matrices (with respect to the chosen branch of ${\rm log}(z)$) are
\[S_-(u,A,\eps)=\left(
  \begin{array}{cc}
    e^{\frac{t_1}{2\eps}} &  0  \\
    \frac{ \frac{\bar{a}}{\eps}\cdot  e^{\frac{t_2+t_1}{4\eps}} \left(\frac{u_2-u_1}{\eps} \right)^{\frac{t_1-t_2}{2\pi \I\eps }}}{\Gamma \left(1-\frac{\lambda_1-t_1}{2\pi \I\eps} \right)\Gamma \left(1-\frac{\lambda_2-t_1}{2\pi \I\eps} \right)}  & e^{\frac{t_2}{2\eps}}
  \end{array}\right), \ \ S_+(u,A,\eps)=\left(
  \begin{array}{cc}
    e^{\frac{t_1}{2\eps}} &  \frac{ \frac{a}{\eps}\cdot e^{\frac{t_2+t_1}{4\eps}} \left(\frac{u_2-u_1}{\eps}\right)^{\frac{t_2-t_1}{2\pi \I\eps }}}{\Gamma\left(1+\frac{\lambda_1-t_1}{2\pi \I\eps}\right)\Gamma\left(1+\frac{\lambda_2-t_1}{2\pi \I\eps}\right)}   \\
    0 & e^{\frac{t_2}{2\eps}}
  \end{array}\right).\]
By definition, for $t=u_2-u_1>0$ we have $S_\pm (u_{\rm cat}(t),A,\eps)=S_\pm (u,A,\eps)$.
Furthermore, we have  
$$
t^{\frac{\delta_1(A)}{2\pi \I}}= {\rm diag} \left( (u_2-u_1)^{\frac{t_1}{2\pi\I\eps}},(u_2-u_1)^{\frac{t_2}{2\pi\I\eps}} \right),
$$ 
and we obtain
 \[S_-^{\rm reg}(u_{\rm cat},A,\eps)^\dagger=S_+^{\rm reg}(u_{\rm cat},A,\eps)=\left(
  \begin{array}{cc}
    e^{\frac{t_1}{2\eps}} &  \frac{a}{\eps} \cdot \frac{ \frac{1}{\eps}^{\frac{t_2-t_1}{2\pi \I\eps }} e^{\frac{t_1+t_2}{4\eps}}}{\Gamma\left(1+\frac{\lambda_1-t_1}{2\pi \I\eps}\right)\Gamma\left(1+\frac{\lambda_2-t_1}{2\pi \I\eps}\right)}   \\
    0 & e^{\frac{t_2}{2\eps}}
  \end{array}\right).\] 
\end{ex}

One can use \autoref{reglimitcat} to obtain information on the WKB expansion of regularized Stokes matrices on the caterpillar line:

\begin{thm}\label{thm:conj_at_ucat}
For all $A\in\Herm_0(n)$, the regularized Stokes matrices $S_\pm^{\rm reg}(u_{\rm cat}(t), A, \eps)$ verify \autoref{conj:WKB_section2}, and one has
\begin{equation}        \label{eq:rhombus_interlacing_section2}
    l^{(k)}_i(u_{\rm cat}(t), A) = \frac{1}{2} \sum_{j=k-i+1}^k \lambda^{(k)}_j(A),
\end{equation}
where $\lambda^{(k)}_j(A)$ are  eigenvalues of the Hermitian matrices $A^{(k)}$ ordered from bottom to top.
\end{thm}
\begin{proof}
Since elements of $S_\pm^{\rm reg}(u_{\rm cat}(t), A, \eps)$ are given by explicit formulas, one can check the statements of the theorem by a direct (albeit tedious) computation. We illustrate this strategy for $n=2$.
Recall that $r \in \mathbb{R}$ we have the following asymptotic expansion of the $\Gamma$-function with repect to the parameter $\eps \to 0$:
\begin{equation}\label{gammaasy}
\log \, \Gamma\left(1+\frac{r}{2\pi \I \eps }\right) = 
- \frac{r}{2\pi \I\eps }\, \log(\eps) + \frac{r}{2\pi \I \eps}\log\left(\frac{|r|}{2\pi}\right)- \frac{|r|}{4\eps}-\frac{r}{2\pi \I\eps }+\frac{1}{2}\log\left(\frac{r}{\I \eps}\right) + O(\eps).
\end{equation}
Following \autoref{ex:n2-stokes-matrices} and equation \eqref{gammaasy}, consider the $n=2 $ case with 
$$
A=\left(
  \begin{array}{cc}
    t_1 & a  \\
    \bar{a} & t_2
  \end{array}
\right)
$$ 
and eigenvalues $\lambda_1 < \lambda_2$. Then, the asymptotics of the entry $S_\pm^{\rm reg}(u_{\rm cat}(t), A, \eps)_{12}=S_+(u_{\rm cat}(t),A,\eps)_{12}$ is given by
\begin{align*}
& (S_+)_{12}=\frac{ \frac{a}{\eps}\cdot e^{\frac{t_2+t_1}{4\eps}} \left(\frac{u_2-u_1}{\eps}\right)^{\frac{t_2-t_1}{2\pi \I\eps }}}{\Gamma\left(1+\frac{\lambda_1-t_1}{2\pi \I\eps}\right)\Gamma\left(1+\frac{\lambda_2-t_1}{2\pi \I\eps}\right)}
\sim \ e^{\eps^{-1} \left( \frac{\lambda_2}{2} + \I \varphi \right)}\left(\frac{a}{\sqrt{(\lambda_1 - t_1)(t_1 - \lambda_2)}} +O(\eps)\right).
\end{align*}
This matches \eqref{eq:rhombus_interlacing_section2} since in this case $\Delta^{(2)}_1 = (S_+)_{12}$ and 
$l^{(2)}_1=\lambda_2/2$, as required. 
Observe that the coefficients in front of the terms $\eps^{-1} \log(\eps)$ an $\log(\eps)$ vanish confirming the WKB behavior: the actual leading term is proportional to $\eps^{-1}$. Also, observe that the condition of $A \in {\rm Herm}_0(n)$ is necessary: in our example, $a=0$ implies $(S_+)_{12}=0$ and the WKB expansion does not apply.
For completeness, we give an explicit formula for the imaginary part $\varphi$ of the leading WKB exponent:
$$
\varphi=\frac{(t_1-t_2)}{2\pi} \log(u_2-u_1) +\frac{t_1-t_2}{2\pi}+\frac{(\lambda_1-t_1)}{2\pi}\log\left(\frac{t_1 - \lambda_{1}}{2\pi}\right) +\frac{(\lambda_2 - t_1)}{2\pi} \log\left(\frac{\lambda_2 - t_1}{2\pi}\right).
$$
For $n= 3$, we collect detailed calculations in \autoref{app:Stokes_3}. For $n>3$, the calculations are similar but more tedious.

We prove equation \eqref{eq:rhombus_interlacing_section2} using a combination of the following two results. First, Proposition 4.1 in \cite{Xu} shows that the map $A \mapsto S_+^{\rm reg}(u_{\rm cat}, A, \eps)$ intertwines (up to the factor of $\eps^{-1}$) the Gelfand-Tsetlin and the Flaschka-Ratiu integrable systems. That is, for any $t=u_2-u_1>0$
\begin{equation}\label{eq:mu=lambda}
\log \, \lambda^{(k)}_i(u_{\rm cat}(t), A, \eps) = 
\eps^{-1} \lambda^{(k)}_i(A).
\end{equation}
Hence,
$$
\eta^{(k)}_i(u_{\rm cat}(t), A) = \lambda^{(k)}_i(A).
$$
And by \autoref{pro:l=sumnu} we have
$$
l^{(k)}_i(u_{\rm cat}(t), A) = \frac{1}{2} \sum_{j=k-i+1}^k \nu^{(k)}_j(u_{\rm cat}(t), A) = \frac{1}{2} \sum_{j=k-1+i}^k \lambda^{(k)}_j(A),
$$
as required.
\end{proof}

We conclude this section with the following observation:
\begin{pro}     \label{pro:limit_l}
Under \autoref{conj:WKB_section2}, 
    the leading WKB exponents of $S_+(u, A, \eps)$ admit a limit for $A \in {\rm Herm}_0(n)$ and  $u \to u_{\rm cat}(t)$, and 
    $$
l^{(k)}_i(u, A) \to_{u \to u_{\rm cat}(t)} l^{(k)}_i(u_{\rm cat}(t), A) = \frac{1}{2} \sum_{j=k-i+1}^k \lambda^{(k)}_j(A).
    $$
\end{pro}

\begin{proof}
By \autoref{pro:SSreg}, $\lambda^{(k)}_i(S^{\rm reg}) = \lambda^{(k)}_i(S)$. Hence,
they have the same leading WKB exponents $\eta^{(k)}_i(u, A)$.
Under \autoref{conj:WKB_section2}, the WKB expansion of $\lambda^{(k)}_i(S)$ is uniform for $u \to u_{\rm cat}(t)$. Therefore,
$$
{\rm lim}_{u \to u_{\rm cat}(t)} \eta^{(k)}_i(u, A) = 
\eta^{(k)}_i(u_{\rm cat}(t), A).
$$
By \autoref{pro:l=sumnu}, the WKB exponent $l^{(k)}_i$'s are related to $\eta^{(k)}_i$ by a linear transformation if the rhombus inequalities are strict. By assumption, $A \in {\rm Herm}_0(n)$ and the rhombus inequalities are strict on the caterpillar line. By continuity, they are also strict on a small open neighborhood of the caterpillar line which implies the desired result.
    \end{proof}

One of the goals of this paper is to understand the rhombus inequalities away from the caterpillar line.
To address this question, in the next sections, we use more powerful geometric tools of spectral curves and spectral networks.

\section{Spectral curves}\label{sec:spec-cat} 
In this section, we study the degeneration of the spectral curves as $u$ approaches a point on the caterpillar line. The main results are \autoref{thm:smo-nei}, \autoref{cor:vani}, and \autoref{thm:limit_periods}. In \autoref{ssec:ggenus}, we consider generic values of parameters $(u,A)$ and compute the genus of smooth spectral curves. In \autoref{ssec:degeneration}, we consider spectral curves for the case when $u=E_n$ (the elementary matrix with $(E_n)_{nn}=1$). In \autoref{ssec:degeneration_cat}, we consider spectral curves as $u$ approaches a point on the caterpillar line and define vanishing cycles on these curves. In \autoref{ssec:distinguished}, we introduce distinguished cycles as linear combinations of vanishing cycles, and in \autoref{thm:limit_periods} we show that \autoref{mainconj} holds on the caterpillar line.  Finally, in \autoref{ssec:near-diagonal} we discuss the situation of $A$ being close to a diagonal matrix which will be our main playground in the next section.

Even though in this paper we impose a reality condition on the $u$- and $A$-parameters, most proofs in this section are in the \emph{complex setting} for more generality. The reality condition induces a $\mathbb{Z}_2$-symmetry as we discuss in \autoref{pro:realness}.

\subsection{Spectral curves: generic case} \label{ssec:ggenus}

Let $\mathfrak{g} = \mathfrak{gl}_{n}(\mathbb{C})$, and $\mathfrak{h} \subset \mathfrak{g}$ the Cartan Lie subalgebra of diagonal matrices. We start by considering 
$u\in\mathfrak{h}_{\rm reg}(\mathbb{C})$ and $A\in\mathfrak{g}^{*}=\mathfrak{gl}_{n}(\mathbb{C})^{*}\simeq \mathfrak{gl}_{n}(\mathbb{C})$.

We denote by $\Gamma (u,A)$ the \emph{spectral curve} of the equation \eqref{eq:intro-main-ode} living in $T^{*}_{\mathbb{P}^{1}}$ and defined by
\begin{equation}       
\det\left[\left(\I u-\frac{A}{2\pi \I z}\right) dz - \omega I_{n} \right]=0 \, .\label{eq:sp-d}
\end{equation}
Here $\omega = \mu(z)  dz$, and $\mu(z)$ is a meromorphic function on $\mathbb{P}^{1}$.
The 1-form $\omega$ is the pull-back to the spectral curve $\Gamma (u,A)$ of the tautological 1-form of $T^{*}_{\mathbb{P}^{1}}$. For this reason, one refers to $\omega$ as to the \emph{canonical  1-form}. Equation \eqref{eq:sp-d} shows that 
$\omega$ has poles at points lying  over $z=0$ and $z=\infty$. We call such points \emph{punctures}. In conclusion, $\Gamma(u,A)$ is an open curve with punctures lying over $z=0$ and $z=\infty$. 

One can reformulate equation \eqref{eq:sp-d} as a polynomial equation for a meromorphic function $\mu(z)$:
\begin{equation} 
P(\mu,z^{-1})=\det\left(\I u - \frac{1}{2\pi \I}\frac{A}{z}-\mu I_{n}\right)=0\label{eq:sp-c}
\end{equation}
This equation provides an embedding of $\Gamma(u,A)$ into $\mathbb{C}^2$.
For $i= 1, \dots, n$, we denote by $u_i$ the $i^{th}$ entry on the diagonal
of $u$, and by $t_i = A_{ii}$ the $i^{th}$ diagonal entry of $A$. 

This subsection aims to show that for $u$ and $A$ generic the spectral curve 
$\Gamma(u,A)$ is smooth, and to compute its genus.
First, we prove the following two lemmas. 
\begin{lem}
\label{lem:The-discriminant-}
Let $g(y)$ be a polynomial of degree $n$ with coefficients $b_0, \dots, b_{n-1}$:
$$
g(y)=y^{n}+b_{n-1}y^{n-1}+\cdots+b_{1}y+b_{0}.
$$
Then, its discriminant  can be represented in the form
\[
\Delta=b_{0}\cdot h_{1}(b_{0},b_{1},\cdots,b_{n-1})+b_{1}^{2}\cdot h_{2}(b_{1},\cdots,b_{n-1}).
\]
\end{lem}

\begin{proof}
Recall that the discriminant $\Delta(b_0, \dots, b_{n-1})$ can be written as 
\[
\Delta(b_{0},\cdots,b_{n})=(-1)^{\frac{n(n-1)}{2}}\prod_{i\neq j}(r_{i}-r_{j}),
\]
where $r_{i}$'s are the roots of the polynomial $g(y)$. Set $b_{0}=0$ and let $r_{n}=0$.
Then, we have $b_1=(-1)^{n-1}\prod_{i=1}^{n-1} r_i$
and
\begin{align}
    \Delta & = (-1)^{\frac{n(n-1)}{2}} \prod_{i=1}^{n-1}(r_{i}-r_{n})\prod_{i=1}^{n-1}(r_{n}-r_{i})\prod_{\substack{i\neq j\\
i,j\leq n-1
}
}(r_{i}-r_{j})\\
     & = (-1)^{\frac{n(n-1)}{2}}(-1)^{n-1}b_{1}^{2}\prod_{\substack{i\neq j\\
i,j\leq n-1
}
}(r_{i}-r_{j}).
\end{align}
Then, the statement follows from the fact that $\prod_{\substack{i\neq j\\
i,j\leq n-1
}
}(r_{i}-r_{j})$ is a polynomial in $b_{1},\cdots,b_{n-1}$. 
\end{proof}
\begin{lem}
\label{lem:simple-zero-smoothness}
Let $C$ be a plane curve, and suppose that on an analytic neighborhood
of $(0,0)$ it is defined by equation
\[
g(y,z)=y^{n}+f_{1}(z)y^{n-1}+f_{2}(z)y^{n-2}+\cdots+f_{n-1}(z)y+f_{n}(z) =0,
\]
where $f_{1}(z),\cdots,f_{n}(z)$ are polynomials, and
$f_{n-1}(z)$ and $f_{n}(z)$ vanish at $z=0$. Then, $C$ is smooth at
$(0,0)$ if the discriminant of $g(y,z)$ viewed as a polynomial in $y$
has a simple zero at $z=0$. 
\end{lem}

\begin{proof}
Denote by $\Delta(z)$ the discriminant of $g(y,z)$ viewed as a polynomial
in $y$. Note that vanishing of $f_{n-1}(z)$ at $z=0$ implies that
$\Delta(z)$ also vanishes at $z=0$. By \autoref{lem:The-discriminant-},
we can write $\Delta(z)$ as
\[
\Delta(z)=f_{n}(z)\cdot h_{1}(f_{1},\cdots,f_{n-1},f_{n})+f_{n-1}^{2}(z)\cdot h_{2}(f_{1},\cdots,f_{n-1}) \, .
\]
If $\Delta$ has a simple zero at $z=0$, then 
\[
\Delta'(0)=f'_{n}(0)h_{1}(f_1(0), \dots, f_n(0))\neq 0.
\]
Thus, we must have $f'_{n}(0)\neq0$, and
\[
\frac{\partial g}{\partial z}(0,0)=f'_{n}(0)\neq 0.
\]
This implies that the curve $C$ is smooth at $(0,0)$, as required. 
\end{proof}

The following proposition describes the structure of the spectral curve $\Gamma(u, A)$ for $(u, A)$ generic. 

\begin{pro} \label{pro:generic-smooth}
For $u \in \mathfrak{h}_{\mathrm{reg}} (\mathbb{C})$ and $A \in \mathfrak{gl}_{n}(\mathbb{C})$ satisfying the following genericity condition:
\begin{enumerate}
    \item[i)] eigenvalues of $A$ are  distinct,
    \item[ii)] the discriminant of the equation
$P(\mu,z^{-1})$ viewed as a polynomial in $\mu$, has only simple zeroes for $z \in \mathbb{C} \setminus \{0 \} $, 
\end{enumerate}
the spectral curve $\Gamma(u,A)$ is isomorphic to a smooth curve of genus $\frac{(n-1)(n-2)}{2}$ with $2n$ punctures. 
\end{pro}

\begin{proof}
By condition ii), it follows from \autoref{lem:simple-zero-smoothness} that the spectral curve $\Gamma(u,A)$
is smooth. 

Next, we consider the punctures of $\Gamma(u,A)$. Over $z=\infty$, there are $n$ distinct punctures $(\mu = u_i, z= \infty)$ for $i=1, \dots, n$. 
Near $z=0$, $\mu$ has $n$ branches
\begin{equation}  \label{eq:branches}
\mu_{i}(z)=-\frac{\lambda_{i}^{(n)}}{2\pi \I} \frac{1}{z}+ \cdots, \,\, {\rm for} \,\, i=1, \dots, n,
\end{equation} 
where $\lambda_i^{(n)}$ are the eigenvalues of $A$. These branches are distinct since $A$ has $n$ distinct eigenvalues. Therefore, after blowing up $\mathbb{P}^1_{\mu} \times \mathbb{P}^1_{z}$ at $(\mu = \infty, z=0)$, the closure $\overline{\Gamma}(u,A)$ of $\Gamma(u,A)$ inside the blowup will be a  smooth compactification of $\Gamma(u,A)$. Moreover, $\bar{\Gamma}(u,A) \rightarrow \mathbb {P}^1$ will be unramified at punctures over $z= \infty$ and $z=0$.

By the Riemann-Hurwitz formula, we have
\begin{equation} \label{eq:riemann-hurwitz}
2g(\overline{\Gamma})-2 =  n \cdot (2g(\mathbb{P}^1)-2) + \sum_{p \in \overline{\Gamma}}
(e_p - 1) \, ,
\end{equation}
where $e_p$ is the ramification index at $p \in \overline{\Gamma}$.

Since $\partial_{\mu} P(\mu,z^{-1})$ has only simple zeroes for $z \in \mathbb{C} \setminus \{0 \}$ as dictated by condition ii)  and $\bar{\Gamma}(u,A) \rightarrow \mathbb {P}^1$ are unramified at punctures over $z=0$ and $z= \infty$, all ramification points in $\overline{\Gamma}$ have ramification index 2, and the number of ramification points is the number of zeroes of $\partial_{\mu}P(\mu,z^{-1})$.
Since $\partial_{\mu}P(\mu,z^{-1})$ is a meromorphic function on $\overline{\Gamma}$, 
the number of zeroes and
poles (counted with multiplicity) are the same. Therefore, it suffices
to compute the number of poles, which only occur over $z=0$.  On the $i^{th}$ branch \eqref{eq:branches}, we have
\[
 \partial_{\mu}P(\mu,z^{-1})\mid_{\mu=\mu_{i}}=\prod_{j\mid j\neq i}(\mu_{i}-\mu_{j}) \\
  = \prod_{j\mid j\neq i} \left(-\frac{\lambda_{i}^{(n)} - \lambda_{j}^{(n)}}{2\pi \I}\frac{1}{z} + O(1) \right)
\]
Therefore, $\partial_{\mu}P(\mu,z^{-1})$ has a pole of order $(n-1)$ on each branch over $z=0$, and the total pole order of $\partial_{\mu}P(\mu,z^{-1})$ is $n(n-1)$.
The, equation \eqref{eq:riemann-hurwitz} yields
$$
2g(\overline{\Gamma})-2=n(-2)+n(n-1) \Rightarrow
g(\overline{\Gamma})=\frac{(n-1)(n-2)}{2},
$$
as desired.
\end{proof}

In the next proposition, we collect information on the residues of 
the canonical 1-form $\omega$ at punctures of $\Gamma(u, A)$ 
lying over $z=0$ and $z=\infty$:

\begin{pro}        \label{pro:residue_z=0}
For $u \in \mathfrak{h}_{\mathrm{reg}} (\mathbb{C})$ and $A \in \mathfrak{gl}_{n}(\mathbb{C})$ with distinct eigenvalues, the canonical 1-form $\omega$ has residues 
$-\frac{\lambda^{(n)}_i}{2\pi \I} (1\leq i \leq n)$  at punctures lying over $z=0$
and residues $\frac{t_i}{2\pi \I} (1 \leq i \leq n)$ at punctures lying over $z=\infty$.   
\end{pro}

\begin{proof}
For punctures lying over $z=0$, equation \eqref{eq:branches} directly gives the values $-\frac{\lambda^{(n)}_i}{2\pi \I}$ of residues corresponding to $n$ branches.

Near $z=\infty$, $\mu(z)$ also has $n$ branches of the form
\begin{equation}
\mu_{i}(z)=\I \, u_{i}-\frac{1}{2\pi \I}\frac{t_{i}}{z}+O\left(\frac{1}{z^{2}}\right), \quad 1\leq i \leq n,\label{eq:residue-reg}
\end{equation}
where $t_i=A_{ii}$ is the $i^{th}$ diagonal entry of $A$. Indeed, let $\mu(z) = \Sigma_{k=0}^{\infty} c_kz^{-k}$ and substitute this power series into the defining equation of $\Gamma(u,A)$ to obtain
\[
\det \left(\I \, u - \frac{A}{2\pi \I z} - \mu I_n\right) = \prod_{i=1}^{n} \left(\I \, u_i - \frac{t_i}{2 \pi \I z} - c_0 - \frac{c_1}{z}\right) + h(z^{-1}) = 0.
\]
Here $h(z^{-1})$ is a power series in $z^{-1}$ with the lowest degree $\geq 2$.
Thus, for the constant term in $z^{-1}$ to vanish,  $c_0$ has to be equal to $\I \, u_i$ for $i=1, \dots, n$. By putting $c_0 = \I \, u_i$ (for some $i$), comparison of coefficients in front of $z^{-1}$ yields 
$$
c_1 = - \frac{t_i}{2 \pi  \I}
$$
giving the value of the residue at $i^{th}$ branch.
\end{proof}

\subsection{Degeneration: $u=E_n$}       \label{ssec:degeneration}

In this subsection, we study the spectral curve $\Gamma(u, A)$ in a degenerate case 
of $u=E_n (n\geq 3)$, where $E_n$ is the diagonal matrix with  $(E_n)_{nn}=1$, and all other matrix entries equal to zero. It turns out that in this case $\Gamma(u, A)$ is a rational curve.

\begin{pro}      \label{pro:curve_u=E_n}
    For $u=E_n (n\geq 3)$ and  $A \in \mathfrak{gl}_{n}(\mathbb{C})$ satisfying the two genericity conditions in \autoref{pro:generic-smooth} together with the third condition that all eigenvalues of $A^{(n-1)}$ are distinct,  the spectral curve $\Gamma(u,A)$ is isomorphic to $\mathbb{P}^{1}$ with $2n$ punctures.
\end{pro}

\begin{proof}
We continue to follow the notations used in the proof of \autoref{pro:generic-smooth}.
We have
\begin{align} \label{eq:pole-irr}
 P(\mu,z^{-1}) = \left(\I-\frac{1}{2\pi\I}\frac{t_{n}}{z}-\mu\right)\cdot{\rm det}\left(-\frac{1}{2\pi\I}\frac{A^{(n-1)}}{z}-\mu I_{n-1}\right)+h(\mu, z^{-1})  
\end{align}
where $h(\mu, z^{-1})$ has degree $\leq n-2$ in the $\mu$ variable. And if  $h(\mu, z^{-1}) \neq 0$, it is a homogeneous polynomial in $\mu$ and $z^{-1}$ of degree $n$. 

Note that for $n \geq 3$, $E_n \notin \mathfrak{h}_{\rm reg} (\mathbb{C})$. In this case,  $\partial_{\mu}P(\mu,z^{-1})$ will have non-simple zeros at $z=\infty$.
Let us fist compute the order of zero of $\partial_{\mu}P(\mu,z^{-1})$ at $z=\infty$. We claim that near $z=\infty$, 
\begin{align}\label{eq:inter-claim}
    P (\mu, z^{-1}) = (-1)^{n} \cdot \prod_{i=1}^{n-1} \left( \mu + \frac{\lambda_{i}^{(n-1)}}{2\pi \I z} + O(z^{-2}) \right) \cdot \left( \mu - \I + \frac{t_n}{2 \pi \I z}  + O(z^{-2})  \right).
\end{align}
Indeed, let $\mu(z^{-1})$ be one of $n$ branches of solutions of
the equation
\begin{equation} \label{eq:PP-nu}
P \left (\mu(z^{-1}), z^{-1} \right) = 0
\end{equation}
in a neighborhood of $z=\infty$.    
Expand $\mu(z^{-1}) = \sum_{k=m}^{\infty} c_{k} z^{-k}$ with $c_{m} \neq 0 \,(m\geq 0)$.
\begin{itemize}
    \item
If $\mu(z^{-1})$ does not vanish at $z = \infty$, 
using \eqref{eq:pole-irr}, the constant coefficient of $P \left (\mu(z^{-1}), z^{-1} \right)$ 
is equal to $\I (-1)^{n-1}c_0^{n-1}$ $+(-1)^{n}c_0^{n}$.
Requiring this expression to vanish yields $c_0 = \I$. Next, the coefficient of the $z^{-1}$
term is equal to $(-\frac{t_{n}}{2\pi \I}-c_{1})(-1)^{n-1}c_{0}^{n-1}$, and
requiring it in turn to vanish gives $c_{1}= - \frac{t_n}{2\pi \I}$.

\item 
If $\mu(z^{-1})$ has a zero at $z = \infty$, then $c_0 =0$. Using \eqref{eq:pole-irr}, the coefficient of $z^{-(n-1)}$in $P \left (\mu(z^{-1}), z^{-1} \right)$ is $\I \det\left(-\frac{1}{2\pi\I}A^{(n-1)}-c_{1}I_{n-1}\right)$. Requiring it to vanish shows that $c_{1}$ is an eigenvalue
of $A^{(n-1)}$ scaled by $-\frac{1}{2\pi\I}$.
\end{itemize}
This finishes the proof of the claim \eqref{eq:inter-claim}. Let $\mu_i (z^{-1})\,(1\leq i \leq n-1)$ be the branch near $z= \infty$ such that $\mu_i (z^{-1}) = -\frac{\lambda_{i}^{(n-1)}}{2\pi \I z} + O(z^{-2})$. Let $\mu_n (z^{-1})$ be the branch near $z= \infty$ such that $\mu_n (z^{-1}) = \I - \frac{t_n}{2 \pi \I z}  + O(z^{-2})$. Then, for $1 \leq i \leq n-1$, 
\begin{align*}
\partial_{\mu} P \left (\mu(z^{-1}), z^{-1} \right) \mid_{\mu= \mu_{i}} &= (-1)^n \prod_{j\mid j\neq i}\left(\mu_{i}(z^{-1})-\mu_{j}(z^{-1})\right) \\
  &= (-1)^n \prod_{j\mid j\neq i, n} \left(-\frac{\lambda_{i}^{(n-1)} - \lambda_{j}^{(n-1)}}{2\pi \I z}+ O(z^{-2})   \right)\cdot \left( -\I + \frac{t_n-\lambda_{i}^{(n-1)}}{ 2 \pi \I z} + O(z^{-2})  \right),
\end{align*}
and for $i = n$, 
\begin{align*}
\partial_{\mu} P \left (\mu(z^{-1}), z^{-1} \right) \mid_{\mu= \mu_{n}} &= (-1)^n \prod_{j\mid j\neq n} \left(\mu_{n}(z^{-1})-\mu_{j}(z^{-1})\right) \\
  &= (-1)^n \prod_{j\mid j\neq i, n} \left(\I + \frac{\lambda_{j}^{(n-1)} - t_n}{2\pi \I z}+ O(z^{-2})   \right)
\end{align*}
By our assumption that $A^{(n-1)}$ has distinct eigenvalues, we conclude that
$\partial_{\mu} P \left (\mu, z^{-1} \right)$ has a zero of order $(n-1)(n-2)$ at $z=\infty$. 

Now, it follows from computations we have already done in \autoref{pro:generic-smooth} that the number of zeros of $\partial_{\mu}P(\mu, z^{-1})$ for $z \in \mathbb{C}\setminus \{0 \}$,  is equal to the order of poles at $z=0$ minus the order of zeros at $z=\infty$, i.e., 
\[
n(n-1) - (n-1)(n-2) = 2(n-1). 
\]
All these zeros are simple by the genericity condition ii), so $\Gamma(u,A)$ is still smooth by \autoref{lem:The-discriminant-}. However, the compactification $\bar{\Gamma}(u,A)$ defined in the proof of \autoref{pro:generic-smooth} now has singularity at $(\mu =0, z= \infty)$. To resolve the singularity at $(\mu =0, z= \infty)$, we
blow up the ambient $\mathbb{P}^{1}_{\mu} \times \mathbb{P}^{1}_{z}$ at $(\mu = 0, z = \infty)$.

Let $\tilde{\overline{\Gamma}} (u,A)$ be the closure of $\Gamma(u,A)$ after blowing up $\mathbb{P}^{1}_{\mu} \times \mathbb{P}^{1}_{z}$ at $(\mu = 0, z = \infty)$ and $(\mu = \infty, z=0)$. In particular, $\tilde{\overline{\Gamma}} (u,A)$ is smooth and unramified over $z=0$ and $z=\infty$. Moreover, the ramification points of $\tilde{\overline{\Gamma}} (u,A) \rightarrow \mathbb{P}^1$  are the same as those of $\Gamma(u,A) \rightarrow \mathbb{P}^1$. They are all of index $2$ by the genericity condition ii) and their number is equal to the number of zeros of $\partial_{\mu}P(\mu, z^{-1})$ for $z \in \mathbb{C}\setminus \{0 \}$. Now we can use the Riemann-Hurwitz formula again, giving
$$
2g(\tilde{\overline{\Gamma}})-2=n(-2)+2n-2 \Rightarrow 
g(\tilde{\overline{\Gamma}})=0,
$$
i.e., $\tilde{\overline{\Gamma}}$ is a rational curve. 
\end{proof}

The following proposition gives values of residues of $\omega$ at punctures of the spectral curve:

\begin{pro}\label{pro:resi}
    For $u=E_n (n \geq 3)$ and $A$ such that eigenvalues of $A^{(k)}$ are distinct for $k=n-1, n$, the canonical 1-form $\omega$ has residues 
    \\

    i) $\frac{\lambda_i^{(n-1)}}{2\pi \I}$ for $i=1, \dots, n-1$ for $(n-1)$ punctures lying over $z=\infty$;

    ii) $\frac{t_n}{2\pi \I}$ at the $n^{th}$ puncture lying over $z=\infty$;

    iii) $-\frac{\lambda_{i}^{(n)}}{2\pi \I}$ for $i=1, \dots, n$ at punctures lying over $z=0$.
\end{pro}

\begin{proof}
For punctures lying over $z=\infty$, the statement follows from equations \eqref{eq:inter-claim}. For punctures lying over $z=0$, the result and the proof follow verbatim the result and the proof of \autoref{pro:residue_z=0}.
\end{proof}

\subsection{Degeneration: $u=u_{\rm cat}(t)$}\label{ssec:degeneration_cat}

For the rest of this Section, we assume that $n\geq3$.
Notice that we also have the translation symmetry on the level of spectral curves, not only on the level of Stokes matrices (c.f. \autoref{subsect-dCP}): $\Gamma(u,A)$ and $\Gamma(u -cI_n,A)$ are isomorphic for any $c\in \mathbb{C}$ and their periods are the same. So, throughout the rest of this section, we just fix $u_1$ to be  $0$, and only vary parameters $u_2, u_3, \cdots, u_n$, as a concrete way to parametrize the space $\mathfrak{h}_{\rm reg} (\mathbb{C})/\mathbb{C}$ or $\mathfrak{h}_{\rm reg} (\mathbb{R})/\mathbb{R}$. Under this parametrization, the caterpillar line limit $u \rightarrow u_{\rm cat}(t)$ acquires the form 
$$
u_2=t, \hskip 0.3cm \frac{u_{j+1}}{u_j} \rightarrow +\infty \,\, {\rm for} \,\, 
j=2, \dots, n-1,
$$
where $t>0$ is a fixed constant that determines the position on the caterpillar line.

Since most of our proofs work in the complex setting, it will be more convenient to consider the complex De Concini-Procesi space (c.f. \autoref{subsect-dCP}) $\widehat{\mathfrak{h}_{\rm{reg}}}(\mathbb{C})$  in this subsection. The real De Concini-Procesi space and the caterpillar line are just the restriction of $\widehat{\mathfrak{h}_{\rm{reg}}}(\mathbb{C})$ to the real locus. 

\begin{pro}\label{prop:1-par-degen}
Consider the 1-parameter family of spectral curves $\Gamma(u,A)$ $(u \in \mathfrak{h}_{\mathrm{reg}}(\mathbb{C}))$ we get by fixing $ u_2 $, $\cdots, u_{n-1}$ and varying $u_n$. Then, for a generic fixed $A \in \mathfrak{gl}_{n}(\mathbb{C})$, at $u_n =  \infty$, the spectral curve degenerates into two irreducible components $\Gamma_{n-1}$ and $\Gamma^{(n)}$; $\Gamma_{n-1}$ is a smooth curve of genus $\frac{(n-2)(n-3)}{2}$ with $2(n-1)$ punctures, and $\Gamma^{(n)}$ is isomorphic to $\mathbb{P}^{1}$ with $2n$ punctures.
\end{pro}
\begin{proof}
By fixing $u_2, \dots, u_{n-1}$ and allowing $u_{n}$ to vary in $\mathbb{C}$, we view \eqref{eq:sp-c}
as defining a $1$-parameter family of spectral curves parametrized
by $\mathbb{C}_{u_{n}}$ in $\mathbb{P}^{1}_{\mu}\times\mathbb{P}^{1}_{z}\times\mathbb{C}_{u_{n}}$. We denote this 1-parameter family by $\mathscr{X}$.

Next, let $\mu'=\frac{1}{\mu}$ and $u'_{n}=\frac{1}{u_{n}}$.
Let $V=\mathbb{P}^{1}_{\mu}\times\mathbb{P}^{1}_{z}\times\mathbb{P}^{1}_{u'_{n}}$.
Let $\tilde{V}$ be the blowup of $V$
at $\mu=\infty,z=0,u_{n}=\infty$. Then, $\tilde{V}$ is a closed subscheme
of $V\times\mathbb{P}^{2}$. Let $y_{1},y_{2},y_{3}$ be homogeneous
coordinates on $\mathbb{P}^{2}$ satisfying 
\[
\mu'y_{3}=u'_{n}y_{1},\quad z y_{3}=u'_{n}y_{2}.
\]
Set $y_{3}=1$ and let $ w = \frac{\mu}{u_{n}}=\frac{u'_{n}}{\mu'}=\frac{1}{y_{1}}$
and $v = z u_{n} = y_2$. 

Now, to take the limit $u_{n}\rightarrow \infty$, we need to consider the strict transform of $\mathscr{X}$ in $\tilde{V}$ and then take its closure $\bar{\mathscr{X}}$ in $\tilde{V}$ when $u_n=\infty$. 
Given $u_n \in \mathbb{C}$, on the $\mathbb{P}^2$ component of $V\times{\mathbb{P}^2}$, the spectral curve needs to satisfy the equation 
\begin{align}
\det\left(\frac{1}{u_{n}}\cdot \I u-\frac{A}{2\pi \I}\frac{1}{z u_{n}}-\frac{\mu}{u_{n}}\cdot I_{n}\right) = \det\left(\frac{1}{u_{n}}\cdot \I u-\frac{A}{2\pi \I v}-wI_{n}\right) = 0 \, .\label{eq:a}
\end{align}
On the chart where $z^{-1},\mu \in \mathbb{C}$, $V$ and $\tilde{V}$ are isomorphic,
\eqref{eq:sp-c} reads 
\[
\left(\I u_{n} - \frac{t_n}{2 \pi \I z} - \mu\right) \cdot\det\left(\I u^{(n-1)}-\frac{A^{(n-1)}}{2\pi \I z}-\mu I_{n-1}\right)+f=0 \, ,
\]
where $f$ is a polynomial in the variables $z^{-1}$ and $\mu$ 
whose coefficients
do not involve $u_{n}$. For $z^{-1}, \mu \in\mathbb{C}$ and $u_{n} \neq 0$,
we divide both sides of the above equation by $\I u_n$ and get
\[
\det\left(\I u^{(n-1)}-\frac{A^{(n-1)}}{2\pi \I z}-\mu I_{n-1}\right) \left[1 +\frac{1}{\I u_{n}}(- \frac{t_n}{2 \pi \I z} - \mu)\right] +\frac{f}{\I u_{n}}=0 \, .
\]
When $u_{n} = \infty$, the above equation goes to the limit
\begin{align}\label{eq:c}
   \det \left( \I u^{(n-1)}-\frac{A^{(n-1)}}{2\pi \I z}-\mu I_{n-1} \right) = 0 \, . 
\end{align}
By \autoref{pro:generic-smooth}, for a generic $A\in \mathfrak{gl}_n(\mathbb{C})$, \eqref{eq:c} defines a spectral curve $\Gamma_{n-1}$ of genus $\frac{(n-2)(n-3)}{2}$ with $n-1$ punctures above $z=0$ and $n-1$ punctures above $z=\infty$. We denote by $s_{i}$ the puncture above $z=\infty$ on $\Gamma_{n-1}$ where the canonical differential $\omega = \mu dz$ has residue $\frac{t_{i}}{2\pi \I}$, and by $r_i$ the puncture above $z=0$ where $\omega$ has residue $-\frac{\lambda_{i}^{(n-1)}
}{2 \pi \I}$.

If $z^{-1}= \mu =\infty$, then, when $u_n=\infty$, we get an extra component in $\bar{\mathscr{X}}\setminus{\mathscr{X}}$, given by the limit  of \eqref{eq:a}, which is
\begin{align} \label{eq:sp-d-b}
   P(w,v^{-1})= \det\left(\I E_{n}-\frac{A}{2\pi \I v}-wI_{n}\right)=0 \, ,
\end{align}
which defines a spectral curve $\Gamma^{(n)}$. Notice that on $\Gamma^{(n)}$ the canonical differential is
\begin{align}
    \omega = \mu dz = u_n w \cdot \frac{dv}{u_n} = w \cdot dv \, .
    \end{align}
By \autoref{pro:curve_u=E_n}, for a generic $A \in \mathfrak{gl}_n (\mathbb{C})$, $\Gamma^{(n)}$ is isomorphic to $\mathbb{P}^{1}$ with $2n$ punctures. 
There are $n$ punctures above $v=\infty$. $(n-1)$ of them are the punctures $r_i$ $(1\leq i \leq n-1)$ that $\Gamma^{(n)}$ shares with $\Gamma_{n-1}$. We label the last puncture over $v=\infty$ by $s_n$. The canonical differential has residue $\frac{t_{n}}{2\pi\I}$ at $s_n$. There are $n$ punctures on  $\Gamma^{(n)}$ above $v=0$. We label by $q_i$ $(1 \leq i \leq n)$ the puncture above $v=0$ where the canonical differential has residue $-\frac{\lambda_{i}^{(n)}}{2\pi\I}$.
\end{proof}

For future use, we introduce the following notation. For $2\leq k \leq n$, let $\Gamma^{(k)}$ be the spectral curve defined by \eqref{eq:sp-d-b} with $E_n$ replaced with $E_k$ and $A$ with $A^{(k)}$, i.e., given by the equation
\begin{align} \label{eq:sp-d-k}
   P_{k}(\mu, z^{-1})= \det\left(\I E_{k}-\frac{A^{(k)}}{2\pi \I z}- \mu I_{k}\right)=0 \, .
\end{align}
We emphasize that we still view $\Gamma^{(k)}$ as a punctured curve with punctures over $z=0$ and over $z= \infty$.

From now on, it will be convenient for us to restrict $A$ to $\mathrm{Herm}_{0}(n)$. However, to consider algebraic families of spectral curves, we still want to allow the $u$-parameters to vary in $\widehat{\mathfrak{h}_{\rm{reg}}}(\mathbb{C})$. 

\begin{defi}\label{def:discr}
Denote by $\mathcal{B} \subset \mathrm{Herm}_{0}(n)$ the locus where for each $3\leq k \leq n$ the discriminant of the polynomial $P_{k}(\mu,z^{-1})$ has only simple zeroes as a polynomial in $\mu$ for $z \in \mathbb{C} \setminus \{ 0 \}$. 
We call $\mathcal{B}$ the \emph{generic locus} of $\mathrm{Herm}(n)$. 
 \end{defi}
 \begin{rmk}
Here is the reason why we define $\mathcal{B}$ to be contained in $\mathrm{Herm}_{0}(n)$ in the first place. If the spectral curve defined by \eqref{eq:sp-d-k} is smooth for each $k$, $A$ has to be in $\mathrm{Herm}_{0}(n)$. For otherwise, \eqref{eq:sp-d-k} will split for some $k$. Then, the corresponding spectral curve $\Gamma^{(k)}$ will be reducible. 
 \end{rmk}

\begin{lem}\label{lem:all-iso}
Fix a $A \in \mathcal{B}$.  For any  $u_1, u_2 \in \mathbb{C}$ such that $u_1 \neq u_2$, the spectral curve defined by
\[
  \det\left(\I u^{(2)}-\frac{A^{(2)}}{2\pi \I z}- \mu I_{2}\right)=0 \, 
\] is isomorphic to $\Gamma^{(2)} (A)$. Here, $u^{(2)}$ is the diagonal matrix whose first diagonal entry is $u_1$ and the second one is $u_2$.
\end{lem}
\begin{proof}
It suffices to consider the case where $u_1 =0$. The map $z \mapsto z \cdot (u_2 - u_1)$, $\mu \mapsto \mu/ (u_2 - u_1)$ gives the isomorphism to 
$\Gamma^{(2)} (A)$.
\end{proof}

\begin{lem} \label{lem:smo}
For a fixed $A \in \mathcal{B}$,  $\Gamma^{(k)} (A)$ is isomorphic to  $\mathbb{P}^{1}$ with $2k$ punctures for $2 \leq k \leq n$.  
\end{lem}
\begin{proof}
For $3 \leq k \leq n$, since $A^{(k)}$ satisfies the genericity conditions in \autoref{pro:curve_u=E_n} for $A \in \mathcal{B}$, $\Gamma^{(k)}$ is isomorphic to  $\mathbb{P}^{1}$ with $2k$ punctures for $3 \leq k \leq n$ by \autoref{pro:curve_u=E_n}. 

For $k=2$, it is a straightforward computation that  $P_{2}(\mu,z^{-1})$ has only simple zeroes as a polynomial in $\mu$ for $z \in \mathbb{C} \setminus \{ 0 \}$ if $A_{12}\cdot A_{21} \neq 0$. But this condition is automatically satisfied for $A \in \mathcal{B}$. Then, the case of $k=2$ follows from \autoref{pro:generic-smooth}.
\end{proof}

Now, we are ready to describe the degeneration of the spectral curves at a point on the caterpillar line:
\begin{thm}\label{thm:smo-nei}
Fix $A \in \Herm(n)$ and vary $u \in \widehat{\mathfrak{h}_{\rm{reg}}}(\mathbb{C})$. Let $u_{\rm cat} (t)$ $(t>0)$ be a point on the caterpillar line. The family of spectral curves can be extended over $u_{\mathrm{cat}}(t)$, 
and the fiber over $u_{\mathrm{cat}}(t)$ has $(n-1)$ components isomorphic to $\Gamma^{(k)}(A) (2 \leq k \leq n)$. Moreover, if $A \in \mathcal{B}$,  the $(n-1)$ components in the fiber over $u_{\mathrm{cat}}(t)$ are irreducible and isomorphic to $\mathbb{P}^{1}$ with $2k$ punctures $(2 \leq k \leq n)$ and there exists an open neighborhood $U\subset\widehat{\mathfrak{h}_{\rm{reg}}}(\mathbb{C})$ around $u_{\rm{cat}} (t)$ such that for any $u\in U \cap \mathfrak{h}_{\rm{reg}} (\mathbb{C})$, the spectral curve $\Gamma(u,A)$ is smooth.
\end{thm}
 
\begin{proof} 
The limit $u \rightarrow u_{\rm cat} (t)$ can be achieved iteratively. First, we fix $u_2, \dots, u_{n-1}$ and 
let $u_{n} \to + \infty$; then we fix $u_2, \dots, u_{n-2}$ and let $u_{n-1} \to + \infty$; and so on until we reach the point $u_{\mathrm{cat}} (t)$ by letting $u_3 \to + \infty$ while fixing $u_2$. Thus, we are iterating the degeneration process in \autoref{prop:1-par-degen}, and the spectral curve degenerates into $(n-1)$ components at $u_{\rm cat}(t)$. These components are isomorphic to $\Gamma^{(k)}(A) (2 \leq k \leq n)$ by \autoref{prop:1-par-degen} and \autoref{lem:all-iso}. If $A \in \mathcal{B}$, it follows from \autoref{lem:smo} these components are smooth, irreducible, and isomorphic to $\mathbb{P}^1$ with $2k$ punctures. 

 Again, for $A \in \mathcal{B}$, for $u\in \mathfrak{h}_{\rm reg} (\mathbb{C})$ close enough to $u_{\rm{cat}} (t)$, \eqref{eq:sp-c} has only simple zeroes; thus, by \autoref{lem:The-discriminant-} and \autoref{lem:simple-zero-smoothness}, $\Gamma(u,A)$ is smooth.
  In particular, there exists a neighborhood $U\subset\widehat{\mathfrak{h}_{\rm{reg}}}(\mathbb{C})$ around $u_{\rm{cat}} (t)$ such that for any $u\in U \cap \mathfrak{h}_{\rm reg}$, the discriminant of \eqref{eq:sp-c} has only simple zeroes. The existence of such a $U$ follows from the fact that $\mathrm{Herm}_{0}(n) \setminus \mathcal{B}$ is cut out by polynomial conditions. Then, by \autoref{lem:The-discriminant-} and \autoref{lem:simple-zero-smoothness}, the spectral curve $\Gamma(u,A)$ is smooth for any $u\in U \cap \mathfrak{h}_{\rm reg} (\mathbb{C})$.

We also give an alternative proof for the first claim of the proposition, so that readers can see all the degenerations in one step rather than iteratively. Fixing $u_2$, we view \eqref{eq:sp-c} as defining a family of spectral curves $\mathscr{X}$ in $V= \mathbb{P}^{1}_{\mu} \times \mathbb{P}^{1}_{z} \times \mathbb{P}^{1}_{u_3} \times \mathbb{P}^{1}_{d_4} \times \cdots \times \mathbb{P}^{1}_{d_n}$ where $d_i = \frac{u_i}{u_{i-1}}$ for $4 \leq i \leq n$.

After dividing both sides by $(\I u_3) \cdots (\I u_n)$, \eqref{eq:sp-c} reads 
\begin{align}
    \det \left(\I u^{(2)} - \frac{A^{(2)}}{2\pi \I z} - \mu I_2 \right) \left[1+ \frac{1}{\I u_3}(-\frac{t_3}{2\pi \I z} - \mu) \right]\cdots \left[ 1 + \frac{1}{\I u_n}(-\frac{t_n}{2\pi \I z} - \mu) \right] + h \label{eq:m}
\end{align} 
where $h$ is a polynomial in $z^{-1}$ and $\mu$, whose coefficients all contain some negative powers of $u_3, \dots, u_n$. For $z^{-1} \in \mathbb{C}$, as we go to the limit $\frac{u_{j+1}-u_{j}}{u_{j}-u_{j-1}}\rightarrow \infty$, or equivalently $\frac{u_{j+1}}{u_j} \rightarrow \infty$, for all $j=2,\dots,n-1$, we have $u_3, \dots, u_n \rightarrow \infty$. Then, \eqref{eq:m} goes to the limit
\begin{align}
    \det \left(\I u^{(2)} - \frac{A^{(2)}}{2\pi \I z}- \mu I_2\right) = 0, \label{eq:l}
\end{align}
which defines a spectral curve $\Gamma^{(2)}$ isomorphic to 4-punctured $\mathbb{P}^{1}$ for generic $A$. 

Let $\nu_k = z u_k$ and $w_k =  \tfrac{\mu}{u_k}  $ for $3 \leq k \leq n$. We do a sequence of blowups of $V$ as follows. First, we blow up the subscheme of $V$ defined by $z^{-1} = \mu =  u_3 = \infty$ and get the scheme $V_3$. Inductively, we blow up the subscheme of $V_i$ defined by $v_k^{-1} = w_k= d_{k+1} = \infty$ and get the scheme $V_{k+1}$ for $3 \leq k \leq n-1$.  

We claim that  as $\frac{u_{j+1}-u_{j}}{u_{j}-u_{j-1}}\rightarrow \infty$ for all $j=2,\dots,n-1$, for each $3 \leq k \leq n$, we get a component $\Gamma^{(k)}$ of the closure of $\mathscr{X}$ inside $V_n$, given by
\begin{align}
    \det\left(\I E_{k}-\frac{A^{(k)}}{2\pi \I \nu_k}-w_k I_{k}\right)=0 \quad (3 \leq k \leq n) = 0, \label{eq:n}
\end{align}
which defines a $(2k)$-punctured $\mathbb{P}^{1}$ for generic $A$. 
The proof of this claim uses analysis similar to what we did in the proof of \autoref{prop:1-par-degen}. Consider the case $k=3$. Let $d'_4 = d_4$, $d'_5 = d_4 d_5$, $\dots$, $d'_n = d_4 d_5 \cdots d_n$. Now \eqref{eq:sp-c}, after dividing both sides by $u_3^{n}(\I d'_4) (\I d'_5)\cdots(\I d'_n)$, reads
\begin{align}
    \det \left(\I \frac{u^{(3)}}{u_3} - \frac{A^{(3)}}{2\pi \I \nu_3} - w_3 I_3 \right) \prod_{i=4}^{n} \left[1+\frac{1}{\I d'_i} \left(- \frac{t_i}{2\pi \I \nu_3} - w_3 \right) \right] +h_3 \, , \label{eq:hj}
\end{align}
where $h_3$ is a polynomial in $\nu_3^{-1}$ and $w_3$ whose coefficients all contain some negative power of $d'_4,\,d'_5, \cdots \, d'_n$. In the limit $u_3, d_4, \cdots, d_n \rightarrow \infty$, for $\nu_3 \neq 0$, 
 \eqref{eq:hj} becomes
\begin{align}
  \det\left(\I E_{3}-\frac{A^{(3)}}{2\pi \I \nu_3}-w_3 I_{3}\right)=0 \, .
\end{align}
The proof for $4\leq k \leq n$ is similar.
\end{proof} 

\begin{defi} (\emph{Spectral curves at points on the caterpillar line})
Fix a point $u_{\rm cat}(t)$ $(t>0)$ on the caterpillar line. It follows from \autoref{thm:smo-nei} that for any $A \in \Herm(n)$, taking the limit $u \rightarrow u_{\rm cat} (t)$, the family of spectral curves degenerates into $n-1$ (not necessarily irreducible or smooth) components, each of which is defined by \eqref{eq:n} for $3 \leq k \leq n$ and by \eqref{eq:l} for $k=2$. We call the degeneration at $u_{\rm cat} (t)$ the \emph{spectral curve at $u_{\rm cat} (t)$} and denote it by $\Gamma (u_{\rm cat} (t), A)$.
\end{defi}

\begin{cor} \label{cor:vani}
 Fix $A \in \mathcal{B}$ and a point $u_{\rm cat} (t)$ ($t>0$) on the caterpillar line. There exists a punctured open neighborhood $B(u_{\rm cat} (t), A) \subset U_{\rm id}$ around $u_{\rm cat} (t)$ such that for all $u \in B(u_{\rm cat} (t), A)$, $\Gamma(u,A)$ is smooth. Moreover, for any $u \in B(u_{\rm cat} (t), A)$, there exists \emph{vanishing cycles} $V^{(k)}_{j} (u,A) (1 \leq j \leq k \leq n)$, such that  $V^{(k)}_{j} (u,A)$ that becomes a small loop $V^{(k)}_{j} \left(u_{\rm cat} (t),A \right)$ on $\Gamma(u_{\rm cat} (t), A)$ as $u \rightarrow u_{\rm cat} (t)$. For $2 \leq k \leq n$, $V^{(k)}_{j} \left(u_{\rm cat} (t),A \right) (1 \leq j \leq k)$ is on $\Gamma^{(k)}$, oriented clockwise around the puncture where the canonical 1-form $\omega$ has residue $-\frac{\lambda_{j}^{(k)}}{2\pi \I}$. For $k=1$, $V_1^{(1)} (u_{\rm cat} (t), A)$ is on $\Gamma^{(2)}$, oriented counterclockwise around the puncture where $\omega$ has residue $\lambda_{1}^{(1)} = \frac{t_1}{2 \pi i}$. Here, $t_1$ is the first diagonal entry of $A$.
\end{cor}

\begin{proof}
Let $U$ be a neighborhood of $u_{\rm cat} (t)$ in $\widehat{\mathfrak{h}_{\rm{reg}}}(\mathbb{C})$  as in \autoref{thm:smo-nei}. Let $B(u_{\rm cat} (t), A) = U \cap U_{\rm id}$. Then, the first statement of this corollary follows from \autoref{thm:smo-nei}. 

It remains to prove statements about vanishing cycles. For $k=1$ and $k=n$, these vanishing cycles are just cycles around corresponding punctures with the prescribed residue by \autoref{pro:residue_z=0} and \autoref{pro:resi}.  For $2 \leq k \leq n-1$,  it follows from \autoref{pro:resi} that the punctures where $\omega$ has the residue $-\frac{\lambda_{j}^{(k)}}{2\pi \I}$ are the nodes of the (compactified) spectral curve $\Gamma(u_{\rm cat} (t), A)$.  Let $\Delta \hookrightarrow U$ be an analytic disc such that $\Delta \setminus \{0 \}$ is contained in $U \cap \mathfrak{h}_{\rm reg} (\mathbb{C})$ and $0$ is mapped to $u_{\rm cat} (t)$. Let $\mathcal{C} \rightarrow \Delta$ be the family of (fiberwise compactified) spectral curves $\Gamma (u,A) (u \in \Delta)$. Since $\Gamma (u,A)$ is smooth for $u \in \Delta \setminus \{0 \}$ and $\Gamma (u_{\rm cat} (t), A)$ has only nodal singularities, $\mathcal{C}$ is smooth and $\mathcal{C} \rightarrow \Delta$ is holomorphic Morse. Then, the existence of vanishing cycles follows from the standard holomorphic Morse theory, e.g., Chapter II of \cite{voisin II}. 
\end{proof}

\begin{rmk}
The irregularity in the definition of vanishing cycles in the case of $k=1$ came from the fact that in this paper we use the \emph{caterpillar line} rather than going to the \emph{caterpillar point} as in the existing literature \cite{HKRW, Sp}). 
We made this choice because the formulas for Stokes
matrices are already exact on the caterpillar line: from this point of view,
there is no need to make the last degeneration to the caterpillar point.
\end{rmk}

\subsection{Distinguished cycles on $\Gamma(u, A)$}     \label{ssec:distinguished}
Fix $A \in \mathcal{B}$ and a point $u_{\rm cat} (t)$ ($t>0$) on the caterpillar line. Let $B(u_{\rm cat} (t), A) \subset U_{\rm id}$ be an open neighborhood as in \autoref{cor:vani}. For $u \in B(u_{\rm cat} (t), A)$, we define \emph{distinguished cycles} $C^{(k)}_i (u,A)$ on $\Gamma (u,A)$ as  linear combinations of vanishing cycles
\begin{equation}      \label{eq:cycles_C^k_i}
C^{(k)}_i(u, A) = \sum_{j=k-i+1}^k V^{(k)}_j(u, A).
\end{equation}
For $u= u_{\rm cat} (t)$, we define \emph{distinguished cycles} 
\begin{equation} \label{eq: cycles_cat}
    C^{(k)}_i \left(u_{\rm cat} (t), A \right) = \sum_{j=k-i+1}^k V^{(k)}_j \left(u_{\rm cat} (t), A \right)
\end{equation}
on $\Gamma \left( u_{\rm cat} (t), A \right) $ where 
$ V^{(k)}_j \left(u_{\rm cat} (t), A \right)$ are as in \autoref{cor:vani}.

\begin{thm}       \label{thm:limit_periods} 
   For $u \in B(u_{\rm cat} (t), A) \subset U_{\rm id}$, we have
    \begin{equation}
      \frac{1}{2} \,  {\rm lim}_{u \to u_{\rm cat}(t)} Z \left(C^{(k)}_i(u,A) \right) = \frac{1}{2}Z \left(C^{(k)}_i(u_{\rm cat} (t),A) \right) = -l^{(k)}_i(u_{\rm cat}, A).
    \end{equation}
In particular, \autoref{mainconj} holds on the caterpillar line.
\end{thm}

\begin{proof}
   By \autoref{cor:vani},
    $$
{\rm lim}_{u \to u_{\rm cat}(t)} Z \left(V^{(k)}_j(u,A) \right) =  \int_{V^{(k)}_j \left(u_{\rm cat} (t), A \right)} \omega = -\lambda^{(k)}_j.
    $$
Next, we compute
$$
\frac{1}{2} \,  {\rm lim}_{u \to u_{\rm cat}(t)} Z \left(C^{(k)}_j(u,A) \right) = 
\frac{1}{2}Z \left(C^{(k)}_i(u_{\rm cat} (t),A) \right)
= -\frac{1}{2} \sum_{j=k-i+1}^k \lambda^{(k)}_j = -l^{(k)}_i(u_{\rm cat}, A),
$$
where in the last equality we have used \autoref{thm:conj_at_ucat}.  
\end{proof}
Under \autoref{conj:WKB_section2}, \autoref{thm:limit_periods} combined with \autoref{pro:limit_l} suggests that near the caterpillar line, the distinguished cycles $C^{(k)}_i$'s may be suitable for the WKB analysis of the asymptotic behavior of minor coordinates. We will further explore this via the theory of spectral networks in the next section. 

Finally, we note that the periods of the canonical form over the vanishing cycles are real:  
\begin{pro}\label{pro:realness} 
Let $u \in B(u_{\rm cat} (t), A)$ and $A \in \mathcal{B}$. Then, the periods of the canonical 1-form $\omega$ over vanishing cycles $V_{j}^{(k)}(u,A)$ are real:
$$
Z \left(V^{(k)}_j \right) \in \mathbb{R}.
$$
\end{pro}

\begin{proof}
Under the reality conditions \eqref{eq:intro-reality} on $(u,A)$,
the spectral curve $\Gamma(u,A)$ has an involution $\iota(z,\mu) = (\bar{z},-\bar{\mu})$. Since 
$\omega = \mu\,dz$ we have 
$\iota^*\omega = -\bar\omega$. At $u = u_{\rm cat}(t)$, the punctures are fixed under $\iota$,
which implies that $\iota_* V_i^{(k)} = -V_i^{(k)}$. Indeed, since $\iota$ reverses 
the orientation of $\Gamma(u,A)$, it exchanges positive and negative orientations
of small loops around the punctures. By continuity, this is still true
for $u \in B(u_{\rm cat} (t), A)$. Thus, for any $V=V^{(k)}_j(u, A)$ we have
\begin{equation}
Z(V) = \int_{V} \omega = -\int_{\iota_* V} \iota^* \omega = \int_{-V} (- \bar\omega) = \int_V \bar\omega = \overline{\int_V \omega} = \overline{Z(V)},
\end{equation}
as desired.
\end{proof}

\subsection{Spectral curves for $A$ nearly diagonal}      \label{ssec:near-diagonal}

In this section, we consider the special case when $u$ is near the caterpillar line, and $A$ is nearly diagonal (that is, close to a diagonal matrix). In that case, we can explicitly see the degeneration phenomena described in the previous sections. This is also a preparation for the spectral network analysis of the next section. 

Throughout this section, we assume the reality conditions \eqref{eq:intro-reality}.

\subsubsection{The diagonal case}

Recall that we denote by $t_i = A_{ii}$ the diagonal entries of $A$. 
If $A$ is diagonal, then the eigenvalues of the operator
$\I\left(u + \frac{A}{2\pi z}\right)$ are
\begin{equation} \label{eq:rho-diagonal}
  \mu_i(z) = \I \left(u_i + \frac{t_i}{2\pi z}\right) \, .
\end{equation}
These are also the sheets of the spectral curve $\Gamma(u,A)$.

Note that $\mu_i(\overline{z}) = -\overline{\mu_i(z)}$, so in particular
 $\mu_i(z) \in \I \Real$ for $z \in \Real$.
The eigenvalues $\mu_i(z)$  are all distinct as long as $z \notin \{z_{ij}\}$ where
\begin{equation} \label{eq:def-zij}
  z_{ij} = - \frac{t_i - t_j}{2 \pi (u_i - u_j)} \, .
\end{equation}
At $z = z_{ij}$, we have $\mu_i(z) = \mu_j(z)$.
If all $z_{ij}$ with $i > j$ are distinct and nonzero, then they are
nodal singularities of $\Gamma(u,A)$, and
$\Gamma(u,A)$ is smooth away from these
$\frac{n(n-1)}{2}$ singular points.

We will generally be interested in the situation where
all $u_i$ are distinct; then we may as well assume
\begin{equation}
u_1 < \cdots < u_n,
\end{equation}
and we do so from now on.
We also often assume $u$ is close to the caterpillar line, as in the next proposition.
\begin{pro} \label{prop:branch-point-ordering}
Let $\Lambda = \max_{j \neq k,j' \neq k'} \abs{\frac{t_j - t_k}{t_{j'} - t_{k'}}}$.
Suppose $\frac{u_{i+1} - u_i}{u_i - u_{i-1}} > \Lambda$
for all $i$.
Then, if $j < j'$, $\abs{z_{ij}} > \abs{z_{ij'}}$.
\end{pro}
Said otherwise, if we divide the set of branch points $z_{ij}$ with $i < j$
into $n-1$ subsets 
$\{z_{12}\}$, $\{z_{13},z_{23}\}$, \dots, $\{z_{1,n}, z_{2,n}, \dots, z_{n-1,n}\}$,
these subsets are arranged in order of decreasing distance from the origin.
Within each subset, the ordering of the $z_{ij}$ matches the 
ordering of the $t_i$.

We often further assume $t_1 < \cdots < t_n$; this assumption is not essential, but it will simplify the exposition and the pictures below.
Then all $z_{ij} < 0$, and \autoref{prop:branch-point-ordering} says
\begin{equation}
z_{12} < z_{13} < z_{23} < \cdots < z_{1,n} < z_{2,n} < \cdots < z_{n-1,n} \, .
\end{equation}
See \autoref{fig:singular-points-gl4} for the $n=4$ case.
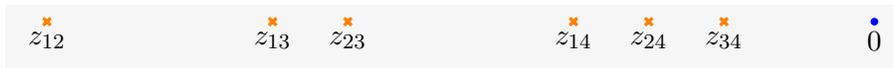
\begin{figure}[h]
\centering
\begin{tikzpicture}[withbackgroundrectangle]
    \foreach \x/\label in {1/z_{12},4/z_{13},5/z_{23},8/z_{14},9/z_{24},10/z_{34}} {
        \draw[branchpoint] plot coordinates {(\x,0)};
        \node[below] at (\x,0) {$\label$};
    }
    \draw[singularpoint] plot coordinates {(12,0)};
    \node[below] at (12,0) {$0$};
\end{tikzpicture}
\caption{The schematic arrangement of the points $z_{ij} \in \Comp$ in case $n=4$, with $A$ diagonal, when $t_1 < t_2 < t_3 < t_4$ and $u$ is sufficiently close to the caterpillar line.} \label{fig:singular-points-gl4}
\end{figure}

\subsubsection{The near-diagonal case}

For $u$ close to the caterpillar line, the 
effect of the off-diagonal entries of $A$ 
is to broaden the singular points of $\Gamma(u,A)$ into vertical branch cuts, as described in the next proposition.

Here and below, we will be considering some open subset of $\fh_{\reg}(\Real) \times \Herm(n)$, 
whose closure contains the locus where $u$ is on the caterpillar line and $A$ is diagonal.
When $(u,A)$ belongs to this open subset we say ``$u$ is sufficiently close to the caterpillar line and $A$ 
is sufficiently close to diagonal.''

\begin{pro} \label{prop:spectral-curve-structure}
For $u$ sufficiently close to the caterpillar line, and $A$ sufficiently
close to diagonal, with all off-diagonal entries of $A$ nonzero, and $t_1 < \cdots < t_n$:
\begin{enumerate}
\item For large $\abs{z}$, the operator $\I\left(u + \frac{A}{2\pi z}\right)$ has 
eigenvalues $\mu_i(z)$ obeying $\lim_{z \to \infty} \mu_i(z) = \I u_i$.
Each $\mu_i$ can be continued to an analytic function on $\Comp \setminus \cup_{j \neq i} B_{ij}$,
where each $B_{ij}$ is a vertical segment $B_{ij} = [z_{ij}^+, z_{ij}^-] \subset \Comp$,
and the ordering of the $\re z_{ij}^\pm$ is the same as the ordering of the $z_{ij}$.
\item 
Continuation across $B_{ij}$ exchanges $\mu_i(z)$
with $\mu_j(z)$, while all other $\mu_k(z)$ extend holomorphically
across $B_{ij}$. 
\item The $n(n-1)$ endpoints $z_{ij}^\pm$ are simple ramification points
of $\Gamma(u,A)$. 
\item $\Gamma(u,A)$ is smooth.
\item Each $\mu_i$ obeys the reality condition
    \begin{equation} \label{eq:omega-symmetry}
\mu_i(\overline{z}) = -\overline{\mu_i(z)} \, .
\end{equation}
\end{enumerate}
\end{pro}

\begin{figure}[h]
\centering
\begin{tikzpicture}[withbackgroundrectangle]
    \foreach \x/\label/\transpo in {1/z_{12}/12,4/z_{13}/13,5/z_{23}/23,8/z_{14}/14,9/z_{24}/24,10/z_{34}/34} {
        \draw[orange, thick, mark=x, mark options={orange, line width=1pt}] plot coordinates {(\x,0)};
        \draw[orange, thick, mark=x, mark options={orange, line width=1pt}] plot coordinates {(\x,1)};
        \draw[branchcut] (\x,0) -- (\x,1);
        \node[below] at (\x,0) {$\label^-$};
        \node[above] at (\x,1) {$\label^+$};
        \node[cutlabel] at (\x,0.5) {$(\transpo)$};
    }
    \draw[singularpoint] plot coordinates {(12,0.5)};
    \node[below] at (12,0.5) {$0$};
\end{tikzpicture}

\caption{The arrangement of the branch points $z^\pm_{ij}$ and branch cuts $B_{ij}$ in case $n=4$, when the conditions of \autoref{prop:spectral-curve-structure} are satisfied. Each branch cut is labeled with the induced transposition $(ij)$ of the sheets.} \label{fig:cuts-example}
\end{figure}
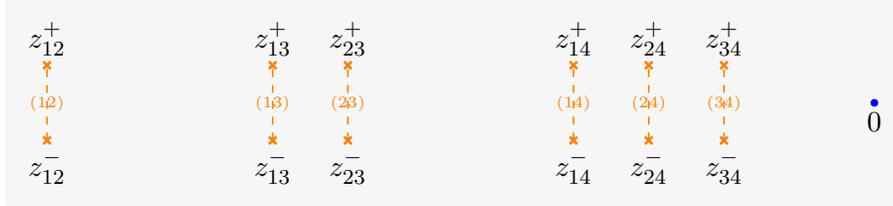

See \autoref{fig:cuts-example} for a schematic picture
of the cuts when $n=4$.

\begin{proof} Fix a pair $(i,j)$ and choose any neighborhood $U_{ij}$ which contains 
$z_{ij}$ and excludes neighborhoods of all other $z_{kl}$.
When $A$ is diagonal, the eigenvalues $\mu_k(z)$ are given by \eqref{eq:rho-diagonal},
so they are all distinct in $U_{ij}$,
except that $\mu_i(z_{ij}) = \mu_j(z_{ij})$.
It follows that,
when $A$ is near diagonal, 
$M(z) = \I\left(u + \frac{A}{2\pi z}\right)$ has $n-2$ eigenvalues $\mu_k(z)$ ($k \notin \{i,j\}$) which vary analytically with $z$ on $U_{ij}$;
the other two eigenvalues $\mu_i(z)$, $\mu_j(z)$ vary analytically
with $z$ only on some $U'_{ij}$ obtained by deleting a small disc around $z_{ij}$.
Still, the local discriminant
\begin{equation}
\Delta_{ij}(z) = - (\mu_i(z)  - \mu_j(z))^2    
\end{equation}
extends analytically to the whole $U_{ij}$.

When $A$ is diagonal, $\Delta_{ij}$ has a double zero at $z_{ij}$ and vanishes
nowhere else on $U_{ij}$. For $A$ near diagonal, there are two possibilities: either the
zero of $\Delta_{ij}$ remains double or it splits into a pair of simple zeroes.
We have
$\Delta_{ij}(\overline{z}) = \overline{\Delta_{ij}(z)}$, and moreover $\Delta_{ij}(z) \ge 0$ for $z \in \Real$. It follows that
$\Delta_{ij}(z)$ cannot have a simple zero on the real line, so if the double zero splits 
into a pair of simple zeroes, these two simple zeroes must be complex conjugate to one another.

Now fix a loop $c_{ij} \subset U'_{ij}$ which encircles $z_{ij}$.
The period 
\begin{equation}
p_{ij} = \oint_{c_{ij}} \sqrt{\Delta_{ij}(z)} \, dz = \I \oint_{c_{ij}} (\mu_i(z) - \mu_j(z)) \, dz
\end{equation}
can be nonvanishing only if $\Delta_{ij}(z)$ has two simple zeroes inside $c_{ij}$ (as opposed to one double zero).
The quantity $p_{ij}$ varies smoothly with $A$, and we can compute
it directly to quadratic order in the expansion
around $A^\diag$, as follows.
Applying second-order perturbation theory for the eigenvalues of $M(z)$ \cite{Kato} gives
\begin{equation}
\mu_i(z) = \mu_i^\diag(z) + \sum_{k \neq i} \frac{\abs{A_{ik}}^2}{4 \pi^2 z^2 (\mu^\diag_i(z) - \mu^\diag_k(z))} + O((A-A^\diag)^3)
\end{equation}
All terms on the right side are regular in $U_{ij}$, except for the $k = j$ term in the sum, which has a pole at $z = z_{ij}$.
Thus we obtain residues
\begin{align}
  p_{ij} = \I \oint_{c_{ij}} (\mu_i(z) - \mu_j(z)) \, dz &= 2 \I \oint_{c_{ij}} \frac{\abs{A_{ij}}^2}{4 \pi^2 z^2 (\mu^\diag_i(z) - \mu^\diag_j(z))} \, dz + O((A - A^\diag)^3) \\
  &= 2 \I \frac{\abs{A_{ij}}^2}{t_i - t_j} + O((A - A^\diag)^3) \, .
\end{align}

For $A$ close enough to diagonal, with all off-diagonal
entries $A_{ij} \neq 0$, we conclude that all $p_{ij} \neq 0$. Thus all the double zeroes are
split. Once we know this, the remaining assertions are straightforward.
\end{proof}

\subsubsection{Distinguished cycles}

We now discuss the vanishing cycles $V^{(k)}_j$ in the nearly diagonal case.
We draw $n$ concentric loops $V^{(k)}$, oriented counterclockwise,
such that the annulus between $V^{(k)}$ and $V^{(k-1)}$ contains
the branch cuts $B_{jk}$ for $j=1, \dots, k$. The preimage of this annulus
in $\Gamma(u,A)$ consists of one $2k$-holed sphere and $n-k$ cylinders. See \autoref{fig:annuli}.
\begin{figure}[h]
\centering
\begin{tikzpicture}[withbackgroundrectangle]
    \foreach \x/\label/\transpo in {-2.5/z_{4,3}/34,-1.8/z_{4,2}/24,-1.6/z_{3,2}/23,-1.1/z_{4,1}/14,-0.9/z_{3,1}/13,-0.7/z_{2,1}/12} {
        \draw[orange, thick, mark=x, mark options={orange, line width=1pt}] plot coordinates {(\x,-0.3)};
        \draw[orange, thick, mark=x, mark options={orange, line width=1pt}] plot coordinates {(\x,0.3)};
        \draw[branchcut] (\x,-0.3) -- (\x,0.3);
    }
    \draw[singularpoint] plot coordinates {(0,0)};
    \draw[witharrow=0.1,path,brown] (0,0) circle (0.5); \node[right,brown,xshift=-3] at (0.5,0) {$V^{(4)}$};
    \draw[witharrow=0.1,path,brown] (0,0) circle (1.3); \node[right,brown,xshift=-3] at (1.3,0) {$V^{(3)}$};
    \draw[witharrow=0.1,path,brown] (0,0) circle (2.1); \node[right,brown,xshift=-3] at (2.1,0) {$V^{(2)}$};
    \draw[witharrow=0.1,path,brown] (0,0) circle (2.9); \node[right,brown,xshift=-3] at (2.9,0) {$V^{(1)}$};
\end{tikzpicture}

\caption{Dividing the base $\Comp$ into annuli, each containing one group of branch cuts. We show the example $n=4$.} \label{fig:annuli}
\end{figure}
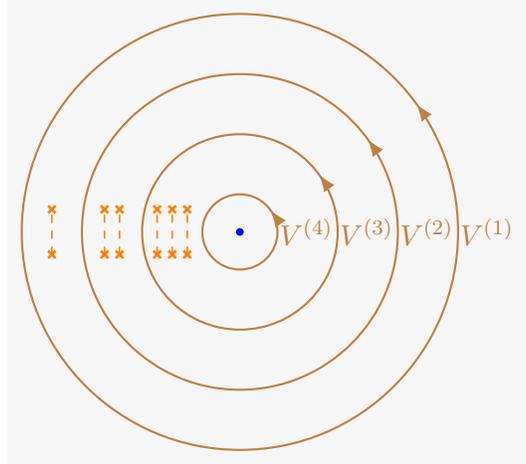

\begin{pro}
The vanishing cycle $V_j^{(k)}$ is the lift of $V^{(k)}$ to the $j^{\mathrm{th}}$
sheet of $\Gamma(u,A)$. 
\end{pro}

\begin{proof}
We consider the degeneration process from the previous sections, applied
to the $n$-fold coverings described above.
Taking $u_n \to \infty$ sends all of the 
$z_{in} \to 0$; said otherwise, in that limit, 
the group of branch cuts closest to $z=0$ disappears,
and so $V^{(n-1)}$ can be contracted to a small loop around $z=0$.
This continues to be true after blowing up
as in \autoref{prop:1-par-degen}: 
in the fiber over $u_n = \infty$, the lift of $V^{(n-1)}$
to sheet $j$ can be contracted to the puncture over $z = 0$ on that sheet.
This matches the description of $V^{(n-1)}_j$ given above. Iterating the
degeneration process we get the same statement for the other $V^{(k)}_j$. 
\end{proof}

\section{WKB predictions via spectral networks} \label{sec:spectral-networks}

In this section, we consider \autoref{mainconj} from another point of view,
as an example of the general program of \emph{exact WKB}. This program has a 
long history; see e.g. the pioneering works \cite{Vo,Silverstone,DDP},
and \cite{BNR,AHKKNSST,AKT2} which discuss the new features of the higher-order case
($n>2$ in our notation). Recently exact WKB has been reinterpreted in more geometric
language, and also connected to various other areas of mathematics and physics, 
e.g. \cite{GMN2,Gaiotto,NRS,HN,IKKS,KT,GGM}. We quickly review here the features which will be
important for us.

Suppose given a pencil of flat connections on a Riemann surface $C$, of the
form 
\begin{equation}
\nabla(\eps) = \eps^{-1} \varphi + D + \cdots 
\end{equation}
where $\varphi$ is a Higgs field on $C$ and $D$ is a connection.
Then the exact WKB program predicts the following picture. The monodromy
and/or Stokes data of $\nabla(\eps)$ can be expressed in terms of 
distinguished ``spectral coordinates'' $X_\gamma$ on the moduli space of
flat connections, and these distinguished coordinates
have exponential asymptotics of the form
\begin{equation} \label{eq:loose-sn-prediction}
    X_\gamma(\eps) \sim \exp(Z(\gamma) / \eps + \I \phi_\gamma) \, .
\end{equation}
The leading coefficients $Z(\gamma)$ 
appearing in \eqref{eq:loose-sn-prediction} are period integrals on cycles $\gamma$ in the spectral curve 
determined by $\varphi$.
The mapping between cycles $\gamma$ and coordinate functions $X_\gamma$
can be given by a general algorithm described in \cite{GMN2,HN}, involving the technology of spectral networks.

The formula \eqref{eq:loose-sn-prediction} 
evidently resembles \autoref{mainconj}. In this section we 
make that resemblance precise, by applying
the method of \cite{GMN2,HN} to the
particular family of ODEs \eqref{eq:intro-main-ode}, when $u$ is near
the caterpillar point and the matrix $A$ is close to diagonal.
We find that the $X_\gamma$'s are not exactly equal to the minor coordinates
$\Delta^{(k)}_i$
which appear in \autoref{mainconj}. Nevertheless, at least for $n=2$ and $n=3$,
we show that \eqref{eq:loose-sn-prediction} indeed leads to
predicted asymptotics for $\abs{\Delta^{(k)}_i}$
in the form of \autoref{mainconj}, and 
gives a rule for determining the cycles $L^{(k)}_i$
appearing there. As expected, we find that $L^{(k)}_i$ matches with the cycle $C^{(k)}_i$
we defined in \eqref{eq:cycles_C^k_i}.
(The obstacle to extending these statements to larger $n$ has to do with the complexity of spectral networks in this case; see \autoref{conj:combi} below.)

In fact, exact WKB leads to an extension of \autoref{mainconj}: it predicts asymptotics for
the full $\Delta^{(k)}_i$, including its phase. To get the full asymptotics, though,
we need to consider periods along open paths as well as closed ones.
We formulate these predictions for $n=2$ and $n=3$ below, and verify them directly 
in the case $n=2$.

\subsection{Periods and inequalities}

Recall the definition of periods,
\begin{defi}
For $\gamma \in H_1(\Gamma(u,A))$, the period is
\begin{equation}
  Z({\gamma}) = \oint_{\gamma} \omega .
\end{equation}
\end{defi}

\begin{defi}
Let
\begin{equation}
  \xi^{(j)}_i = -Z \left({V_i^{(j)}} \right) \, .
\end{equation}
\end{defi}

The symmetry \eqref{eq:omega-symmetry} implies that 
$\xi^{(j)}_i \in \Real$.
Also, $\xi^{(n)}_i$ and $\xi^{(1)}_i$ are residues of $\omega$ over $z = 0$ and $z = \infty$
respectively,
which gives simple formulas:
\begin{equation}
 \xi^{(n)}_i = \lambda^{(n)}_i = \lambda_i, \qquad \xi^{(1)}_i = t_i \, .  
\end{equation}
The intermediate $\xi^{(j)}_i$ for $j = 2, \dots, n-1$ 
are not in general expressed in terms of eigenvalues, but as $u$ approaches the caterpillar line, things are simpler:
\begin{pro}
As $u$ approaches the caterpillar line,
\begin{equation}
\xi_i^{(j)} \to \begin{cases} \lambda_i^{(j)} & \text{ for } i \le j \, , \\ t_i & \text{ for } i > j \, . \end{cases}
\end{equation}
\end{pro}

In particular, at the caterpillar line the $\xi_i^{(j)}$ for $i \le j$
obey the Cauchy interlacing inequalities. It follows by continuity that
they also obey these inequalities sufficiently close to the caterpillar line.
One spin-off of our description of the spectral networks is
a direct geometric reproof of these inequalities, for $n=2$ or $n=3$. 
This involves the following homology classes:
\begin{defi} \label{def:standard-hypers}
 Let
\begin{equation} \label{eq:cycle-homology-rel}
  {h_{ij}} = V^{(j-1)}_i - V^{(j)}_i , \qquad {\tilde h_{ij}} = V^{(j)}_{i+1} - V^{(j-1)}_{i}.
\end{equation}
\end{defi}
The class ${h_{ij}} \in H_1(\Gamma(u,A))$ is represented by the 1-cycle which lies on sheet $i$ and goes counterclockwise around the cut $B_{ij}$, $i < j$.

Again using the symmetry \eqref{eq:omega-symmetry}, we have $Z({{h_{ij}}}), Z({{\tilde h_{ij}}}) \in \Real$.
Below we will use spectral networks to prove the following:
\begin{thm} \label{prop:interlacing-positivity}
For $n = 2$ or $n = 3$,
for all $1 \le i < j \le n$,
$Z({{h_{ij}}}) < 0$ and $Z({{\tilde h_{ij}}}) < 0$.
\end{thm}
Then we obtain as a consequence:
\begin{cor}
For $u$ sufficiently close to the caterpillar line, and $A$ sufficiently close to diagonal, the $\xi^{(j)}_{i}$ for $i \le j$ obey the (strict) interlacing inequalities.
\end{cor}

See \autoref{fig:interlacing-gl3} for the $n=3$ example.

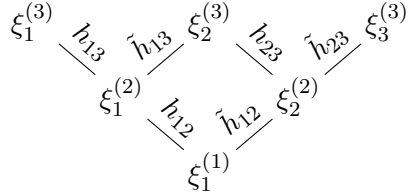
\begin{figure}[h]
\centering
\begin{tikzpicture}[node distance=1.5cm]
    \node (mu31) at (0,0) {$\xi^{(3)}_1$};
    \node[right=of mu31] (mu32) {$\xi^{(3)}_2$};
    \node[right=of mu32] (mu33) {$\xi^{(3)}_3$};

    \node (mu21) at ($(mu31)!0.5!(mu32) - (0,1)$) {$\xi^{(2)}_1$}; 
    \node (mu22) at ($(mu32)!0.5!(mu33) - (0,1)$) {$\xi^{(2)}_2$}; 

    \node (mu11) at ($(mu21)!0.5!(mu22) - (0,1)$) {$\xi^{(1)}_1$}; 

    \draw[shorten >=-3pt,shorten <=-3pt] (mu31) -- node[midway,above,sloped]{${h_{13}}$} (mu21);
    \draw[shorten >=-3pt,shorten <=-3pt] (mu21) -- node[midway,above,sloped]{${\tilde{h}_{13}}$} (mu32);
    \draw[shorten >=-3pt,shorten <=-3pt] (mu32) -- node[midway,above,sloped]{${h_{23}}$} (mu22);
    \draw[shorten >=-3pt,shorten <=-3pt] (mu22) -- node[midway,above,sloped]{${\tilde{h}_{23}}$} (mu33);
    \draw[shorten >=-3pt,shorten <=-3pt] (mu21) -- node[midway,above,sloped]{${h_{12}}$} (mu11);
    \draw[shorten >=-3pt,shorten <=-3pt] (mu11) -- node[midway,above,sloped]{${\tilde{h}_{12}}$} (mu22);
\end{tikzpicture}

\caption{The coordinates $\xi_i^{(j)}$ which enter the interlacing inequalities in the case $n=3$.
The larger coordinates are to the right, e.g. $\xi_1^{(3)} < \xi_1^{(2)}$.
For each pair of coordinates that are involved in an interlacing inequality, 
we indicate the corresponding $h$, e.g. $Z({{h_{13}}}) = - \xi_1^{(2)} + \xi_1^{(3)}$.} \label{fig:interlacing-gl3}
\end{figure}

\subsection{Spectral networks for near-diagonal matrices and Stokes asymptotics}

\subsubsection{General setup}

The conjectures of \cite{GMN,GMN2} give a description of the
$\eps \to 0$ asymptotics of the Stokes matrices $(S_\pm)(u,A,\eps)$
along any ray in the $\eps$-plane.
To write this description explicitly one needs a technical device,
the \textit{spectral network} $\cW(u,A,\vartheta)$, where
$\vartheta = \arg \eps$. We briefly review the definition of $\cW(u,A,\vartheta)$
in \autoref{app:sn-review}.

We also need to introduce a relative homology group $H(u,A,\vartheta)$
defined as follows. 
\begin{defi}
Consider the real oriented blow-up $\widehat\Gamma(u,A)$ of $\Gamma(u,A)$ at 
the $n$ preimages of $z = \infty$.
Then $\widehat\Gamma(u,A)$ has boundary consisting of $n$ circles.
On each of these $n$ circles we mark two points $\infty_i^{+}$, $\infty_i^{-}$, lying over the two anti-Stokes rays at 
arguments $\vartheta$, $\vartheta + \pi$.
Let $L(\vartheta) = \{ \infty_1^+, \infty_1^-, \infty_2^+, \infty_2^-, \dots, \infty_n^+, \infty_n^- \}$, and define
\begin{equation}
  H(u,A,\vartheta) = H_1(\widehat\Gamma(u,A) ; L(\vartheta)) \, .
\end{equation}
\end{defi}

$H(u,A,\vartheta)$ contains the subgroup $H_1(\Gamma(u,A))$ of closed paths, 
as a sublattice of index $2n-1$.

There are some subtle minus signs in the spectral coordinates, which need
to be incorporated into our bookkeeping. For this purpose, it is convenient to pass to a $\bbZ_2$ extension of $H(u,A,\vartheta)$, as follows. (These signs are probably best ignored at first; we emphasize that they are not related to the signs of periods or interlacing inequalities.)

\begin{defi} 
Let $\widetilde\Gamma(u,A)$ be the surface obtained
by puncturing $\widehat\Gamma(u,A)$ at each branch point. Then
\begin{equation}
 \widetilde H(u,A,\vartheta) = (H_1(\widetilde\Gamma(u,A) ; L(\vartheta)) \times \bbZ_2) / I \, ,
\end{equation}
where $I$ is the subgroup generated by elements $(\ell,1)$ with $\ell$ a small loop around a branch point. 
\end{defi}

The projection
$\widetilde H(u,A,\vartheta) \to H(u,A,\vartheta)$ 
given by $(\gamma,y) \mapsto \gamma$ is the $\bbZ_2$ extension.

\begin{pro} \ 
\begin{enumerate}
\item The network $\cW(u,A,\vartheta)$ determines
\textit{spectral coordinates} $X^\vartheta_\gamma(u,A,\eps) \in \Comp^\times$, indexed by 
$\gamma \in \widetilde H(u,A,\vartheta)$.
They are
algebraic functions of the Stokes matrix entries $(S_\pm)_{ij}(u,A,\eps)$,
obeying the relations $X^\vartheta_\gamma X^\vartheta_{\gamma'} = X^\vartheta_{\gamma + \gamma'}$
(and thus $X^\vartheta_{(0,0)} = 1$), and $X^\vartheta_{(0,1)} = -1$.

\item Each Stokes matrix entry can be expressed as a linear
combination of spectral coordinates, of the general form
\begin{equation} \label{eq:stokes-matrix-expansion}
  (S_\pm)_{ij}(u,A,\eps) = \sum_{\gamma \in \tilde H(u,A,\vartheta)} c^{\vartheta,\pm}_{\gamma,ij}(u,A) X^\vartheta_\gamma(u,A,\eps)
\end{equation}
where the coefficients $c^{\vartheta,\pm}_{\gamma,ij}(u,A) \in \bbQ$, and the sum may be infinite. 
\end{enumerate}

\end{pro}
The $c^{\vartheta,\pm}_{\gamma,ij}(u,A)$ are determined by the path-lifting rule,
which we review in \autoref{app:sn-review}.
Now we recall the main conjecture about WKB asymptotics and spectral networks, specialized to our case
(see \autoref{app:sn-review} for a more general discussion):
\begin{conj} \label{conj:wkb-and-sn}
As $\eps \to 0$ with $\arg \eps \in (\vartheta-\frac{\pi}{2}, \vartheta+\frac{\pi}{2})$, we have the asymptotics
\begin{equation} \label{eq:asymptotic-prediction}
  X^\vartheta_\gamma(u,A,\eps) \sim \exp(Z(\gamma)/\eps + \I \phi_\gamma) \, ,
\end{equation}
where
\begin{equation} \label{eq:period-def}
  Z(\gamma) = \int^\reg_\gamma \omega \, ,
\end{equation}
and $\phi_\gamma$ are some constants depending on $(u,A)$.
\end{conj}
If $\gamma$ is closed, then the symbol $\int^\reg$ in \eqref{eq:period-def} 
means the ordinary integral.
If $\gamma$ is open, then the ordinary integral would be divergent, and $\int^\reg$ instead means a
regularized integral, as follows.
Suppose $\gamma$ is a path running from $\infty_{i}^\eps$ to $\infty_{i'}^{\eps'}$.
We can define a regularization $\gamma_\reg$ of $\gamma$, by perturbing the endpoints slightly
off the boundary circle; call the perturbed initial and final points $z_i$, $z'_{i'}$ respectively.
Then
\begin{equation} \label{eq:open-period-def}
  Z(\gamma) = \lim_{z_i \to \infty_i^\eps} \lim_{z'_{i'} \to \infty_{i'}^{\eps'}} \left( - P(z'_{i'}) + P(z_i) + \int_{\gamma_\reg} \omega \right)
\end{equation}
where we define the function $P$ in a neighborhood of the boundary as follows:
let $z_i$ be the preimage of $z \in \Comp$ on sheet $i$;
then
\begin{equation}
P(z_i) = \I u_i z - \frac{t_i}{2 \pi \I} \left( \half \log (z) + \half \log (-z) + \I \pi \right) \, ,
\end{equation}
where $\log$ on the right side means the branch we fixed in \autoref{beginsection},
with its cut along the positive imaginary axis.

Loosely speaking, \eqref{eq:asymptotic-prediction} says that the spectral coordinates
$X^\vartheta_\gamma$ are especially well adapted for studying
the asymptotics of the Stokes matrices as $\eps \to 0$ 
with $\arg \eps \in (\vartheta-\frac{\pi}{2}, \vartheta+\frac{\pi}{2})$.
For instance,
using \eqref{eq:stokes-matrix-expansion} and \eqref{eq:asymptotic-prediction} we can 
predict the leading asymptotics of $(S_\pm)_{ij}$.
In the expansion \eqref{eq:stokes-matrix-expansion}, the dominant term will be the one with the largest value of $\re Z(\gamma) / \eps$ (assuming that there is a unique largest value),
and then we predict
\begin{equation} \label{eq:stokes-asymp-prediction}
  (S_\pm)_{ij}(u,A,\eps) \sim c^{\vartheta,\pm}_{\gamma,ij} \exp(Z(\gamma) / \eps + \I \phi_\gamma) \, .
\end{equation}

\subsubsection{The case \texorpdfstring{$n=2$}{n=2}}

First, we discuss the case $n=2$. This case is exceptional in that we can calculate
the Stokes matrix exactly, so we can directly check that the predictions of exact WKB 
indeed hold.

Write
\begin{equation}
  A = \begin{pmatrix} t_1 & a \\ \overline{a} & t_2 \end{pmatrix} \, .
\end{equation}
The branch points of the spectral curve $\Gamma(u,A)$ are
\begin{equation} \label{eq:branch-points-explicit}
  z^\pm_{12} = \frac{t_1 - t_2 \pm 2 \I \abs{a}}{2\pi(u_2 - u_1)} \, .
\end{equation}

\begin{figure}[h]
\centering
\begin{tikzpicture}[>={Latex},withbackgroundrectangle]

    \node[anchor=south west,inner sep=0] (image) at (0,0) {\includegraphics[width=0.26\textwidth]{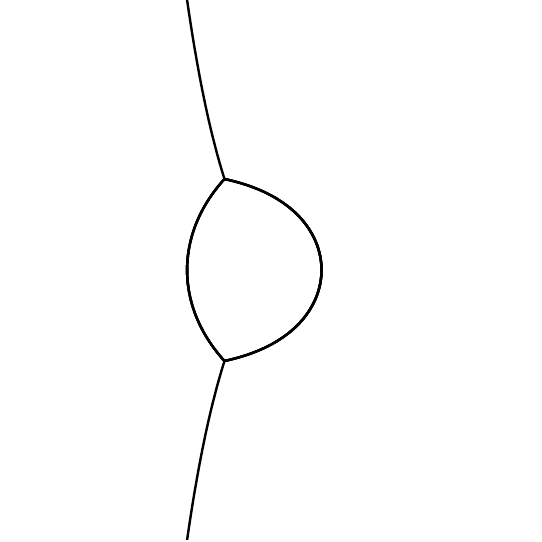}};
    \begin{scope}[x={(image.south east)},y={(image.north west)}]

        \coordinate (bpA) at (0.415,0.33);
        \coordinate (bpB) at (0.415,1-0.33);
        \coordinate (sing) at (0.5,0.5);
        \coordinate (L) at (0,0.5);
        \coordinate (R) at (1,0.5);

        \node[stokeslabel] at (0.34,0-0.04) {$1<2$};
        \node[stokeslabel] at (0.34,1+0.04) {$2<1$};

        \draw[branchpoint] plot coordinates {(bpA)};
        \draw[branchpoint] plot coordinates {(bpB)};
        \draw[branchcut] (bpA) -- (bpB);

        \draw[singularpoint] plot coordinates {(sing)};

    \end{scope}
\end{tikzpicture}
\hspace{0.05cm}
\begin{tikzpicture}[>={Latex},withbackgroundrectangle]

    \node[anchor=south west,inner sep=0] (image) at (0,0) {\includegraphics[width=0.26\textwidth]{sn-gl2-1.pdf}};
    \begin{scope}[x={(image.south east)},y={(image.north west)}]

        \coordinate (bpA) at (0.415,0.33);
        \coordinate (bpB) at (0.415,1-0.33);
        \coordinate (sing) at (0.5,0.5);
        \coordinate (L) at (0,0.5);
        \coordinate (R) at (1,0.5);

        \draw[antistokesmark] ($(L) - (0.01,0)$) circle;
        \draw[antistokesmark] ($(R) + (0.01,0)$) circle;

        \node[stokeslabel] at (0.34,0-0.04) {$1<2$};
        \node[stokeslabel] at (0.34,1+0.04) {$2<1$};

        \draw[branchpoint] plot coordinates {(bpA)};
        \draw[branchpoint] plot coordinates {(bpB)};
        \draw[branchcut] (bpA) -- (bpB);

        \draw[singularpoint] plot coordinates {(sing)};

        \draw[->,thick,blue!85]
        ($(R) + (-0.01,-0.01)$) 
        node[below,sheetlabel,xshift=-20pt] {2} 
        .. controls (0.5,0.45) .. 
        ($(L) + (0.01,-0.01)$)
        node[below,midway,xshift=23pt,yshift=3pt] {${a_{21}}$} 
        node[below,sheetlabel,xshift=20pt] {1}; 

        \draw[->,thick,blue!50] 
        ($(L) + (0.01,0.01)$) 
        node[above,sheetlabel,xshift=20pt] {1} 
        .. controls (0.5,0.55) .. 
        ($(R) + (-0.01,0.01)$)
        node[above,midway,xshift=23pt,yshift=-3pt] {${a_{12}}$} 
        node[above,sheetlabel,xshift=-20pt] {2}; 

        \draw[->,thick,red!60] 
        ($(R) + (-0.01,-0.02)$) 
        node[below,sheetlabel,xshift=-5pt,yshift=-3pt] {$i$} 
        .. controls (0.5,0) .. 
        ($(L) + (0.01,-0.02)$)
        node[below,midway] {${c_i}$} 
        node[below,sheetlabel,xshift=5pt,yshift=-3pt] {$i$}; 

        \draw[->,thick,red!60] 
        ($(L) + (0.01,0.02)$) 
        node[above,sheetlabel,xshift=5pt,yshift=3pt] {$i$} 
        .. controls (0.5,1) .. 
        ($(R) + (-0.01,0.02)$)
        node[above,midway] {${d_i}$} 
        node[above,sheetlabel,xshift=-5pt,yshift=3pt] {$i$}; 

    \end{scope}
\end{tikzpicture}
\hspace{0.05cm}
\begin{tikzpicture}[>={Latex},withbackgroundrectangle]

    \node[anchor=south west,inner sep=0] (image) at (0,0) {\includegraphics[width=0.26\textwidth]{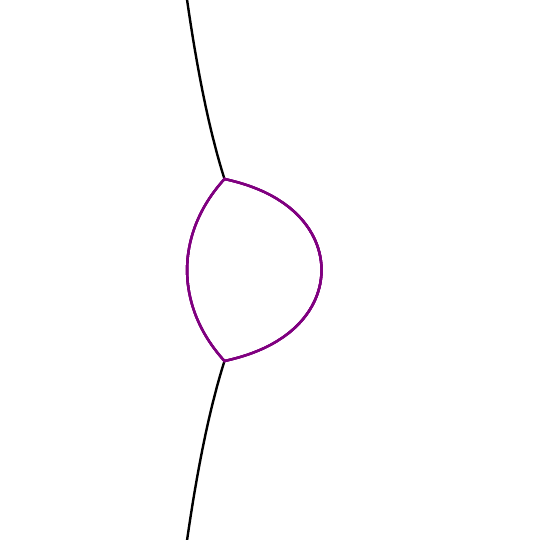}};
    \begin{scope}[x={(image.south east)},y={(image.north west)}]

        \coordinate (bpA) at (0.415,0.33);
        \coordinate (bpB) at (0.415,1-0.33);
        \coordinate (sing) at (0.5,0.5);
        \coordinate (L) at (0,0.5);
        \coordinate (R) at (1,0.5);

        \node[stokeslabel] at (0.34,0-0.04) {$1<2$};
        \node[stokeslabel] at (0.34,1+0.04) {$2<1$};

        \draw[branchpoint] plot coordinates {(bpA)};
        \draw[branchpoint] plot coordinates {(bpB)};
        \draw[branchcut] (bpA) -- (bpB);

        \draw[singularpoint] plot coordinates {(sing)};

        \draw[witharrow=0.9,thick,brown!80,rounded corners=4] 
        ($(bpA) + (-0.03,-0.03)$) -- node[below,xshift=-16pt,yshift=17pt] {1} ($(bpA) + (0.05,-0.03)$) .. controls (0.34,0.5) .. ($(bpB) + (0.05,0.03)$) -- ($(bpB) + (-0.03,0.03)$) .. controls (0.3,0.5) .. node[left,midway,yshift=15pt,xshift=6pt] {${h_{12}}$} ($(bpA) + (-0.03,-0.03)$);

        \draw[witharrow=0.9,thick,darkgreen!70,rounded corners=4] 
        ($(bpA) + (-0.05,-0.02)$) -- node[below,xshift=25pt,yshift=15pt] {2} ($(bpA) + (0.02,-0.02)$) .. controls (0.75,0.5) .. ($(bpB) + (0.02,0.02)$) -- ($(bpB) + (-0.05,0.02)$) .. controls (0.65,0.5) .. node[left,midway,yshift=21pt,xshift=18pt] {${\tilde h_{12}}$} ($(bpA) + (-0.05,-0.02)$);

    \end{scope}
\end{tikzpicture}

\caption{A spectral network $\cW(u,A,\vartheta = 0)$ in the $n=2$ case. Left: the bare spectral network. Middle: the spectral network with paths representing classes in $\tilde H(u,A,\vartheta)$. Numbers $1$, $2$ next to paths indicate the sheet on which the path travels. Right: the spectral network with the saddle connections highlighted in purple, and representatives of their charges in 
$\tilde H(u,A,\vartheta)$ marked.} \label{fig:sn-gl2-1}
\end{figure}
From \eqref{eq:branch-points-explicit} we see directly that,
when $t_1 < t_2$ and $a \neq 0$, the two branch points are simple. It follows that
 $\Gamma(u,A)$ is smooth of genus $g = 0$ with $4$ punctures. The lattice $H(u,A,\vartheta)$ has rank $6$,
generated by the $6$ open paths ${a_{12}}$, ${a_{21}}$, ${c_i}$, ${d_i}$ ($i=1,2$) shown in \autoref{fig:sn-gl2-1}.
The sublattice $H_1(\Gamma(u,A))$ of closed paths has rank $3$;
it is spanned by $c_1 + d_1$, $c_2 + d_2$, and $a_{12} + a_{21}$.
On this sublattice, all periods are real, and the intersection pairing is trivial.

\begin{pro} For $u_1 < u_2$, $t_1 < t_2$ and $a \neq 0$, 
the spectral network $\cW(u,A,\vartheta = 0)$ has the topology shown in
\autoref{fig:sn-gl2-1}.
\end{pro}

\begin{proof} $\cW(u,A,\vartheta = 0)$ is the
critical graph of the quadratic differential
\begin{equation}
 \phi_2 = \frac{(t_1 - t_2 + 2 \pi(u_1 - u_2)z)^2 + 4 \abs{a}^2}{16 \pi^2 z^2} \, dz^2 \, .
\end{equation}
We can determine its topology with computer assistance, or more 
synthetically, using the following general constraints.
Each branch point emits three trajectories, which must end up either at a
Stokes ray or at the other branch point. Moreover, no two trajectories can cross, 
the picture has
to be symmetric under reflection in the real axis, and no two 
trajectories can be homotopic to one another.
This information is enough to determine the picture uniquely.
\end{proof}

$\cW(u,A,\vartheta = 0)$ is degenerate: there are two saddle connections connecting the two branch points. 
Their charges are the classes
${h_{12}}$ and ${\tilde h_{12}}$ defined in \autoref{def:standard-hypers},
and shown in the right of \autoref{fig:sn-gl2-1}. 
They can be written alternatively as 
\begin{equation}
    {h_{12}} = {c_2} + {d_2} - {a_{12}} - {a_{21}}, \quad {\tilde h_{12}} = {c_1} + {d_1} - {a_{12}} - {a_{21}} \, .
\end{equation}
The fact that saddle connections with charges ${h_{12}}$ and
${\tilde h_{12}}$
exist proves \autoref{prop:interlacing-positivity} in this case,
using \autoref{prop:finite-web-period-phase}.

\subsubsection*{Stokes matrix asymptotics}

Now we discuss the $\eps \to 0$ 
asymptotics of the Stokes matrix entries.
Since $S_- = S_+^*$ for real $\eps$,
we will just discuss $S_+$.
The first step is to expand the entries of $S_+$ in terms of
spectral coordinates:

\begin{pro} The coordinate expansion \eqref{eq:stokes-matrix-expansion} of $S_+$ has the form
\begin{align} 
    (S_+)_{ii} &= X^0_{{c_i}}, \label{eq:Sii-expansion-n2} \\
    (S_+)_{12} &= X^0_{{a_{21}}} + \cdots \label{eq:S12-expansion-n2}
\end{align}
where $\cdots$ is a (infinite, convergent) 
sum of terms $c_\gamma X^0_\gamma$, with each $Z(\gamma) - Z({a_{21}}) \in \Real_-$.
\end{pro}

\begin{figure}[h]
\centering
\begin{tikzpicture}[>={Latex},withbackgroundrectangle]

    \node[anchor=south west,inner sep=0] (image) at (0,0) {\includegraphics[width=0.26\textwidth]{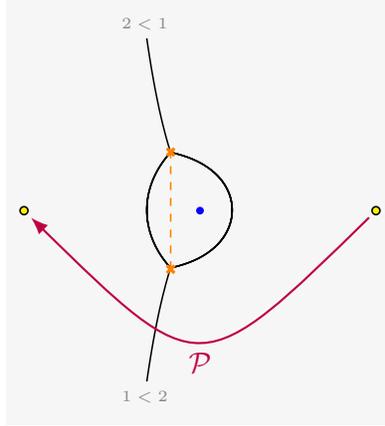}};
    \begin{scope}[x={(image.south east)},y={(image.north west)}]

        \coordinate (bpA) at (0.415,0.33);
        \coordinate (bpB) at (0.415,1-0.33);
        \coordinate (sing) at (0.5,0.5);
        \coordinate (L) at (0,0.5);
        \coordinate (R) at (1,0.5);

        \draw[antistokesmark] ($(L) - (0.01,0)$) circle;
        \draw[antistokesmark] ($(R) + (0.01,0)$) circle;

        \node[stokeslabel] at (0.34,0-0.04) {$1<2$};
        \node[stokeslabel] at (0.34,1+0.04) {$2<1$};

        \draw[branchpoint] plot coordinates {(bpA)};
        \draw[branchpoint] plot coordinates {(bpB)};
        \draw[branchcut] (bpA) -- (bpB);

        \draw[singularpoint] plot coordinates {(sing)};

        \draw[->,path,purple] 
        ($(R) + (-0.01,-0.02)$) 
        .. controls (0.5,0) .. 
        ($(L) + (0.01,-0.02)$)
        node[below,midway] {$\cP$}; 

    \end{scope}
\end{tikzpicture}

\caption{A spectral network $\cW(u,A,\vartheta = 0)$ in the $n=2$ case,
with the path $\cP$ marked.} \label{fig:sn-gl2-path}
\end{figure}

\begin{proof}
Let $\cP$ be the arc shown in \autoref{fig:sn-gl2-path}.
Let $\Lift$ be the path-lifting functor induced by $\cW(u,A,\vartheta=0)$.
The expansion of the entries of $S_+$ in spectral coordinates is given
by $\Lift(\cP)$. According to \autoref{prop:path-lifting-canonical},
we have 
\begin{equation}
\Lift(\cP) = {c_1} + {c_2} + {a_{21}} + \cdots    
\end{equation}
where all terms $\gamma$ appearing in the $\cdots$ 
go from sheet $2$ to sheet $1$, and
have $Z(\gamma) - Z({{a_{21}}}) < 0$.
The entry $(S_+)_{ij}$ is the sum of spectral coordinates for 
the terms that begin on sheet $j$ and 
end on sheet $i$; this gives the desired results.
\end{proof}

As $\eps \to 0$ with $\arg \eps \in (-\pi/2,\pi/2)$, 
the $\cdots$ terms in \eqref{eq:S12-expansion-n2} 
are exponentially suppressed; thus
the asymptotic we predict using \eqref{conj:wkb-and-sn} 
comes from the $X^0_{{a_{21}}}$ term, and is
\begin{equation} \label{eq:s12-prediction-n2}
  (S_+)_{12} \sim \exp(Z({{a_{21}}}) / \eps) \, .
\end{equation}
For the diagonal entries we just have a single term in \eqref{eq:Sii-expansion-n2},
and thus the asymptotic is
\begin{equation}
    (S_+)_{ii} \sim \exp(Z({{c_i}}) / \eps) \, .
\end{equation}
Thus we are predicting altogether
\begin{equation}
   \eps \log (S_+)_{ij}  \to \begin{pmatrix} Z({{c_1}}) & Z({{a_{21}}})  \\ 
  & Z({{c_2}}) \end{pmatrix} .    
\end{equation}
To understand this prediction more explicitly, we need to study the regularized open periods $Z(\gamma)$, which we do next.

\subsubsection*{Real parts of periods}

The real parts of the periods 
are relatively simple to obtain. Indeed, we have
\begin{equation}
\re Z(\gamma) = \frac12 Z(\gamma - \iota(\gamma)) \, ,    
\end{equation}
and $\gamma - \iota(\gamma)$ is always a closed cycle,
which we can identify directly. Moreover, since $\Gamma$ has genus zero,
the resulting closed periods are just sums of residues.
We obtain:
\begin{align}
 \re Z({{a_{21}}}) &= \frac12 (Z({{a_{21}}}) + Z({{a_{12}}})) = - \frac12 Z({V^{(2)}_2}) = - \frac12 \xi_2^{(2)} = - \frac12 \lambda_2 \, , \\
 \re Z({{c_i}}) &= \frac12 Z({{c_i} + {d_i}}) = - \frac12 Z({V^{(1)}_i}) = \frac12 \xi_i^{(1)} = \frac12 t_i \, .
\end{align}
Thus we can summarize the predicted leading asymptotic magnitudes as
\begin{equation} \label{eq:leading-asymptotics-gl2}
  \eps \log \abs{(S_+)_{ij}} \to \frac12 \begin{pmatrix} t_1 & \lambda_2  \\ & t_2 \end{pmatrix} \,.
\end{equation}
Moreover, this structure is in accordance with the general \autoref{mainconj}. Indeed, 
the minor coordinates are 
$\Delta_1^{(1)} = (S_+)_{11}$, 
$\Delta_1^{(2)} = (S_+)_{12}$, 
$\Delta_2^{(2)} = (S_+)_{11} (S_+)_{22}$,
and we have just seen that 
each entry $(S_+)_{ij}$ 
is controlled by a corresponding period; this
matches \autoref{mainconj}, with the cycles
\begin{equation}
    L_1^{(1)} = V_1^{(1)}, \quad L_1^{(2)} = V_2^{(2)}, \quad L_2^{(2)} = V_1^{(1)} + V_2^{(1)} \, .
\end{equation}
Comparing this with \eqref{eq:cycles_C^k_i}, we see that
\autoref{mainconj} holds with 
$L_i^{(k)} = C_i^{(k)}$, as we expected.

\subsubsection*{Imaginary parts of periods}

The imaginary parts of the periods are subtler than the real parts;
they require us to use the detailed formula \eqref{eq:open-period-def} defining the regularization of open paths.

The paths ${c_i}$ can be isotoped into a small
neighborhood of $z = \infty$ and computed explicitly from the
series expansion of $\omega$ there, with the result:
\begin{pro} The regularized period on the path ${c_i}$ is
\begin{equation} \label{eq:Zci}
    Z({{c_i}}) = \frac12 t_i \, .
\end{equation}
\end{pro}

For ${a_{21}}$ the computation is more interesting, since this
path cannot be isotoped to a small neighborhood of $z = \infty$. One can still compute the integral in \eqref{eq:open-period-def} directly: indeed, since $\Gamma$
has genus zero, the 1-form $\omega$ must have an 
elementary antiderivative, just involving square roots
and logarithms. Using this antiderivative in \eqref{eq:open-period-def} gives the following result:
\begin{pro} The regularized period on the path ${a_{21}}$ is
\begin{equation} \label{eq:reg-period-a21-n2}
  Z({{a_{21}}}) = \frac{\lambda_2}{2} - \frac{\I}{2 \pi} \left( (t_2-t_1) \log \frac{2 \pi \rme (u_2-u_1)}{\abs{a}} + \frac12(\lambda_2-\lambda_1) \left(\log(t_1 - \lambda_{1}) - \log(\lambda_2 - t_1)  \right)\right) \, .
\end{equation}
\end{pro}

\subsubsection*{A comparison}

Now we can confront our asymptotic predictions with ground truth; 
indeed we know the exact form of $S^+$,
given in \autoref{ex:n2-stokes-matrices}. Comparing it to
\eqref{eq:Zci} and \eqref{eq:reg-period-a21-n2} 
gives:
\begin{cor} As $\eps \to 0$ with $\arg \eps \in (-\pi/2,\pi/2)$,
we have $\eps \log S_{12}^+ \to Z({{a_{21}}})$
and $\eps \log S_{ii}^+ = Z({c_i})$.
\end{cor}
In particular, it follows that 
\autoref{conj:wkb-and-sn} is indeed
true in this case.

\subsubsection{The case \texorpdfstring{$n=3$}{n=3}}

Next, let us consider the case $n=3$.

By \autoref{prop:spectral-curve-structure},
for $u$ sufficiently close to the caterpillar line
and $A$ sufficiently close to diagonal,
the spectral curve $\Gamma(u,A)$ is smooth, with
genus $g = 1$ and $m = 6$ punctures.
It follows that 
the lattice $H(u,A,\vartheta)$ has rank $2g + 2m - 2 = 12$, and the closed
sublattice $H_1(\Gamma(u,A))$ has rank $2g + m - 1 = 7$.
Inside the closed sublattice we can consider the kernel of the intersection pairing, 
$K(u,A) \subset H_1(\Gamma(u,A))$, of rank $m - 1 = 5$.
On $K(u,A)$ all periods are sums of residues, and in particular, they are real.

We begin by describing the spectral networks:
\begin{pro} For $u_1 < u_2 < u_3$, 
$t_1 < t_2 < t_3$, $u$ sufficiently close to the caterpillar line, and $A$
sufficiently close to diagonal,
$\cW(u,A,\vartheta = 0)$ is equivalent to the spectral network shown in \autoref{fig:gl3-sn}.
\end{pro}

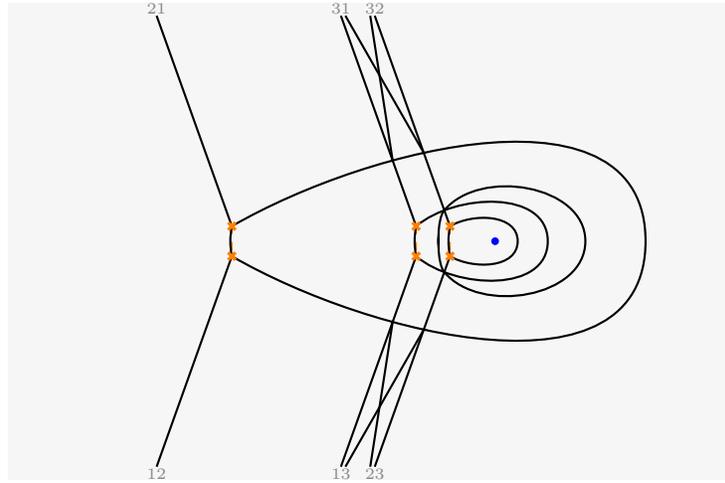
\begin{figure}[h]
\centering
\begin{tikzpicture}[withbackgroundrectangle]

 \useasboundingbox (-6.3,-3) rectangle (3,3); 

  \newcommand{\xOne}{-3.5}
  \newcommand{\xTwo}{-1.05}
  \newcommand{\xThree}{-0.6}
  \newcommand{\yBP}{0.2}
  \newcommand{\yStokes}{3}

  \coordinate (bpA) at (\xOne,\yBP);
  \coordinate (bpB) at (\xOne,-\yBP);
  \coordinate (bpC) at (\xTwo,\yBP);
  \coordinate (bpD) at (\xTwo,-\yBP);
  \coordinate (bpE) at (\xThree,\yBP);
  \coordinate (bpF) at (\xThree,-\yBP);
  \coordinate (sp) at (0,0);
  \coordinate (stokes21) at (\xOne-1,\yStokes);
  \coordinate (stokes31) at (\xTwo-1,\yStokes);
  \coordinate (stokes32) at (\xThree-1,\yStokes);
  \coordinate (stokes12) at (\xOne-1,-\yStokes);
  \coordinate (stokes13) at (\xTwo-1,-\yStokes);
  \coordinate (stokes23) at (\xThree-1,-\yStokes);

  \draw[branchcut] (bpA) -- (bpB);
  \draw[branchcut] (bpC) -- (bpD);
  \draw[branchcut] (bpE) -- (bpF);

  \draw[wall] (bpA) to[out=260, in=100] (bpB);
  \draw[wall] (bpC) to[out=260, in=100] (bpD);
  \draw[wall] (bpE) to[out=260, in=100] (bpF);

  \draw[wall,name path = out21] (bpA) to (stokes21);
  \draw[wall,name path = out31] (bpC) to (stokes31);
  \draw[wall,name path = out32] (bpE) to (stokes32);
  \draw[wall,name path = out12] (bpB) to (stokes12);
  \draw[wall,name path = out13] (bpD) to (stokes13);
  \draw[wall,name path = out23] (bpF) to (stokes23);

  \draw[wall,name path = ht12] (bpA) to[out=30,in=90] (2,0) to[out=270,in=330] (bpB);
  \draw[wall,name path = ht13] (bpC) to[out=40,in=90] (0.7,0) to[out=270,in=320] (bpD);
  \draw[wall,name path = ht23] (bpE) to[out=30,in=90] (0.3,0) to[out=270,in=330] (bpF);

  \path[name intersections={of=ht13 and out23, by={inner-int-1}}];
  \path[name intersections={of=ht13 and out32, by={inner-int-2}}];
  \path[name intersections={of=ht12 and out13, by={outer-int-2113}}];
  \path[name intersections={of=ht12 and out23, by={outer-int-1223}}];
  \path[name intersections={of=ht12 and out31, by={outer-int-1231}}];
  \path[name intersections={of=ht12 and out32, by={outer-int-2132}}];

  \draw[wall] (inner-int-2) to[out=50,in=90] (1.2,0) to[out=270,in=310] (inner-int-1);
  \draw[wall] (inner-int-2) to[out=230,in=270] (-0.75,0) to[out=90,in=130] (inner-int-1);

  \draw[wall] (outer-int-2113) to ($(stokes23) - (0.06,0)$);
  \draw[wall] (outer-int-1223) to ($(stokes13) + (0.06,0)$);
  \draw[wall] (outer-int-1231) to ($(stokes32) - (0.06,0)$);
  \draw[wall] (outer-int-2132) to ($(stokes31) + (0.06,0)$);

  \draw[branchpoint] plot coordinates {(bpA)};
  \draw[branchpoint] plot coordinates {(bpB)};
  \draw[branchpoint] plot coordinates {(bpC)};
  \draw[branchpoint] plot coordinates {(bpD)};
  \draw[branchpoint] plot coordinates {(bpE)};
  \draw[branchpoint] plot coordinates {(bpF)};

  \draw[singularpoint] plot coordinates {(sp)};

  \node[stokeslabel] at ($(stokes21) + (0,0.1)$) {$21$};
  \node[stokeslabel] at ($(stokes12) + (0,-0.1)$) {$12$};
  \node[stokeslabel] at ($(stokes32) + (0,0.1)$) {$32$};
  \node[stokeslabel] at ($(stokes23) + (0,-0.1)$) {$23$};
  \node[stokeslabel] at ($(stokes31) + (0,0.1)$) {$31$};
  \node[stokeslabel] at ($(stokes13) + (0,-0.1)$) {$13$};
\end{tikzpicture}
\caption{A spectral network $\cW(u,A,\vartheta = 0)$ in the $n=3$ case. We compress
the notation for the wall labels by writing $ij$ instead of $i<j$.}
\label{fig:gl3-sn}
\end{figure}

\begin{proof} Using the code included with the arXiv version of this paper,
we can check
that, for one specific example of $(u,A)$, the spectral network $\cW(u,A,\vartheta = 0)$ is indeed of the form shown
in \autoref{fig:gl3-sn}.\footnote{To be explicit, consider the case
\begin{equation}
u = \left(0, \frac14, 1\right)\,, \quad A = \begin{pmatrix} 0 & \frac16 & 1 \\ \frac16 & 3 & \frac13 \\ 1 & \frac13 & 4 \end{pmatrix} \, . 
\end{equation}
A computer-generated approximation of the resulting spectral network at $\vartheta = 0.0005$ is shown below.
\begin{center}
\begin{tikzpicture}[>={Latex},withbackgroundrectangle]
    \node[anchor=south west,inner sep=0] (image) at (0,0) {\includegraphics[width=0.19\textwidth]{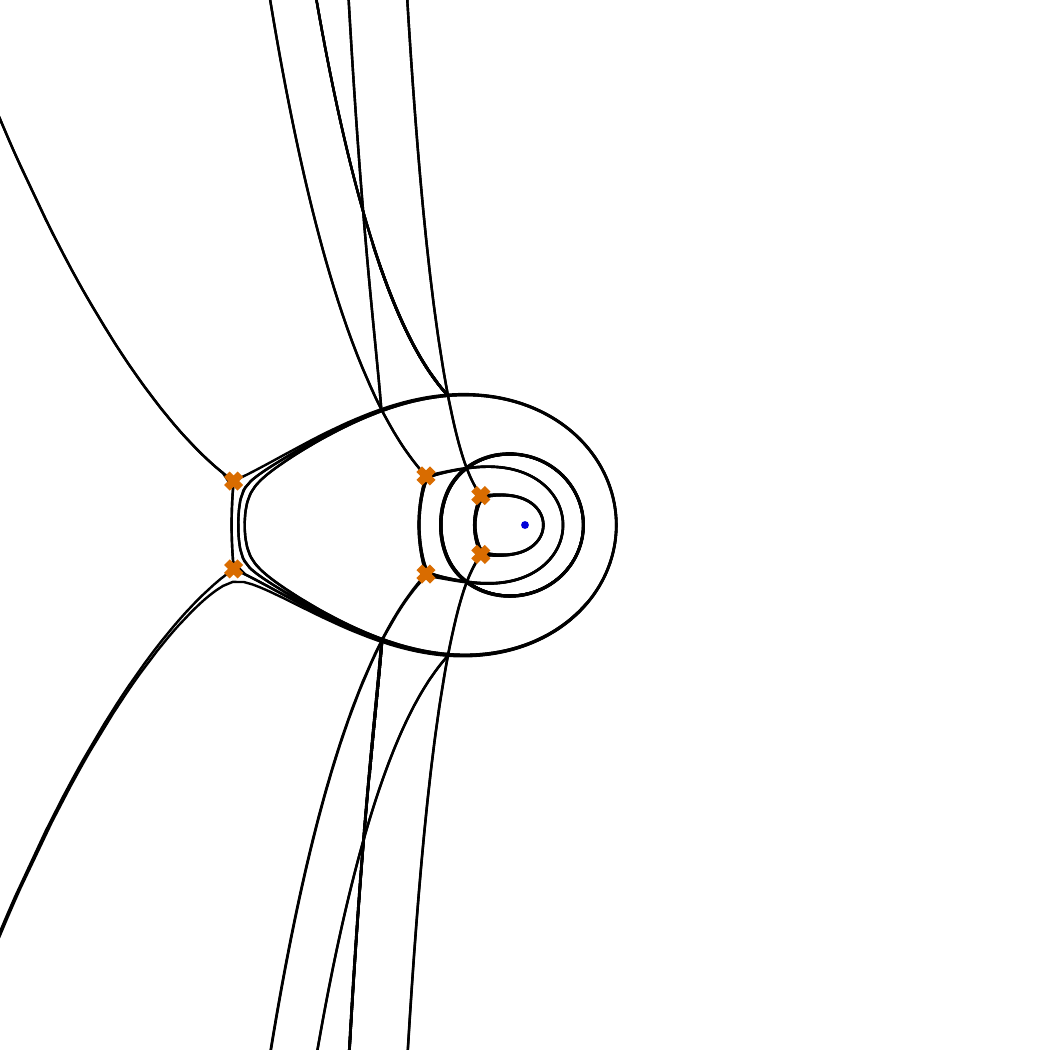}};
\end{tikzpicture}
\hspace{0.1cm}
\begin{tikzpicture}[>={Latex},withbackgroundrectangle]
    \node[anchor=south west,inner sep=0] (image) at (0,0) {\includegraphics[width=0.19\textwidth]{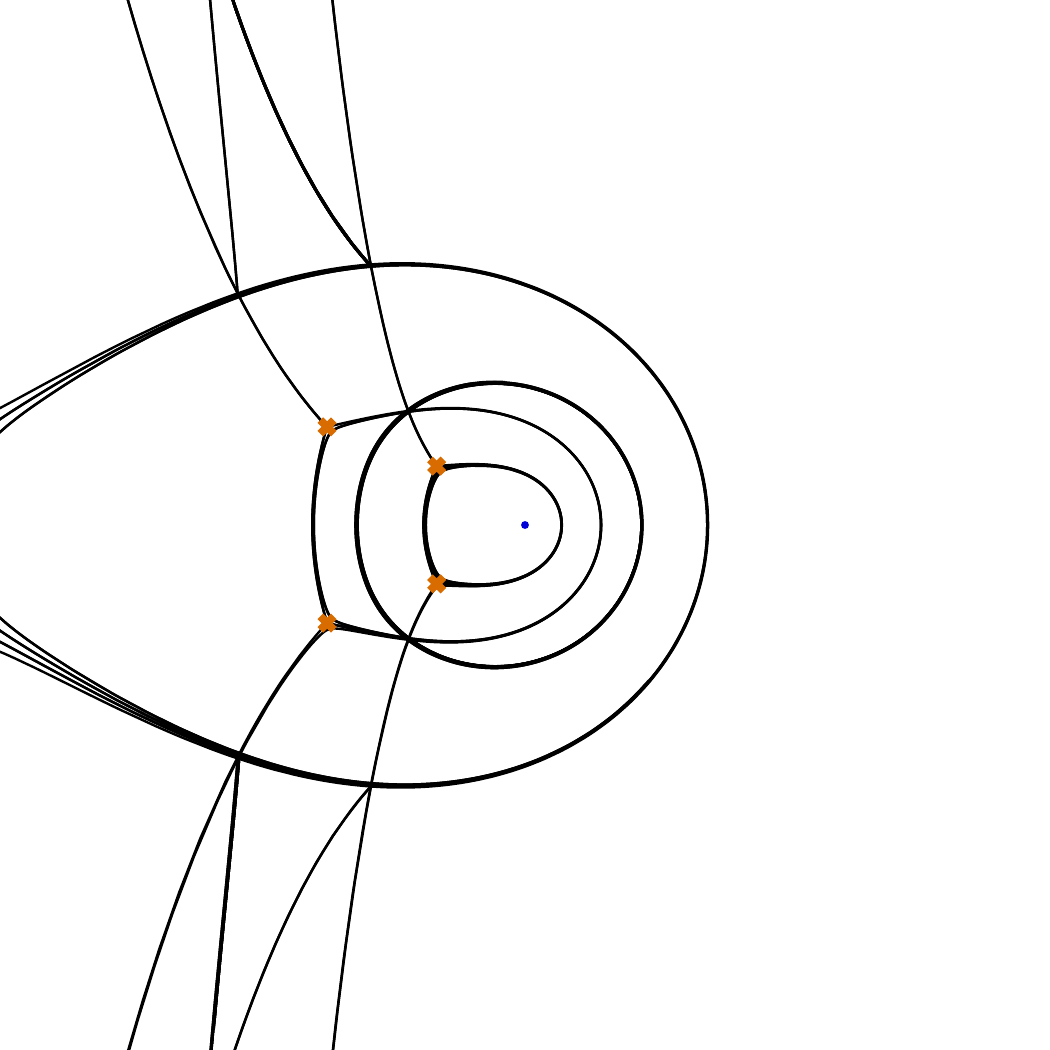}};
\end{tikzpicture}
\end{center}
}

Next, we consider what happens when we vary $(u,A)$. 
It follows from \autoref{pro:equivalences-from-variations} that
the lifting maps for the spectral network 
$\cW(u,A,\vartheta = 0)$ vary by equivalences, except
when $(u,A)$ cross 
a locus where some $\gamma \notin K(u,A)$ has real period.
Because $K(u,A)$ has index $2$, the only way
this can happen is if every period becomes real; in particular,
$Z({{\Delta}})$ would have to be real,
where $\Delta$ is shown in \autoref{fig:gamma-delta}. 

\begin{figure}[h]
\centering

\begin{tikzpicture}[withbackgroundrectangle]
    \foreach \x/\label/\transpo in {1/z_{12}/12,4/z_{13}/13,5/z_{23}/23} {
        \draw[orange, thick, mark=x, mark options={orange, line width=1pt}] plot coordinates {(\x,0)};
        \draw[orange, thick, mark=x, mark options={orange, line width=1pt}] plot coordinates {(\x,1)};
        \draw[branchcut] (\x,0) -- (\x,1);
        \node[below] at (\x,0) {$\label^-$};
        \node[above] at (\x,1) {$\label^+$};
        \node[cutlabel] at (\x,0.5) {$(\transpo)$};
    }
    \draw[singularpoint] plot coordinates {(7,0.5)};
    \node[below] at (7,0.5) {$0$};

    \draw[purple, thick] (0.75,0.11) -- (4.25,0.11) node[midway, above] {$\Delta$} -- (5.25,0.11) to[bend left=30] (5.25,-0.11) -- (0.75,-0.11) to[bend left=30] (0.75,0.11);
    \draw[purple, thick, -latex] (3,0.11) -- (2.9,0.11); 
    \node[purple,below,sheetlabel] at (3,-0.04) {2};
\end{tikzpicture}

\caption{The cycle $\Delta$ on $\Gamma(u,A)$.} \label{fig:gamma-delta}
\end{figure}
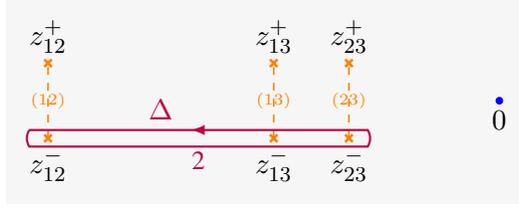

But we can compute 
$Z({\Delta})$ in the limit of diagonal $A$:
using the explicit formulas \eqref{eq:rho-diagonal}, \eqref{eq:def-zij} we get
\begin{equation}
    Z({\Delta}) \to -\frac{\I}{2 \pi} \left( (t_2 - t_1) \log \frac{(t_3-t_1)(u_2-u_1)}{(t_2-t_1)(u_3-u_1)} + (t_3 - t_2) \log \frac{(t_3-t_2)(u_3-u_1)}{(t_3-t_1)(u_3-u_2)} \right) \, .
\end{equation}
Near the caterpillar line $u_3 - u_2 \gg u_2 - u_1$, this becomes
\begin{equation}
    Z({\Delta}) \sim \frac{\I}{2 \pi} (t_2 - t_1) \log \frac{u_3-u_1}{u_2-u_1} \, ,
\end{equation}
and thus the imaginary part does not vanish for $u$ sufficiently close to the caterpillar line and $A$ sufficiently close to diagonal.
It follows that, as we vary $(u,A)$ while maintaining this condition,
$\cW(u,A,\vartheta=0)$ varies by equivalences. 
\end{proof}

Let us discuss the spectral network in \autoref{fig:gl3-sn}.
Looking first at the region near $z = \infty$,
we see a structure which resembles the $n=2$ case in \autoref{fig:sn-gl2-1}:
two branch points, connected by two saddle connections, and emitting two walls to
infinity, labeled $1<2$ and $2<1$.
One of the two saddle connections looks like a large ring, almost encircling the singularity
at $z = 0$.

Near $z=0$ there is some more complicated fine structure,
with two walls of types $1<3$ and $2<3$ emitted out to the large-$z$ region.
Along the way from small $z$ to large $z$, 
these walls scatter off the ring, creating extra walls of
type $1<3$ and $2<3$. Thus
the Stokes matrix entries $(S_+)_{13}$ and $(S_+)_{23}$ 
ultimately involve a mixture between the structure
of the small-$z$ and large-$z$ regions.
In contrast, $(S_+)_{12}$ comes only from the
large-$z$ region.

Another important feature of $\cW(u,A,\vartheta=0)$ is that it contains finite webs. In
\autoref{fig:sn-gl3-3} we indicate $6$ finite webs whose
charges are the classes ${h_{ij}}$ and ${\tilde h_{ij}}$.
$5$ of these $6$ finite webs are saddle connections, similar to the ones
that we had in the $n=2$ case. The last one, with charge ${\tilde h_{13}}$,
is a five-string tree.
The existence of these finite webs proves \autoref{prop:interlacing-positivity} 
in case $n=3$, using \autoref{prop:finite-web-period-phase}.

\begin{figure}
\centering

\begin{tikzpicture}[withbackgroundrectangle,scale=1.4,every node/.style={scale=1}]

  \clip (-5.3,-2) rectangle (3,2);

  \newcommand{\xOne}{-3.5}
  \newcommand{\xTwo}{-1.05}
  \newcommand{\xThree}{-0.6}
  \newcommand{\yBP}{0.2}
  \newcommand{\yStokes}{3}

  \coordinate (bpA) at (\xOne,\yBP);
  \coordinate (bpB) at (\xOne,-\yBP);
  \coordinate (bpC) at (\xTwo,\yBP);
  \coordinate (bpD) at (\xTwo,-\yBP);
  \coordinate (bpE) at (\xThree,\yBP);
  \coordinate (bpF) at (\xThree,-\yBP);
  \coordinate (sp) at (0,0);
  \coordinate (stokes21) at (\xOne-1,\yStokes);
  \coordinate (stokes31) at (\xTwo-1,\yStokes);
  \coordinate (stokes32) at (\xThree-1,\yStokes);
  \coordinate (stokes12) at (\xOne-1,-\yStokes);
  \coordinate (stokes13) at (\xTwo-1,-\yStokes);
  \coordinate (stokes23) at (\xThree-1,-\yStokes);

  \draw[branchcut] (bpA) -- (bpB);
  \draw[branchcut] (bpC) -- (bpD);
  \draw[branchcut] (bpE) -- (bpF);

  \draw[wall,darkgreen] (bpA) to[out=260, in=100] (bpB);
  \draw[wall,darkgreen] (bpC) to[out=260, in=100] (bpD);
  \draw[wall,darkgreen] (bpE) to[out=260, in=100] (bpF);

  \draw[wall,gray!45,name path = out21] (bpA) to (stokes21);
  \draw[wall,gray!45,name path = out31] (bpC) to (stokes31);
  \draw[wall,gray!45,name path = out32] (bpE) to (stokes32);
  \draw[wall,gray!45,name path = out12] (bpB) to (stokes12);
  \draw[wall,gray!45,name path = out13] (bpD) to (stokes13);
  \draw[wall,gray!45,name path = out23] (bpF) to (stokes23);

  \draw[wall,purple,name path = ht12] (bpA) to[out=30,in=90] (2,0) to[out=270,in=330] (bpB);
  \draw[wall,gray!45,name path = ht13] (bpC) to[out=40,in=90] (0.7,0) to[out=270,in=320] (bpD);
  \draw[wall,purple,name path = ht23] (bpE) to[out=30,in=90] (0.3,0) to[out=270,in=330] (bpF);

  \path[name intersections={of=ht13 and out23, by={inner-int-1}}];
  \path[name intersections={of=ht13 and out32, by={inner-int-2}}];
  \path[name intersections={of=ht12 and out13, by={outer-int-2113}}];
  \path[name intersections={of=ht12 and out23, by={outer-int-1223}}];
  \path[name intersections={of=ht12 and out31, by={outer-int-1231}}];
  \path[name intersections={of=ht12 and out32, by={outer-int-2132}}];

  \draw[wall,blue] (inner-int-2) to[out=50,in=90] (1.2,0) to[out=270,in=310] (inner-int-1);
  \draw[wall,gray!45] (inner-int-2) to[out=230,in=270] (-0.75,0) to[out=90,in=130] (inner-int-1);

  \draw[wall,gray!45] (outer-int-2113) to ($(stokes23) - (0.06,0)$);
  \draw[wall,gray!45] (outer-int-1223) to ($(stokes13) + (0.06,0)$);
  \draw[wall,gray!45] (outer-int-1231) to ($(stokes32) - (0.06,0)$);
  \draw[wall,gray!45] (outer-int-2132) to ($(stokes31) + (0.06,0)$);

  \draw[branchpoint] plot coordinates {(bpA)};
  \draw[branchpoint] plot coordinates {(bpB)};
  \draw[branchpoint] plot coordinates {(bpC)};
  \draw[branchpoint] plot coordinates {(bpD)};
  \draw[branchpoint] plot coordinates {(bpE)};
  \draw[branchpoint] plot coordinates {(bpF)};

  \draw[singularpoint] plot coordinates {(sp)};

  \node[stokeslabel] at ($(stokes21) + (0,0.1)$) {$21$};
  \node[stokeslabel] at ($(stokes12) + (0,-0.1)$) {$12$};
  \node[stokeslabel] at ($(stokes32) + (0,0.1)$) {$32$};
  \node[stokeslabel] at ($(stokes23) + (0,-0.1)$) {$23$};
  \node[stokeslabel] at ($(stokes31) + (0,0.1)$) {$31$};
  \node[stokeslabel] at ($(stokes13) + (0,-0.1)$) {$13$};

  \node[right,purple] at (2.0,0) {$\tilde h_{12}$};
  \node[right,purple] at (0.25,0) {$\tilde h_{23}$};
  \node[right,blue] at (0.9,-0.6) {$\tilde h_{13}$};
  \node[right,darkgreen] at (-4.05,0) {$h_{12}$};
  \node[right,darkgreen] at (-1.6,0) {$h_{13}$};
  \node[right,darkgreen] at (-0.65,0) {$h_{23}$};

  \draw[wall,blue] (bpC) to (inner-int-2);
  \draw[wall,blue] (bpE) to (inner-int-2);
  \draw[wall,blue] (bpD) to (inner-int-1);
  \draw[wall,blue] (bpF) to (inner-int-1);

  \draw[branchpoint] plot coordinates {(bpA)};
  \draw[branchpoint] plot coordinates {(bpB)};
  \draw[branchpoint] plot coordinates {(bpC)};
  \draw[branchpoint] plot coordinates {(bpD)};
  \draw[branchpoint] plot coordinates {(bpE)};
  \draw[branchpoint] plot coordinates {(bpF)};

  \draw[singularpoint] plot coordinates {(sp)};

  \node[stokeslabel] at ($(stokes21) + (0,0.1)$) {$21$};
  \node[stokeslabel] at ($(stokes12) + (0,-0.1)$) {$12$};
  \node[stokeslabel] at ($(stokes32) + (0,0.1)$) {$32$};
  \node[stokeslabel] at ($(stokes23) + (0,-0.1)$) {$23$};
  \node[stokeslabel] at ($(stokes31) + (0,0.1)$) {$31$};
  \node[stokeslabel] at ($(stokes13) + (0,-0.1)$) {$13$};
\end{tikzpicture}

\caption{The $6$ 
finite webs in the spectral network $\cW(u,A,\vartheta = 0)$ responsible for the 
interlacing inequalities in the $n=3$ case. $5$ of them are saddle connections
($3$ green, $2$ purple) and $1$ is a 5-string tree (blue). Each finite web is labeled with its corresponding charge.
Other walls of the network are shown in gray.
\label{fig:sn-gl3-3}}
\end{figure}
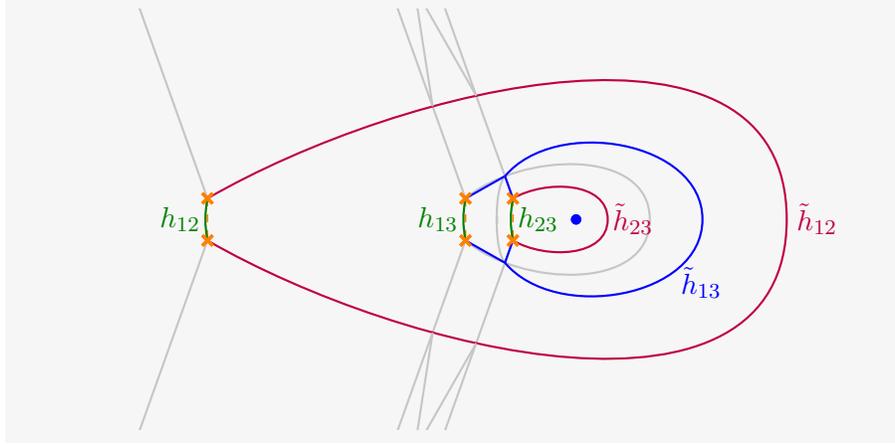

\begin{figure}[h]
\centering
\begin{tikzpicture}[withbackgroundrectangle]

    \path[use as bounding box] (-6.1,-2.8) rectangle (3.1,2.8);

    \def\cutheight{0.4}

    \foreach \x/\i/\j in {-4/1/2,-2/1/3,-1/2/3} {
        \draw[branchpoint] plot coordinates {(\x,-\cutheight)};
        \draw[branchpoint] plot coordinates {(\x,\cutheight)};
        \draw[branchcut] (\x,-\cutheight) -- (\x,\cutheight);
        \node[cutlabel] at (\x,0) {$(\i\j)$};
    }
    \draw[singularpoint] plot coordinates {(0,0)}; 

    \coordinate (L) at (-6,0);
    \coordinate (R) at (3,0);

    \draw[-{Latex},thick,red!70] 
    ($(R) + (-0.01,-0.02)$) 
    .. controls (-3.0,-3.7) .. 
    ($(L) + (0.01,-0.02)$)
    node[above,midway] {${c_i}$}; 

    \draw[-{Latex},thick,red!70] 
    ($(L) + (0.01,0.02)$) 
    .. controls (-3.0,3.7) .. 
    ($(R) + (-0.01,0.02)$)
    node[below,midway] {${d_i}$}; 

    \foreach \x/\i/\j/\lowercol/\uppercol/\labelposmid/\labelposstart/\labelposend in {-4/1/2/color12lines/color12lines/0.35/0.15/0.85,-2/1/3/color13lines/color13lines/0.35/0.10/0.90,-1/2/3/color23lines/color23lines/0.35/0.05/0.95} {

        \coordinate (cutcenter) at (\x,0);

        \draw[-{Latex},thick,rounded corners,\lowercol]
                ($(R) + (-0.01,-0.02)$)
            .. controls ($(cutcenter)+(0,-2.5)$) ..
            ($(cutcenter)+(0.1,-\cutheight+0.15)$)
            --
            node[right,yshift=-12pt] {${a_{\j\i}}$} 
            ($(cutcenter)+(-0.1,-\cutheight+0.15)$)
            .. controls ($(cutcenter)+(0,-2.5)$) ..
            ($(L) + (0.01,-0.02)$);

        \draw[-{Latex},thick,rounded corners,\uppercol]
                ($(L) + (0.01,0.02)$)
            .. controls ($(cutcenter)+(0,2.5)$) ..
            ($(cutcenter)+(-0.1,\cutheight-0.15)$)
            --
            node[right,yshift=12pt] {${a_{\i\j}}$} 
            ($(cutcenter)+(0.1,\cutheight-0.15)$)
            .. controls ($(cutcenter)+(0,2.5)$) ..
            ($(R) + (-0.01,0.02)$);

    }

    \draw[antistokesmark] ($(L) - (0.01,0)$) circle;
    \draw[antistokesmark] ($(R) + (0.01,0)$) circle;

\end{tikzpicture}

\caption{Open paths ${a_{ij}}$ ($i \neq j$), ${c_i}$, ${d_i}$,
generating $H(u,A,\vartheta=0)$. Each ${a_{ij}}$ begins on sheet $i$ and ends on sheet $j$;
each ${c_i}$ or ${d_i}$ begins and ends on sheet $i$. \label{fig:paths-gl3}}
\end{figure}
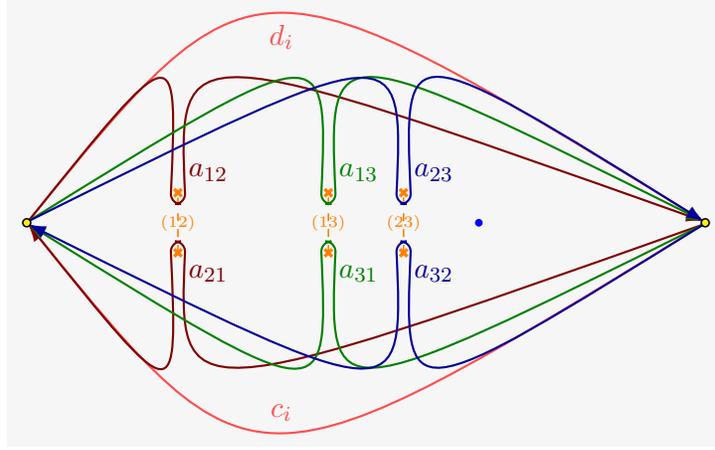

\subsubsection*{Stokes matrix entry asymptotics}

Now suppose that we are interested in studying the $\eps \to 0$ asymptotics of the Stokes matrix entries
$(S_+)_{ij}$.
In \autoref{fig:paths-gl3} we show $12$ open paths ${a_{ij}}$ ($i \neq j$), ${c_i}$, ${d_i}$,
generating $H(u,A,\vartheta=0)$. The Stokes matrix entries can be
expanded in terms
of spectral coordinates associated to these paths:

\begin{pro} \label{pro:stokes-expansion-gl3}
The expansion \eqref{eq:stokes-matrix-expansion} of Stokes matrix entries in spectral coordinates is:
\begin{align}
  (S_+)_{ii} &= X^0_{c_i} , \\
  (S_+)_{12} &= X^0_{{a_{21}}} + \cdots , \\
  (S_+)_{13} &= X^0_{{a_{32}} - {c_2} + {a_{21}}} + \cdots , \\
  (S_+)_{23} &= X^0_{{a_{32}}} + X^0_{{a_{31}}-{a_{21}}+{c_1}} + \cdots , \label{eq:sp23-expansion}
\end{align}
where $\cdots$ represents a sum of terms $X_{\gamma}$ with
$\re Z({\gamma})$ smaller than that for the leading term.
In \eqref{eq:sp23-expansion}, the leading term can be either of the two terms
shown.
\end{pro}

\begin{proof}
All these formulas are obtained using \autoref{prop:path-lifting-canonical}. (For the reader who is interested 
in reproducing them, we note one tricky point: the leading term indicated in $(S_+)_{13}$ arises from a term with
two detours; there are two other terms in $(S_+)_{13}$ which naively appear to have larger $\re Z(\gamma)$, but
those two terms cancel one another.)
\end{proof}

To get from these formulas to the $\eps \to 0^+$ asymptotics for $\abs{(S_+)_{ij}}$, 
we need to compute the real parts of the relevant periods. This can be
done using the involution $\iota$ as we did in the $n=2$ case.
In particular, for the classes appearing in \autoref{pro:stokes-expansion-gl3}
we have
\begin{align}
  \re Z(c_i) &= -\frac12 Z(V^{(1)}_i) = \frac12 t_i \, , \\
  \re Z({{a_{21}}}) &= -\frac12 Z({V^{(2)}_2}) = \frac12 \xi^{(2)}_2 \, , \\
  \re Z({{a_{32}} - {c_2} + {a_{21}}}) &= -\frac12 Z({V^{(3)}_3}) \, = \frac12 \xi^{(3)}_3 \, , \\
  \re Z({{a_{32}}}) &= -\frac12 Z({V^{(3)}_3 - h_{12}})  \, = \frac12 (\xi^{(3)}_3 - \xi^{(1)}_1 + \xi_1^{(2)}) \, , \\
  \re Z({{a_{31}}-{a_{21}}+{c_1}}) &= -\frac12 Z({V^{(3)}_3 - h_{23}})  \, = \frac12 (\xi^{(3)}_3 - \xi^{(2)}_2 + \xi_2^{(3)}) \, .
\end{align}
The resulting asymptotics are:
\begin{pro}
For $u$ sufficiently close to the caterpillar line and $A$ near diagonal, 
if \autoref{conj:wkb-and-sn} holds, the leading $\eps \to 0^+$ asymptotics of the
norms of entries of $S_+$ are
\begin{equation} \label{eq:leading-asymptotics-gl3}
  \eps \log \abs{(S_+)_{ij}} \sim \frac12 \begin{pmatrix} t_1 & \xi_2^{(2)} &  \lambda_3 \\ 
   & t_2 & \max(\lambda_3 - \xi_2^{(2)} + t_2, \lambda_3 - \xi_2^{(2)} + \lambda_2) \\
   &  & t_3 \end{pmatrix} \, .
\end{equation}
\end{pro}

\subsubsection*{Factorization coordinates}

In \eqref{eq:leading-asymptotics-gl3} 
we see that the leading asymptotic of the Stokes matrix entry
$(S_+)_{23}$ is controlled by a piecewise linear function 
of periods, rather than by a single period.
This contrasts with the statement of \autoref{mainconj},
which says that the minor coordinates of $S_+$ each should be
controlled by a single period.
How do we see from the point of view of the spectral network
that the minor coordinates should behave better than the entry 
coordinates?

In this section, we will answer this question, by giving
a spectral-network interpretation of the 
\textit{factorization coordinates}:
these are two related coordinate systems, defined
on two double Bruhat cells $B^+ \cap B^- w_0 B^-$
associated to reduced words for the 
permutation $w_0 = (123)$.
The factorization coordinates are related by a monomial transformation to the minor coordinates, so asymptotic formulas for one easily translate to asymptotic formulas for the other.

Let us briefly recall what the factorization coordinates are.
We consider the reduced word $w_0 = (23)(12)(23)$,
which corresponds to a coordinate system on the subset
\begin{equation}
 U = \{ Y \in B^+ \mid Y_{12} \neq 0 \text{ or } Y_{13} = 0 \} \quad \subset \quad B^+ \, .    
\end{equation}
Namely, any $Y \in U$ admits a unique factorization
\begin{equation} \label{eq:factorization-def-gl3}
    Y = \begin{pmatrix} \delta_1 & 0 & 0 \\ 0 & \delta_2 & 0 \\ 0 & 0 & \delta_3 \end{pmatrix} \begin{pmatrix} 1 & 0 & 0 \\ 0 & 1 & \alpha \\ 0 & 0 & 1 \end{pmatrix} \begin{pmatrix} 1 & \beta & 0 \\ 0 & 1 & 0  \\ 0 & 0 & 1 \end{pmatrix} \begin{pmatrix} 1 & 0 & 0 \\ 0 & 1 & \gamma \\ 0 & 0 & 1 \end{pmatrix} = \begin{pmatrix} \delta_1 & \delta_1 \beta & \delta_1 \beta \gamma \\ 0 & \delta_2 & \delta_2 (\alpha + \gamma) \\ 0 & 0 & \delta_3 \end{pmatrix} \, ,
\end{equation}
and this factorization defines the coordinates $(\delta_1,\delta_2,\delta_3,\alpha,\beta,\gamma)$ of $Y$.

To study the asymptotics of the factorization coordinates of $S_+$,
we apply an equivalence of spectral networks in the sense of
\cite{GMN2}, illustrated in \autoref{fig:gl3-sn-2}.
\begin{figure}[h]
\centering
\begin{tikzpicture}[withbackgroundrectangle]

 \useasboundingbox (-3,-3) rectangle (2,3); 

  \newcommand{\xOne}{-1.7}
  \newcommand{\xTwo}{-0.9}
  \newcommand{\xThree}{-0.6}
  \newcommand{\yBP}{0.2}
  \newcommand{\yStokes}{3}

  \coordinate (bpA) at (\xOne,\yBP);
  \coordinate (bpB) at (\xOne,-\yBP);
  \coordinate (bpC) at (\xTwo,\yBP);
  \coordinate (bpD) at (\xTwo,-\yBP);
  \coordinate (bpE) at (\xThree,\yBP);
  \coordinate (bpF) at (\xThree,-\yBP);
  \coordinate (sp) at (0,0);
  \coordinate (stokes21) at (\xOne-1,\yStokes);
  \coordinate (stokes31) at (\xTwo-1,\yStokes);
  \coordinate (stokes32) at (\xThree-1,\yStokes);
  \coordinate (stokes12) at (\xOne-1,-\yStokes);
  \coordinate (stokes13) at (\xTwo-1,-\yStokes);
  \coordinate (stokes23) at (\xThree-1,-\yStokes);

  \draw[branchcut] (bpA) -- (bpB);
  \draw[branchcut] (bpC) -- (bpD);
  \draw[branchcut] (bpE) -- (bpF);

  \draw[wall] (bpA) to[out=260, in=100] (bpB);
  \draw[wall] (bpC) to[out=260, in=100] (bpD);
  \draw[wall] (bpE) to[out=260, in=100] (bpF);

  \draw[wall,name path = out21] (bpA) to (stokes21);
  \draw[wall,name path = out31] (bpC) to (stokes31);
  \draw[wall,name path = out32] (bpE) to (stokes32);
  \draw[wall,name path = out12] (bpB) to (stokes12);
  \draw[wall,name path = out13] (bpD) to (stokes13);
  \draw[wall,name path = out23] (bpF) to (stokes23);

  \draw[wall,name path = ht12] (bpA) to[out=40,in=90] (2,0) to[out=270,in=320] (bpB);
  \draw[wall,name path = ht13] (bpC) to[out=30,in=90] (0.7,0) to[out=270,in=330] (bpD);
  \draw[wall,name path = ht23] (bpE) to[out=30,in=90] (0.5,0) to[out=270,in=330] (bpF);

  \path[name intersections={of=ht13 and out23, by={inner-int-1}}];
  \path[name intersections={of=ht13 and out32, by={inner-int-2}}];
  \path[name intersections={of=ht12 and out13, by={outer-int-2113}}];
  \path[name intersections={of=ht12 and out23, by={outer-int-1223}}];
  \path[name intersections={of=ht12 and out31, by={outer-int-1231}}];
  \path[name intersections={of=ht12 and out32, by={outer-int-2132}}];

  \draw[wall] (inner-int-2) to[out=20,in=90] (0.6,0) to[out=270,in=340] (inner-int-1);

  \draw[wall] (outer-int-2113) to ($(stokes23) - (0.06,0)$);
  \draw[wall] (outer-int-1223) to ($(stokes13) + (0.06,0)$);
  \draw[wall] (outer-int-1231) to ($(stokes32) - (0.06,0)$);
  \draw[wall] (outer-int-2132) to ($(stokes31) + (0.06,0)$);

  \draw[branchpoint] plot coordinates {(bpA)};
  \draw[branchpoint] plot coordinates {(bpB)};
  \draw[branchpoint] plot coordinates {(bpC)};
  \draw[branchpoint] plot coordinates {(bpD)};
  \draw[branchpoint] plot coordinates {(bpE)};
  \draw[branchpoint] plot coordinates {(bpF)};

  \draw[singularpoint] plot coordinates {(sp)};

  \node[stokeslabel] at ($(stokes21) + (0,0.1)$) {$21$};
  \node[stokeslabel] at ($(stokes12) + (0,-0.1)$) {$12$};
  \node[stokeslabel] at ($(stokes32) + (0,0.1)$) {$32$};
  \node[stokeslabel] at ($(stokes23) + (0,-0.1)$) {$23$};
  \node[stokeslabel] at ($(stokes31) + (0,0.1)$) {$31$};
  \node[stokeslabel] at ($(stokes13) + (0,-0.1)$) {$13$};
\end{tikzpicture}
\hspace{0.4cm}
\begin{tikzpicture}[withbackgroundrectangle]

 \useasboundingbox (-2.3,-3) rectangle (2,3); 

  \newcommand{\xOne}{-0.8}
  \newcommand{\xTwo}{-1.1}
  \newcommand{\xThree}{-0.5}
  \newcommand{\yBPInner}{0.2}
  \newcommand{\yBPOuter}{0.7}
  \newcommand{\yStokes}{3}

  \coordinate (bpA) at (\xOne,\yBPOuter);
  \coordinate (bpB) at (\xOne,-\yBPOuter);
  \coordinate (bpC) at (\xTwo,\yBPInner);
  \coordinate (bpD) at (\xTwo,-\yBPInner);
  \coordinate (bpE) at (\xThree,\yBPInner);
  \coordinate (bpF) at (\xThree,-\yBPInner);
  \coordinate (sp) at (0,0);
  \coordinate (stokes-am) at (\xTwo-1.0,-\yStokes);
  \coordinate (stokes-bm) at (\xTwo-0.6,-\yStokes);
  \coordinate (stokes-cm) at (\xTwo-0.2,-\yStokes);
  \coordinate (stokes-dm) at (\xThree-0.2,-\yStokes);
  \coordinate (stokes-em) at (\xThree+0.2,-\yStokes);
  \coordinate (stokes-ap) at (\xTwo-1.0,\yStokes);
  \coordinate (stokes-bp) at (\xTwo-0.6,\yStokes);
  \coordinate (stokes-cp) at (\xTwo-0.2,\yStokes);
  \coordinate (stokes-dp) at (\xThree-0.2,\yStokes);
  \coordinate (stokes-ep) at (\xThree+0.2,\yStokes);

  \draw[branchcut] (bpA) -- (bpB);
  \draw[branchcut] (bpC) -- (bpD);
  \draw[branchcut] (bpE) -- (bpF);

  \draw[wall,name path = h12] (bpA) to[out=180, in=90] (-1.55,0) to[out=270,in=180] (bpB);
  \draw[wall,name path = h13] (bpC) to[out=260, in=100] (bpD);
  \draw[wall,name path = h23] (bpE) to[out=260, in=100] (bpF);
  \draw[wall,name path = ht12] (bpA) to[out=30,in=90] (2,0) to[out=270,in=330] (bpB);
  \draw[wall,name path = ht13] (bpC) to[out=30,in=90] (0.7,0) to[out=270,in=330] (bpD);
  \draw[wall,name path = ht23] (bpE) to[out=30,in=90] (0.5,0) to[out=270,in=330] (bpF);

  \draw[wall,name path = out12,red] (bpB) to (stokes-am);
  \draw[wall,name path = out13now23,red] (bpD) to (stokes-cm);
  \draw[wall,name path = out23,red] (bpF) to (stokes-em);

  \draw[wall,name path = out21] (bpA) to (stokes-ap);
  \draw[wall,name path = out31now32] (bpC) to (stokes-cp);
  \draw[wall,name path = out32] (bpE) to (stokes-ep);

  \path[name intersections={of=ht13 and out23, by={inner-int-1}}];
  \path[name intersections={of=ht13 and out32, by={inner-int-2}}];
  \draw[wall] (inner-int-2) to[out=20,in=90] (0.6,0) to[out=270,in=340] (inner-int-1);

  \path[name intersections={of=h12 and out13now23, by={outer-int-Am}}];
  \path[name intersections={of=ht12 and out23, by={outer-int-Bm}}];
  \path[name intersections={of=out12 and out13now23, by={outer-int-Cm}}];

  \path[name intersections={of=h12 and out31now32, by={outer-int-Ap}}];
  \path[name intersections={of=ht12 and out32, by={outer-int-Bp}}];
  \path[name intersections={of=out21 and out31now32, by={outer-int-Cp}}];

  \draw[wall] (outer-int-Cm) to (stokes-bm);
  \draw[wall] (outer-int-Bm) to (stokes-dm);
  \draw[wall] (outer-int-Am) to[out=290,in=70] (outer-int-Cm);

  \draw[wall] (outer-int-Cp) to (stokes-bp);
  \draw[wall] (outer-int-Bp) to (stokes-dp);
  \draw[wall] (outer-int-Ap) to[out=70,in=290] (outer-int-Cp);

  \draw[branchpoint] plot coordinates {(bpA)};
  \draw[branchpoint] plot coordinates {(bpB)};
  \draw[branchpoint] plot coordinates {(bpC)};
  \draw[branchpoint] plot coordinates {(bpD)};
  \draw[branchpoint] plot coordinates {(bpE)};
  \draw[branchpoint] plot coordinates {(bpF)};

  \draw[singularpoint] plot coordinates {(sp)};

  \node[stokeslabel] at ($(stokes-am) + (0,-0.1)$) {$12$};
  \node[stokeslabel] at ($(stokes-bm) + (0,-0.1)$) {$13$};
  \node[stokeslabel] at ($(stokes-cm) + (0,-0.1)$) {$23$};
  \node[stokeslabel] at ($(stokes-dm) + (0,-0.1)$) {$13$};
  \node[stokeslabel] at ($(stokes-em) + (0,-0.1)$) {$23$};

  \node[stokeslabel] at ($(stokes-ap) + (0,0.1)$) {$21$};
  \node[stokeslabel] at ($(stokes-bp) + (0,0.1)$) {$31$};
  \node[stokeslabel] at ($(stokes-cp) + (0,0.1)$) {$32$};
  \node[stokeslabel] at ($(stokes-dp) + (0,0.1)$) {$31$};
  \node[stokeslabel] at ($(stokes-ep) + (0,0.1)$) {$32$};
\end{tikzpicture}
\hspace{0.4cm}
\begin{tikzpicture}[withbackgroundrectangle]

 \useasboundingbox (-1.8,-3) rectangle (2,3); 

  \newcommand{\xOne}{-0.8}
  \newcommand{\xTwo}{-1.1}
  \newcommand{\xThree}{-0.5}
  \newcommand{\yBPInner}{0.2}
  \newcommand{\yBPOuter}{0.7}
  \newcommand{\yStokes}{3}

  \coordinate (bpA) at (\xOne,\yBPOuter);
  \coordinate (bpB) at (\xOne,-\yBPOuter);
  \coordinate (bpC) at (\xTwo,\yBPInner);
  \coordinate (bpD) at (\xTwo,-\yBPInner);
  \coordinate (bpE) at (\xThree,\yBPInner);
  \coordinate (bpF) at (\xThree,-\yBPInner);
  \coordinate (sp) at (0,0);
  \coordinate (stokes-am) at (\xOne,-\yStokes);
  \coordinate (stokes-bm) at (\xTwo,-\yStokes);
  \coordinate (stokes-cm) at (\xTwo-0.3,-\yStokes);
  \coordinate (stokes-dm) at (\xThree,-\yStokes);
  \coordinate (stokes-em) at (\xThree+0.3,-\yStokes);
  \coordinate (stokes-ap) at (\xOne,\yStokes);
  \coordinate (stokes-bp) at (\xTwo,\yStokes);
  \coordinate (stokes-cp) at (\xTwo-0.3,\yStokes);
  \coordinate (stokes-dp) at (\xThree,\yStokes);
  \coordinate (stokes-ep) at (\xThree+0.3,\yStokes);

  \draw[branchcut] (bpA) -- (bpB);
  \draw[branchcut] (bpC) -- (bpD);
  \draw[branchcut] (bpE) -- (bpF);

  \draw[wall,name path = h12] (bpA) to[out=180, in=90] (-1.55,0) to[out=270,in=180] (bpB);
  \draw[wall,name path = h13] (bpC) to[out=260, in=100] (bpD);
  \draw[wall,name path = h23] (bpE) to[out=260, in=100] (bpF);
  \draw[wall,name path = ht12] (bpA) to[out=30,in=90] (2,0) to[out=270,in=330] (bpB);
  \draw[wall,name path = ht13] (bpC) to[out=30,in=90] (0.7,0) to[out=270,in=330] (bpD);
  \draw[wall,name path = ht23] (bpE) to[out=30,in=90] (0.5,0) to[out=270,in=330] (bpF);

  \draw[wall,name path = out12,red] (bpB) to (stokes-am);
  \draw[wall,name path = out13now23,red] (bpD) to (stokes-cm);
  \draw[wall,name path = out23,red] (bpF) to (stokes-em);

  \draw[wall,name path = out21] (bpA) to (stokes-ap);
  \draw[wall,name path = out31now32] (bpC) to (stokes-cp);
  \draw[wall,name path = out32] (bpE) to (stokes-ep);

  \path[name intersections={of=ht13 and out23, by={inner-int-1}}];
  \path[name intersections={of=ht13 and out32, by={inner-int-2}}];
  \draw[wall] (inner-int-2) to[out=20,in=90] (0.6,0) to[out=270,in=340] (inner-int-1);

  \path[name intersections={of=h12 and out13now23, by={outer-int-Am}}];
  \path[name intersections={of=ht12 and out23, by={outer-int-Bm}}];

  \path[name intersections={of=h12 and out31now32, by={outer-int-Ap}}];
  \path[name intersections={of=ht12 and out32, by={outer-int-Bp}}];

  \draw[wall] (outer-int-Bm) to (stokes-dm);
  \draw[wall] (outer-int-Am) to (stokes-bm);

  \draw[wall] (outer-int-Bp) to (stokes-dp);
  \draw[wall] (outer-int-Ap) to (stokes-bp);

  \draw[branchpoint] plot coordinates {(bpA)};
  \draw[branchpoint] plot coordinates {(bpB)};
  \draw[branchpoint] plot coordinates {(bpC)};
  \draw[branchpoint] plot coordinates {(bpD)};
  \draw[branchpoint] plot coordinates {(bpE)};
  \draw[branchpoint] plot coordinates {(bpF)};

  \draw[singularpoint] plot coordinates {(sp)};

  \node[stokeslabel] at ($(stokes-am) + (0,-0.1)$) {$12$};
  \node[stokeslabel] at ($(stokes-bm) + (0,-0.1)$) {$13$};
  \node[stokeslabel] at ($(stokes-cm) + (0,-0.1)$) {$23$};
  \node[stokeslabel] at ($(stokes-dm) + (0,-0.1)$) {$13$};
  \node[stokeslabel] at ($(stokes-em) + (0,-0.1)$) {$23$};

  \node[stokeslabel] at ($(stokes-ap) + (0,0.1)$) {$21$};
  \node[stokeslabel] at ($(stokes-bp) + (0,0.1)$) {$31$};
  \node[stokeslabel] at ($(stokes-cp) + (0,0.1)$) {$32$};
  \node[stokeslabel] at ($(stokes-dp) + (0,0.1)$) {$31$};
  \node[stokeslabel] at ($(stokes-ep) + (0,0.1)$) {$32$};
\end{tikzpicture}
\caption{Spectral networks reached from $\cW(u,A,\vartheta = 0)$ of \autoref{fig:gl3-sn} by a sequence of equivalences. 
The figure should be read from left to right.
The configuration on the right is the one we use for the study of the factorization coordinates of $S_+$. We highlight the relevant primary walls
in red.} \label{fig:gl3-sn-2}
\end{figure}
Beginning from $\cW(u,A, \vartheta=0)$ we first slide the branch points $z_{12}^\pm$ to a position near $z_{13}^\pm$ and $z_{23}^\pm$. 
In this process $z_{12}^-$ crosses a line of type $1<3$, and
the configuration of walls is transformed as indicated in
\cite{GMN2} (section 10.6); this change does not affect the topology near infinity.
Similarly $z_{12}^+$ crosses a line of type $3<1$.
Next we move the branch cut $B_{12}$ across the cut $B_{13}$; this 
has the effect of conjugating the transposition on $B_{13}$, changing it from $(13)$ to $(23)$.
The resulting configuration is shown in the middle of \autoref{fig:gl3-sn-2}.
Finally, we simplify the configuration by moving the primary wall of type $1 < 2$ across the primary wall of type $2 < 3$. The final configuration is indicated on 
the right of \autoref{fig:gl3-sn-2}.

Because this process is an equivalence, it does not change the
spectral coordinates associated to closed paths. Moreover, because
in this equivalence no walls cross the anti-Stokes rays, the equivalence also 
does not change the spectral coordinates
or their relation to the Stokes matrix elements.
Still, this equivalence makes manifest
cancellations which would otherwise require computations to check,
and makes the relation to the factorization coordinates easier to 
see.

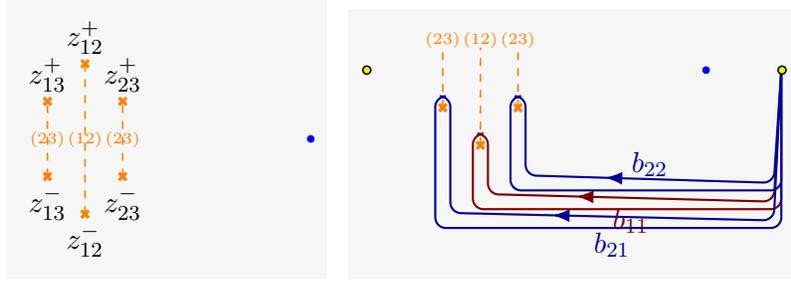
\begin{figure}[h]
\centering
\begin{tikzpicture}[withbackgroundrectangle]
    \foreach \x/\y/\i/\j/\tr in {-3/1/1/2/12,-3.5/0.5/1/3/23,-2.5/0.5/2/3/23} {
        \draw[branchpoint] plot coordinates {(\x,-\y)};
        \draw[branchpoint] plot coordinates {(\x,\y)};
        \draw[branchcut] (\x,-\y) -- (\x,\y);
        \node[cutlabel] at (\x,0) {$(\tr)$};
        \node[below] at (\x,-\y) {$z_{\i\j}^-$};
        \node[above] at (\x,\y) {$z_{\i\j}^+$};
    }
    \draw[singularpoint] plot coordinates {(0,0)};
\end{tikzpicture}
\hspace{0.1cm}
\begin{tikzpicture}[withbackgroundrectangle]

    \def\xa{-3.5}
    \def\xb{-3}
    \def\xc{-2.5}

    \foreach \x/\y/\i/\j/\tr in {-3/1/1/2/12,-3.5/0.5/2/3/23,-2.5/0.5/2/3/23} {
        \draw[branchpoint] plot coordinates {(\x,-\y)};
        \draw[branchcut] (\x,-\y) -- (\x,0.3);
        \node[cutlabel] at (\x,0.4) {$(\tr)$};
    }
    \draw[singularpoint] plot coordinates {(0,0)};
    \coordinate (L) at (-4.5,0);
    \coordinate (R) at (1,0);

    \draw[path,color12lines,witharrow=0.325] ($(R)$) -- ($(R)+(-0.1,-1.75)$) -- (\xb+0.1,-1.75+0.1) node[midway,below] {${b_{11}}$} -- (\xb+0.1,-0.85) -- (\xb-0.1,-0.85) -- (\xb-0.1,-1.75-0.1) -- ($(R)+(0,-1.75-0.1)$) -- ($(R)$);

    \draw[path,color23lines,witharrow=0.3] ($(R)$) -- ($(R)+(-0.1,-1.5)$) -- (\xc+0.1,-1.5+0.1) node[midway,above,yshift=-3] {${b_{22}}$} -- (\xc+0.1,-0.35) -- (\xc-0.1,-0.35) -- (\xc-0.1,-1.5-0.1) -- ($(R)+(0,-1.5-0.1)$) -- ($(R)$);

    \draw[path,color23lines,witharrow=0.3] ($(R)$) -- ($(R)+(-0.1,-2)$) -- (\xa+0.1,-2+0.1) node[midway,below,yshift=-2] {${b_{21}}$} -- (\xa+0.1,-0.35) -- (\xa-0.1,-0.35) -- (\xa-0.1,-2-0.1) -- ($(R)+(0,-2.1)$) -- ($(R)$);

    \draw[antistokesmark] ($(L) - (0.01,0)$) circle;
    \draw[antistokesmark] ($(R) + (0.01,0)$) circle;

\end{tikzpicture}

\caption{Left: the arrangement of the branch cuts after applying the equivalence
of spectral networks described in the main text. 
Right: paths $b_{ij}$ which correspond to the factorization coordinates.} \label{fig:cuts-after-isotopy-gl3}
\end{figure}

Computing the Stokes matrix $S_+$ from this picture gives a factorization of the form
\begin{equation} \label{eq:gl3-stokes-factorization}
    S_+ = D(\delta) M_{23}(\tilde\alpha) M_{13}(\eta_1) M_{12}(\tilde\beta) M_{13}(\eta_2) M_{23}(\tilde\gamma)
\end{equation}
where $M_{ij}(x)$ is a unipotent matrix with $ij$ entry $x$, and $D(\delta)$ is a 
diagonal matrix with $ii$ entry $\delta_i$. The five unipotent factors in \eqref{eq:gl3-stokes-factorization} correspond to the five walls visible on the bottom of the right side of \autoref{fig:gl3-sn-2}, in order (left to right in \eqref{eq:gl3-stokes-factorization} corresponds to left to right in the figure).

The coefficients $\tilde\alpha$, $\tilde\beta$, $\tilde\gamma$, $\eta_1$, $\eta_2$ 
are given by the path-lifting rule as linear combinations of spectral coordinates:
\begin{pro}
For $u$ sufficiently close to the caterpillar line and $A$ near diagonal,
we have
\begin{equation} \label{eq:gl3-tilde-asymptotics}
   \delta_i = X^0_{{c_i}}, \qquad \tilde\alpha = X^0_{{b_{21}}} + \cdots, \qquad \tilde\beta = X^0_{{b_{11}}} + \cdots, \qquad \tilde\gamma = X^0_{{b_{22}}} + \cdots, 
\end{equation}
and
\begin{equation} \label{eq:gl3-subleading-periods}
  \eta_1 = X^0_{({b_{11}}+{b_{21}})+{h_{12}}} + \cdots, \qquad \eta_2 = X^0_{({b_{22}}+{b_{11}})+{\tilde{h}_{12}}} + \cdots,    
\end{equation}
where in each expression 
$\cdots$ means a sum of terms $X^0_\gamma$, with $\re Z(\gamma)$ smaller than that
of the leading term.
\end{pro}

Because of the extra matrices $M_{13}(\eta_1)$ and $M_{13}(\eta_2)$ in the factorization \eqref{eq:gl3-stokes-factorization},
the factorization coordinates $(\alpha,\beta,\gamma)$ are not precisely equal
to $(\tilde\alpha,\tilde\beta,\tilde\gamma)$. Nevertheless, these extra matrices are negligible, in the sense that they do not affect the
asymptotics of the factorization coordinates:
\begin{pro} \label{prop:factorization-as-spectral}
For $u$ sufficiently close to the caterpillar line and $A$ near diagonal,
the factorization coordinates of $S_+$ are
\begin{equation}
    \delta_i = X^{0}_{{c_i}}, \quad (\alpha,\beta,\gamma) = (X^{0}_{{b_{21}}} + \cdots, X^{0}_{{b_{11}}} + \cdots, X^{0}_{{b_{22}}} + \cdots) 
\end{equation}
where in each expression 
$\cdots$ means a sum of terms $X^0_\gamma$, with $\re Z(\gamma)$ smaller than that
of the leading term.
\end{pro}

\begin{proof}
Directly multiplying matrices gives the relations
\begin{equation} \label{eq:matrix-rel}
    \beta = \tilde\beta, \quad \alpha + \gamma = \tilde\alpha + \tilde\gamma, \quad \beta \gamma = \tilde\beta \tilde\gamma + \eta_1 + \eta_2 \, .
\end{equation}
Moreover, from \eqref{eq:gl3-tilde-asymptotics}
and \eqref{eq:gl3-subleading-periods} it follows
that both $\eta_1$ and $\eta_2$ are subleading compared
to $\tilde\beta \tilde\gamma$,
using the facts that $\re Z(h_{12} + b_{21} - b_{22}) = \frac12 (Z(h_{12})+Z({h_{23}})) < 0$ (for $\eta_1$)
and $Z(\tilde h_{12}) < 0$ (for $\eta_2$).
Thus $\beta \gamma$ and $\tilde \beta \tilde \gamma$
differ only by subleading terms.
From there, using \eqref{eq:matrix-rel} it is straightforward to show that
$\alpha$, $\gamma$, $\beta$ 
have the same leading terms
as $\tilde \alpha$, $\tilde \gamma$, $\tilde \beta$ respectively. The latter were given in
\eqref{eq:gl3-tilde-asymptotics}.
\end{proof}

Using the involution $\iota$ again, we compute the real parts of the periods:
\begin{align}
    \re Z({{b_{11}}}) &= -\frac12 Z({\tilde h_{12}}) = \frac12(\xi_2^{(2)} - t_1), \\ 
    \re Z({{b_{21}}}) &= -\frac12 Z({\tilde h_{23}+h_{12}-h_{23}}) = \frac12(\lambda_3 - \xi_2^{(2)} + \lambda_2 - t_2), \\
    \re Z({{b_{22}}}) &= -\frac12 Z({\tilde h_{23}}) = \frac12(\lambda_3 - \xi_2^{(2)}).
\end{align}
Then we can make an explicit prediction for the asymptotics of the norms of  factorization coordinates:
\begin{pro}
For $u$ sufficiently close to the caterpillar line and $A$ near diagonal, 
if \autoref{conj:wkb-and-sn} holds, the leading $\eps \to 0^+$ asymptotics of the norms of the
factorization coordinates of $S_+$ are
\begin{gather}
   \eps \log \delta_i = \frac12 t_i, \\
   \eps \log \abs{\alpha} \sim \frac12(\lambda_3 - \xi_2^{(2)} + \lambda_2 - t_2), \quad \eps \log \abs{\beta} \sim \frac12(\xi_2^{(2)} - t_1), \quad \eps \log \abs{\gamma} \sim \frac12(\lambda_3 - \xi_2^{(2)}) \, .
\end{gather}
\end{pro}
To compare this with \autoref{mainconj}, we note that the minor coordinates are 
\begin{equation}
    \Delta_1^{(2)} = \beta \delta_1, \quad \Delta_1^{(3)} = \beta \gamma \delta_1, \quad \Delta_2^{(3)} = \beta \alpha \delta_1 \delta_2 \, .
\end{equation}
This leads to:
\begin{cor}
For $u$ sufficiently close to the caterpillar line and $A$ near diagonal, 
if \autoref{conj:wkb-and-sn} holds, the leading $\eps \to 0^+$ asymptotics of the norms of the
minor coordinates of $S_+$ are
\begin{gather}
    \eps \log \abs{\Delta_1^{(2)}} \sim \frac12 \xi_2^{(2)}, \quad \eps \log \abs{\Delta_1^{(3)}} \sim \frac12 \lambda_3, \quad \eps \log \abs{\Delta_2^{(3)}} \sim \frac12 (\lambda_2 + \lambda_3)
\end{gather}
\end{cor}
This matches with \autoref{mainconj}, if we take $L_j^{(i)} = C_j^{(i)}$.

\subsubsection*{Dual factorization coordinates}
Another system of factorization coordinates
$(\delta_1,\delta_2,\delta_3,\alpha',\beta',\gamma')$ on $U' \subset B^+$, associated to the reduced word $w_0 = (12)(23)(12)$,
is defined by
\begin{equation}
    M = \begin{pmatrix} \delta_1 & 0 & 0 \\ 0 & \delta_2 & 0 \\ 0 & 0 & \delta_3 \end{pmatrix} \begin{pmatrix} 1 & \alpha' & 0 \\ 0 & 1 & 0 \\ 0 & 0 & 1 \end{pmatrix} \begin{pmatrix} 1 & 0 & 0 \\ 0 & 1 & \beta' \\ 0 & 0 & 1 \end{pmatrix} \begin{pmatrix} 1 & \gamma' & 0 \\ 0 & 1 & 0 \\ 0 & 0 & 1 \end{pmatrix}  = \begin{pmatrix} \delta_1 & \delta_1 (\alpha' + \gamma') & \delta_1 \alpha' \beta' \\ 0 & \delta_2 & \delta_2 \beta' \\ 0 & 0 & \delta_3 \end{pmatrix}  .
\end{equation}
This system also fits into our story, as follows.
So far we have mainly considered the regime where $u_3 - u_2 \gg u_2 - u_1$.
We could also consider the 
opposite situation, where $u_2 - u_1 \gg u_3 - u_2$
(still with $u_1 < u_2 < u_3$ and $t_1 < t_2 < t_3$.)
In that case, all of our discussion goes through, with various
sign flips and index reversals. In particular, we obtain
spectral coordinate expansions indexed by the paths in \autoref{fig:cuts-after-isotopy-other-blocking-gl3}:
\begin{pro} \label{prop:dual-factorization-as-spectral}
For $u$ sufficiently close to the dual caterpillar line,
the dual factorization coordinates of $S_+$ are
\begin{equation}
    \delta_i = X^0_{{c_i}}, \quad (\alpha',\beta',\gamma') = (X^0_{{b'_{21}}} + \cdots, X^0_{{b'_{11}}} + \cdots, X^0_{{b'_{22}}} + \cdots) \, . 
\end{equation}
where in each expression 
$\cdots$ means a sum of terms $X^0_\gamma$, with $\re Z(\gamma)$ smaller than that
of the leading term.
\end{pro}

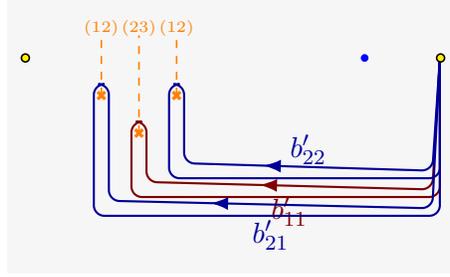
\begin{figure}[h]
\centering
\begin{tikzpicture}[withbackgroundrectangle]

    \def\xa{-3.5}
    \def\xb{-3}
    \def\xc{-2.5}

    \foreach \x/\y/\i/\j/\tr in {-3/1/2/3/23,-3.5/0.5/1/2/12,-2.5/0.5/1/2/12} {
        \draw[branchpoint] plot coordinates {(\x,-\y)};
        \draw[branchcut] (\x,-\y) -- (\x,0.3);
        \node[cutlabel] at (\x,0.4) {$(\tr)$};
    }
    \draw[singularpoint] plot coordinates {(0,0)}; 
    \coordinate (L) at (-4.5,0);
    \coordinate (R) at (1,0);

    \draw[path,color12lines,witharrow=0.3] ($(R)$) -- ($(R)+(-0.1,-1.75)$) -- (\xb+0.1,-1.75+0.1) node[midway,below] {${b'_{11}}$} -- (\xb+0.1,-0.85) -- (\xb-0.1,-0.85) -- (\xb-0.1,-1.75-0.1) -- ($(R)+(0,-1.75-0.1)$) -- ($(R)$);

    \draw[path,color23lines,witharrow=0.3] ($(R)$) -- ($(R)+(-0.1,-1.5)$) -- (\xc+0.1,-1.5+0.1) node[midway,above,yshift=-3] {${b'_{22}}$} -- (\xc+0.1,-0.35) -- (\xc-0.1,-0.35) -- (\xc-0.1,-1.5-0.1) -- ($(R)+(0,-1.5-0.1)$) -- ($(R)$);

    \draw[path,color23lines,witharrow=0.3] ($(R)$) -- ($(R)+(-0.1,-2)$) -- (\xa+0.1,-2+0.1) node[midway,below,yshift=-2] {${b'_{21}}$} -- (\xa+0.1,-0.35) -- (\xa-0.1,-0.35) -- (\xa-0.1,-2-0.1) -- ($(R)+(0,-2.1)$) -- ($(R)$);

    \draw[antistokesmark] ($(L) - (0.01,0)$) circle;
    \draw[antistokesmark] ($(R) + (0.01,0)$) circle;

\end{tikzpicture}

\caption{The analogues of the paths in 
\autoref{fig:cuts-after-isotopy-gl3}, arising near the dual caterpillar line, where
$u_2 - u_1$ is much larger than $u_3 - u_2$.} \label{fig:cuts-after-isotopy-other-blocking-gl3}
\end{figure}

\subsubsection{The case of general \texorpdfstring{$n$}{n}}

Finally we discuss general $n$, with $u_1 < \cdots < u_n$ and 
$t_1 < \cdots < t_n$, and $A$ near-diagonal. 
The spectral curve $\Gamma(u,A)$ has genus $g = \frac{(n-2)(n-1)}{2}$ and $2n$ punctures; the lattice $H(u,A,\vartheta)$ has rank $n^2 + n$.

We will not attempt a complete 
description of $\cW(u,A,\vartheta = 0)$, 
but note some general features.
The structure in the upper half-plane is the mirror-reflection of the structure
in the lower half-plane, so we only describe the lower half-plane.
The network is divided into $n-1$ domains $A^{(k)}$, corresponding to the annuli in
\autoref{fig:annuli}:
\begin{itemize}
\item The domain $A^{(n-1)}$, bounded by $V^{(n)}$ and $V^{(n-1)}$, 
contains $n-1$ branch cuts and emits $n-1$ walls to the south, of types $i < n$.

\item The domain $A^{(n-2)}$, bounded by $V^{(n-1)}$ and $V^{(n-2)}$, contains $n-2$ branch cuts which emit walls of
types $i < (n-1)$. The 
walls of type $i < n$ emerging from $A^{(n-1)}$ 
also pass through $A^{(n-2)}$, and their intersections with walls of type $i' < i$ generate
additional walls of type $i' < n$. Altogether, then, $A^{(n-2)}$
emits walls of type $i < (n-1)$ and $i < n$ to the south.

\item
Continuing inductively, the domain $A^{(j-1)}$
emits walls of type $i < k$, for each $(i,k)$ such that $i < k$ and $k \ge j$, to the south.
\end{itemize}

The spectral network $\cW(u,A,\vartheta = 0)$ contains finite webs.
There are short saddle connections with charges $h_{ij}$ for all
$i < j$. There are also saddle connections with charge $\tilde{h}_{ij}$
for $i = j-1$, and five-string trees with charge $\tilde{h}_{ij}$ for
$i < j-1$.

By an equivalence we can adjust the spectral network in the lower half-plane 
to a configuration where the branch points are arranged in a triangle.
(This equivalence involves moving each group of branch points to a row of
the triangle, one row at a time, starting from the top.)
The $j$-th row from the bottom consists of $j$ branch points,
each emitting a branch cut of type $(j,j+1)$ due north, and a wall of
type $j < j+1$ due south (as well as two other walls).
When $n > 3$ there are some branch points which are due north of others, so that
the cuts overlap, and similarly the walls overlap; however, this overlap is harmless,
since any two overlapping cuts always involve disjoint pairs $ij$, $kl$.
For convenience, we sometimes displace the branch points 
slightly horizontally to eliminate these overlaps.
 See \autoref{fig:branch-point-triangle-gl5}
for the $n=5$ case.
\begin{figure}[h]
\centering
\begin{tikzpicture}[withbackgroundrectangle]
    \def\dx{0.15}
    \foreach \x/\y/\i/\j/\tr in {
    0-\dx/0/1/2/12,
    -1-\dx/-0.5/2/3/23,
    1-\dx/-0.5/2/3/23,
    -2/-1/3/4/34,
    2/-1/3/4/34,
    0/-1/3/4/34,
    -3/-1.5/4/5/45,
    -1/-1.5/4/5/45,
   1/-1.5/4/5/45,
   3/-1.5/4/5/45
    } {
        \draw[wall] (\x,-\y) -- (\x,-1);
        \draw[branchpoint] plot coordinates {(\x,-\y)};
        \draw[branchcut] (\x,-\y) -- (\x,2.5);
        \node[cutlabel] at (\x,-\y+0.4) {$(\tr)$};
        \node[stokeslabel] at (\x,-\y-0.3) {$\i<\j$};
    }
\end{tikzpicture}
\caption{The configuration of branch points in the bottom half-plane,
after applying an equivalence as 
described in the text, in the case $n=5$. We also show the walls emanating south from 
each branch point. The full spectral network contains various additional walls, not shown here.} 
\label{fig:branch-point-triangle-gl5}
\end{figure}
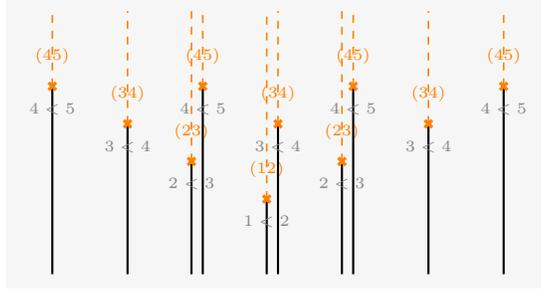

We number the branch points starting from the bottom: the bottom row is $b_{11}$, the second-bottom row $b_{21}, b_{22}$, and so on until the top
row $b_{n-1,1}, \dots, b_{n-1,n-1}$.
Then for each branch point $b_{ij}$ we define a corresponding path ${b_{ij}}$,
which begins from $\infty_{i+1}^+$, follows the boundary arc
clockwise on sheet $i+1$
to the point directly south of $b_{ij}$, then follows the wall up to
$b_{ij}$, then follows the wall back down to the boundary arc on sheet $i$,
and returns along the boundary arc to $\infty_{i}^+$.
In the case $n=3$ these paths are homologous to the ones shown in \autoref{fig:cuts-after-isotopy-gl3}.

We expect that 
the asymptotics of the factorization coordinates are determined by
the periods $Z({{c_i}}) = \frac{t_i}{2}$ and $Z({{b_{ij}}})$:
\begin{conj} \label{conj:combi}
The factorization coordinates of $S_+$ are of the form
\begin{equation}
    \alpha_{ij} = X_{{b_{ij}}} + \cdots
\end{equation}
where $\cdots$ denotes a sum of terms $X_\gamma$ with $\re Z(\gamma) < \re Z({{b_{ij}}})$.
\end{conj}
\autoref{conj:combi} is a purely geometric/combinatorial statement about the shape of the 
spectral network.
Note that if the lines shown in \autoref{fig:branch-point-triangle-gl5} were the only lines
of the spectral network, then we would have simply $\alpha_{ij} = X_{{b_{ij}}}$;
thus \autoref{conj:combi} amounts to the claim that all the other lines of the spectral network
contribute to $\alpha_{ij}$ only through correction terms with smaller $\re Z$.
We saw explicitly above that \autoref{conj:combi} is true in the $n=2$ and $n=3$ cases.

Now we want to work out the concrete asymptotics of the factorization coordinates;
for this we need to know $\re Z_{{b_{ij}}}$.
By following the curves through the 
isotopy, we get
\begin{equation}
    {b_{ij}} = {a_{i+1,j}} - {a_{i,j}}
\end{equation}
where we extend the definition of ${a_{i,j}}$ to the case $i=j$ by setting ${a_{j,j}} = {c_j}$.
Then using the real symmetry we can show that
\begin{equation}
    \re Z({{b_{ij}}}) = - \frac12 \left( -t_{i+1} + t_{i} + \sum_{k=1}^j (\xi_k^{(i+1)} - \xi_k^{(i)}) + \sum_{k=1}^{j-1} (\xi_k^{(i-1)} - \xi_k^{(i)}) \right) \, .
\end{equation}
Finally, translating from factorization coordinates to minor coordinates, we
obtain the desired asymptotics:
\begin{thm} \label{thm:summary-sn-story}
For $u$ near the caterpillar line and $A$ near diagonal,
\autoref{mainconj} follows from \autoref{conj:wkb-and-sn} and \autoref{conj:combi},
with $L^{(i)}_j = C^{(i)}_j$ as defined in \eqref{eq:cycles_C^k_i}.
\end{thm}

\subsection{Spectral coordinates for nondegenerate networks}

In this section we digress a bit, to explain how the usual connection between spectral networks, cluster algebras, and Donaldson-Thomas
invariants (finite web counts) play out
in this example.

\subsubsection{The case $n=2$}

The spectral network $\cW(u,A,\vartheta=0)$ which we have discussed up to 
now is degenerate.
These degenerate networks are often useful --- in particular, here 
they preserve
the natural symmetries of the problem at $\eps \in \Real$ --- 
but they have the drawback
that they obscure the connection to the cluster structure on 
moduli spaces of complex flat connections discussed in \cite{GS}.
Thus in this section, we make a perturbation to reach the nondegenerate
case, by varying the phase $\vartheta$.

In the $n=2$ example, the spectral networks for all $\vartheta \in (0,\pi)$
are isotopic, as are the spectral networks for $\vartheta \in (-\pi,0)$. See 
\autoref{fig:sn-gl2-2} for examples.
\begin{figure}[h]
\centering
\begin{tikzpicture}[>={Latex},withbackgroundrectangle]
    \node[anchor=south west,inner sep=0] (image) at (0,0) {\includegraphics[width=0.19\textwidth]{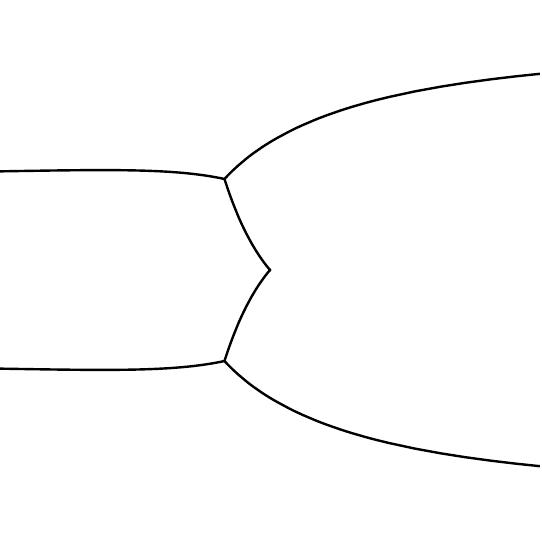}};
    \begin{scope}[x={(image.south east)},y={(image.north west)}]
        \node[stokeslabel] at (-0.0,0.5) {$1<2$};
        \node[stokeslabel] at (1.0,0.5) {$2<1$};

        \coordinate (bpA) at (0.415,0.33);
        \coordinate (bpB) at (0.415,1-0.33);
        \coordinate (sing) at (0.5,0.5);
        \coordinate (L) at (0,0.5);
        \coordinate (R) at (1,0.5);

        \draw[branchpoint] plot coordinates {(bpA)};
        \draw[branchpoint] plot coordinates {(bpB)};
        \draw[branchcut] (bpA) -- (bpB);

        \draw[singularpoint] plot coordinates {(sing)};

    \end{scope}
\end{tikzpicture}
\hspace{0.05cm}
\begin{tikzpicture}[>={Latex},withbackgroundrectangle]
    \node[anchor=south west,inner sep=0] (image) at (0,0) {\includegraphics[width=0.19\textwidth]{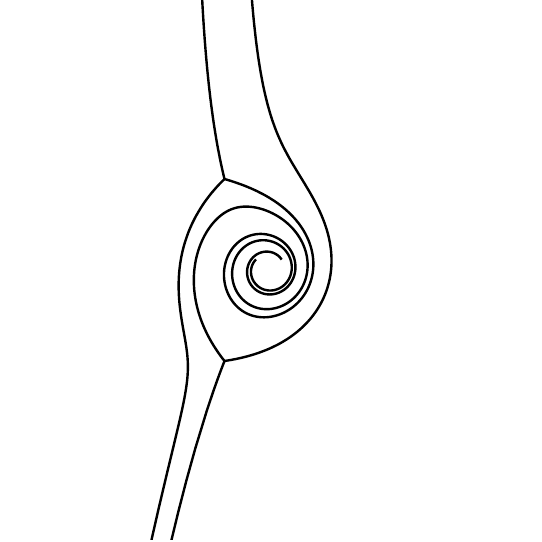}};
    \begin{scope}[x={(image.south east)},y={(image.north west)}]
        \node[stokeslabel] at (0.42,1.05) {$2<1$};
        \node[stokeslabel] at (0.3,-0.05) {$1<2$};

        \coordinate (bpA) at (0.415,0.33);
        \coordinate (bpB) at (0.415,1-0.33);
        \coordinate (sing) at (0.5,0.5);
        \coordinate (L) at (0,0.5);
        \coordinate (R) at (1,0.5);

        \draw[branchpoint] plot coordinates {(bpA)};
        \draw[branchpoint] plot coordinates {(bpB)};
        \draw[branchcut] (bpA) -- (bpB);

        \draw[singularpoint] plot coordinates {(sing)};        
    \end{scope}
\end{tikzpicture}
\hspace{0.05cm}
\begin{tikzpicture}[>={Latex},withbackgroundrectangle]
    \node[anchor=south west,inner sep=0] (image) at (0,0) {\includegraphics[width=0.19\textwidth]{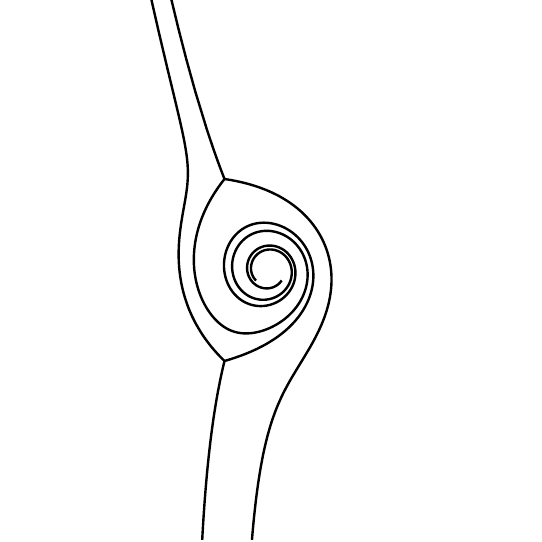}};
    \begin{scope}[x={(image.south east)},y={(image.north west)}]

        \coordinate (bpA) at (0.415,0.33);
        \coordinate (bpB) at (0.415,1-0.33);
        \coordinate (sing) at (0.5,0.5);
        \coordinate (L) at (0,0.5);
        \coordinate (R) at (1,0.5);

        \draw[branchpoint] plot coordinates {(bpA)};
        \draw[branchpoint] plot coordinates {(bpB)};
        \draw[branchcut] (bpA) -- (bpB);

        \draw[singularpoint] plot coordinates {(sing)};

        \node[stokeslabel] at (0.42,-0.05) {$1<2$};
        \node[stokeslabel] at (0.3,1.05) {$2<1$};
    \end{scope}
\end{tikzpicture}
\hspace{0.05cm}
\begin{tikzpicture}[>={Latex},withbackgroundrectangle]
    \node[anchor=south west,inner sep=0] (image) at (0,0) {\includegraphics[width=0.19\textwidth]{sn-gl2-3a.pdf}};
    \begin{scope}[x={(image.south east)},y={(image.north west)}]
        \coordinate (bpA) at (0.415,0.33);
        \coordinate (bpB) at (0.415,1-0.33);
        \coordinate (sing) at (0.5,0.5);
        \coordinate (L) at (0,0.5);
        \coordinate (R) at (1,0.5);

        \draw[branchpoint] plot coordinates {(bpA)};
        \draw[branchpoint] plot coordinates {(bpB)};
        \draw[branchcut] (bpA) -- (bpB);

        \draw[singularpoint] plot coordinates {(sing)};

        \node[stokeslabel] at (-0.0,0.5) {$2<1$};
        \node[stokeslabel] at (1.0,0.5) {$1<2$};
    \end{scope}
\end{tikzpicture}

\caption{Spectral networks $\cW(u,A,\vartheta)$ in the $n=2$ case, with $\vartheta = -\frac{\pi}{2}, -\frac{1}{10}, +\frac{1}{10}, +\frac{\pi}{2}$. (The case $\vartheta = 0$ was shown in \autoref{fig:sn-gl2-1} above.)} \label{fig:sn-gl2-2}
\end{figure}
At $\vartheta = 0$ we meet the degenerate network from \autoref{fig:sn-gl2-1}; this degeneration
separates $\cW(\vartheta_-)$ from $\cW(\vartheta_+)$, so they are not necessarily
isotopic, and indeed they are not isotopic.

Now we describe the spectral coordinate functions at $\vartheta_\pm$, 
following the recipe of \cite{GMN2,HN}.
For convenience, we introduce a branch cut on the positive $z$-axis. Then we consider the 2-dimensional vector 
space $\Sol$ of solutions of \eqref{eq:intro-main-ode} in the cut plane. There is an endomorphism $M: \Sol \to \Sol$
given by counterclockwise monodromy around $z=0$.
In $\Sol$ we consider 5 distinguished vectors:
\begin{itemize}
\item $\Psi^+_i$ is the continuation of the $i$-th column of $F_+$ from the positive $z$-axis below the cut,
\item $\Psi^-_i$ is the continuation of the $i$-th column of $F_- \exp(-\frac12 [A])$ from the negative $z$-axis,
\item $\Psi^0$ is an eigenvector of monodromy around $z = 0$, determined up to overall scale
by the condition that $\Psi^0(z)$ decays as $z \to 0$.
\end{itemize}
Note that $\Psi^-_1$ is a scalar multiple of $\Psi^+_1$ since both obey the same exponential decay condition in the half-plane containing the ray labeled $1 < 2$, and similarly, $\Psi^-_2$ is a scalar multiple of $M \Psi^+_2$ since both obey the same exponential decay condition in the half-plane containing the ray labeled $2 < 1$.
The definition of the Stokes matrix $S_+$ (\autoref{defiStokes}) gives
\begin{equation}
 \Psi^+_2 = (S_+)_{12} \Psi^-_1 + (S_+)_{22} \Psi^-_2 \, .
\end{equation}
Thus, if we let $(\cdot \wedge \cdot)$ denote any nondegenerate skew pairing on $\Sol$,\footnote{For instance,
we could fix some $z$ and then take $(\psi \wedge \psi')$ to be the determinant of a matrix with columns $\psi(z)$, $\psi'(z)$. Changing the choice of $z$ then changes $(\cdot \wedge \cdot)$ only by an overall factor, thanks to \eqref{eq:intro-main-ode}.} we have
\begin{equation} \label{eq:splus-concrete-n2}
    (S_+)_{12} = \frac{(\Psi^+_2 \wedge \Psi^-_2)}{(\Psi^-_1 \wedge \Psi^-_2)} \, .
\end{equation}

Now consider $\vartheta = \vartheta_- \in (-\pi,0)$. 
The rules of \cite{HN} give formulas for the spectral coordinates in this case:
\begin{equation}
  X^{\vartheta_-}_{{c_{1}}} = \frac{\Psi_1^+}{\Psi_1^-}, \quad X^{\vartheta_-}_{{d_{2}}} = \frac{\Psi_2^-}{M \Psi_2^+}, \quad X^{\vartheta_-}_{{c_2}} = \frac{(\Psi_2^+ \wedge \Psi_1^+)}{(\Psi_2^- \wedge \Psi_1^+)}, \quad X^{\vartheta_-}_{{d_1}} = \frac{(\Psi_1^- \wedge \Psi_2^-)}{(M \Psi_1^+ \wedge \Psi_2^-)} \, ,
\end{equation}
\begin{equation}
  X^{\vartheta_-}_{{a_{21}}} = \frac{(\Psi_2^+ \wedge \Psi^0)}{(\Psi_1^- \wedge \Psi^0)}, \quad X^{\vartheta_-}_{{a_{12}}} = \frac{(\Psi_1^- \wedge \Psi^0)}{(M \Psi_2^+,\Psi^0)} \, .
\end{equation}
Combining these formulas with \eqref{eq:splus-concrete-n2} 
we find that the Stokes matrix entry is a sum of two spectral coordinates:
\begin{equation} \label{eq:splus-spectral-n2}
  (S_+)_{12} = X^{\vartheta_-}_{{a_{21}}}(1 - X^{\vartheta_-}_{{h_{12}}}) \, .
\end{equation}

We can make similar computations for $\vartheta = \vartheta_+ \in (0,\pi)$.
There is one important subtlety: at $\vartheta = 0$ the behavior
of the eigenvectors of monodromy as $z \to 0$ reverses. We use
$\Psi^0$ to indicate the eigenvector which decays as $z \to 0$, irrespective
of the value of $\vartheta$; in consequence, $\Psi^0$ jumps
discontinuously as $\vartheta$ crosses $0$. At any rate, applying
again the rules of \cite{HN} leads to
\begin{equation}
  X^{\vartheta_+}_{{c_{1}}} = \frac{\Psi_1^+}{\Psi_1^-}, \quad X^{\vartheta_+}_{{d_{2}}} = \frac{\Psi_2^-}{M \Psi_2^+}, \quad X^{\vartheta_+}_{{c_2}} = \frac{(\Psi_2^+ \wedge \Psi_1^+)}{(\Psi_2^- \wedge \Psi_1^+)}, \quad X^{\vartheta_+}_{{d_1}} = \frac{(\Psi_1^- \wedge \Psi_2^-)}{(M \Psi_1^+ \wedge \Psi_2^-)} \, ,
\end{equation}
\begin{equation} 
  X^{\vartheta_+}_{{a_{21}}} = \frac{(\Psi_2^+ \wedge \Psi_1^+)(\Psi^0 \wedge \Psi^-_2)}{(\Psi^0 \wedge \Psi_1^+)(\Psi_1^- \wedge \Psi_2^-)}, \quad X^{\vartheta_+}_{{a_{12}}} = \frac{(\Psi_1^- \wedge \Psi_2^-)(\Psi^0 \wedge \Psi_1^+)}{(M \Psi^0 \wedge \Psi_2^-)(\Psi_2^+ \wedge \Psi_1^+)}
\end{equation}
from which it follows that
\begin{equation} \label{eq:sminus-spectral-n2}
  (S_+)_{12} = X^{\vartheta_+}_{{a_{21}}}(1 - X^{\vartheta_+}_{{\tilde h_{12}}}) \, .
\end{equation}

Now we have computed all of the $X_\gamma^{\vartheta_{\pm}}$,
and in particular we
can compare $X_\gamma^{\vartheta_+}$ to 
$X_\gamma^{\vartheta_-}$. 
When $\gamma = c_i$ or $\gamma = d_i$ we find simply
that $X_\gamma^{\vartheta_+} = X_\gamma^{\vartheta_-}$.
For $\gamma = a_{12}$ or $\gamma = a_{21}$ the relation is
more nontrivial; for instance, comparing
\eqref{eq:sminus-spectral-n2} to \eqref{eq:splus-spectral-n2}
we get\footnote{Here we simplify the notation by dropping the superscript on $X_{h_{12}}$ and $X_{\tilde h_{12}}$, since those have $X_\gamma^{\vartheta_+} = X_\gamma^{\vartheta_-}$.}
\begin{equation}
    X^{\vartheta_-}_{a_{21}} (1 - X_{h_{12}}) = X^{\vartheta_+}_{a_{21}} (1 - X_{\tilde h_{12}}) \, . 
\end{equation}
 Altogether, the transformation relating
 $X_\gamma^{\vartheta_+}$ to $X_\gamma^{\vartheta_-}$ 
 can be summarized as
\begin{equation} \label{eq:spectral-mutation-gl2}
  X^{\vartheta_+}_{\gamma} = X^{\vartheta_-}_{\gamma} \cdot \left[ (1 - X_{{h_{12}}})^{\IP{\gamma,{h_{12}}}} (1 - X_{{\tilde h_{12}}})^{\IP{\gamma,{\tilde h_{12}}}} \right]
\end{equation}
As expected, this matches with the transformation law \eqref{eq:spectral-jump-formula}, in the case of two saddle connections
with charges ${h_{12}}$ and ${\tilde h_{12}}$.

\subsubsection{Comparing to the Goncharov-Shen cluster structure on \texorpdfstring{${\mathrm G}^*$}{G*}}

We are ready to compare the spectral coordinates to the cluster coordinates in \cite{GoSh}.
Let $D(\vartheta)$ denote a punctured disc with two marked points on the boundary,
at arguments $\vartheta$ and $\vartheta + \pi$. (These marked points will correspond to the anti-Stokes directions.) The marked points divide the
boundary into two arcs.

\begin{defi} A \textit{framed $\PGL_2$-local system with pinning} over $D(\vartheta)$ is:
\begin{enumerate}
    \item A $\PGL_2$-local system $L$ over $D(\vartheta)$,
    \item A flag in $L$ (section of $L / B$) around the puncture and along each boundary arc,
    \item A decoration of the flag (lift from $L/B$ to $L/U$) along each boundary arc.
\end{enumerate}
Let $\fX(\PGL_2,D(\vartheta))$ be the moduli space parameterizing framed $\PGL_2$-local systems with pinning over $D(\vartheta)$.
\end{defi}

Assume that $\vartheta \notin \{0,\pi\}$.
Consider the flat $\GL_2$-connection \eqref{eq:intro-main-ode}
determined by the data 
$(u,A,\eps)$, with $\arg \eps = \vartheta$, 
and $u_2 > u_1$ as usual. 
The covariantly constant sections make up a $\GL_2$-local system over $D(\vartheta)$, and reducing 
by the center gives a $\PGL_2$-local system $L(u,A,\eps)$.
As $z \to \infty$ in a non-anti-Stokes direction, the line of exponentially decaying sections
gives a flag in $L(u,A,\eps)$ on each boundary arc.\footnote{This structure is sometimes called the ``Stokes filtration'', e.g. in the 
terminology of \cite{BoalchTop}.}
Choosing moreover the particular exponentially decaying
sections $\Psi^\pm$
determines a decoration of the flag.
Finally, choosing the monodromy eigenline which decays as $z \to 0$ gives a flag in $L(u,A,\eps)$ around the puncture.
(This is what requires us to have $\vartheta \notin \{0,\pi\}$;
under this condition there is one decaying eigenline and one growing one.)
Thus we obtain a point of $\fX(\PGL_2,D(\vartheta))$.\footnote{The reader might be puzzled about the matching of dimensions: we might have expected that, for each fixed choice of $u$, the map from the space of matrices $A \in \fpgl_2 = \fsl_2$ to the cluster variety should be a local diffeomorphism. This is not quite true, since $\fsl_2$ has dimension 3, while the  
cluster variety has dimension 4. The resolution is that the points of the cluster variety which we obtain are special: they obey the ``outer monodromy condition'' of \cite{GoSh}, explicitly 
$x_1^2 x_2 x_3 x_4^2 = 1$, which 
reduces the dimension by $1$. This condition is an expression of
our choice of the relative normalizations of $\Psi^\pm$ so that the Stokes matrices $S_\pm$ have the same diagonal
entries. Explicitly, at $\vartheta = \vartheta_-$ 
it follows from the relation
$x_1^2 x_2 x_3 x_4^2 = X_{- {c_1} + {c_2} + {d_1} - {d_2}}$ and the normalization conditions 
$X_{{c_i}} = X_{{d_i}}$,
and similarly at $\vartheta = \vartheta_+$.}

Now we can discuss cluster coordinates.
Each spectral network $\cW(u, A, \vartheta)$ with $\vartheta \notin \{0,\pi\}$
induces an ideal triangulation of $D(\vartheta)$, as follows. 
Each of the three walls  emanating from each branch 
point ends up either at a Stokes ray or at $z=0$;
these three ends are the vertices of a triangle containing
the branch point.
See \autoref{fig:one-triangle} for the picture in one triangle.
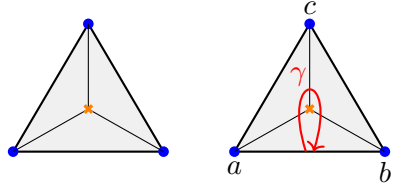
\begin{figure}[h]
\centering
\begin{tikzpicture}
  \useasboundingbox (0,-0.2) rectangle (2,2);
  \coordinate (A) at (0,0);
  \coordinate (B) at (2,0);
  \coordinate (C) at (1,1.7);
  \coordinate (AB) at (1,0);
  \coordinate (BC) at (1.5,0.85);
  \coordinate (CA) at (0.5,0.85);

  \fill[disccolor] (A) -- (B) -- (C) -- cycle; 
  \draw[thick] (A) -- (B) -- (C) -- cycle; 

  \fill[blue] (A) circle (2pt);
  \fill[blue] (B) circle (2pt);
  \fill[blue] (C) circle (2pt);

  \coordinate (Center) at (barycentric cs:A=1,B=1,C=1);

  \draw (Center) -- (A);
  \draw (Center) -- (B);
  \draw (Center) -- (C);


  \draw[branchpoint] plot coordinates {(Center)};
\end{tikzpicture}
\hspace{0.75cm}
\begin{tikzpicture}
  \useasboundingbox (0,-0.2) rectangle (2,2);
  \coordinate (A) at (0,0);
  \coordinate (B) at (2,0);
  \coordinate (C) at (1,1.7);
  \coordinate (AB) at (1,0);
  \coordinate (BC) at (1.5,0.85);
  \coordinate (CA) at (0.5,0.85);

  \fill[disccolor] (A) -- (B) -- (C) -- cycle; 
  \draw[thick] (A) -- (B) -- (C) -- cycle; 

  \fill[blue] (A) circle (2pt);
  \fill[blue] (B) circle (2pt);
  \fill[blue] (C) circle (2pt);

  \coordinate (Center) at (barycentric cs:A=1,B=1,C=1);

  \draw (Center) -- (A);
  \draw (Center) -- (B);
  \draw (Center) -- (C);


  \draw[branchpoint] plot coordinates {(Center)};

  \node[below] at (A) {$a$};
  \node[below] at (B) {$b$};
  \node[above] at (C) {$c$};

  \coordinate (Control1) at (0.6,1.1);
  \coordinate (Control2) at (1.4,1.1);

  \draw[->, red, thick] (0.95,0) .. controls (Control1) and (Control2) .. (1.05,0);

  \node[red] at (0.85,1) {$\gamma$};
\end{tikzpicture}
\caption{Left: the triangle containing one branch point of $\Gamma$. Right: an arc $\gamma$ on $\Gamma$, determined by the choice of
a triangle and an edge thereon. Recall that each vertex has a corresponding distinguished sheet of $\Gamma$; the arc $\gamma$ begins on the sheet corresponding to vertex $a$, and ends on the sheet corresponding
to vertex $b$.}
\label{fig:one-triangle}
\end{figure}
Now, for a given triangle and 
a choice of an edge, we consider the arc $\gamma$ on $\Gamma$
shown in \autoref{fig:one-triangle}.
The corresponding spectral coordinate is
\begin{equation}
    X^\vartheta_{\gamma} = \frac{(\psi_a, \psi_c)}{(\psi_b, \psi_c)} \,
\end{equation}
where $\psi_a$ is a covariantly constant section over the triangle, associated with vertex $a$. 
(If $a$ is a Stokes ray, then $\psi_a$ is the corresponding normalized decaying solution at that ray;
if $a$ is the puncture, then $\psi_a$ is the
monodromy eigensection which decays going into the puncture.)
On the other hand, \cite{GoSh} defines 
a coordinate system on $\fX(\PGL_2,D(\vartheta))$, with cluster
coordinates $x_1, \dots, x_4$ associated to the four edges.
For an edge on the boundary, the cluster coordinate is $X^\vartheta_{-\gamma}$;
for an edge in the interior, the cluster coordinate is the product 
of $X^\vartheta_{-\gamma}$ over the two faces abutting this edge.

Now we apply this to the spectral networks
in \autoref{fig:sn-gl2-2}. The corresponding triangulations of $D(\vartheta)$ are indicated in \autoref{fig:triangulations-gl2} below.
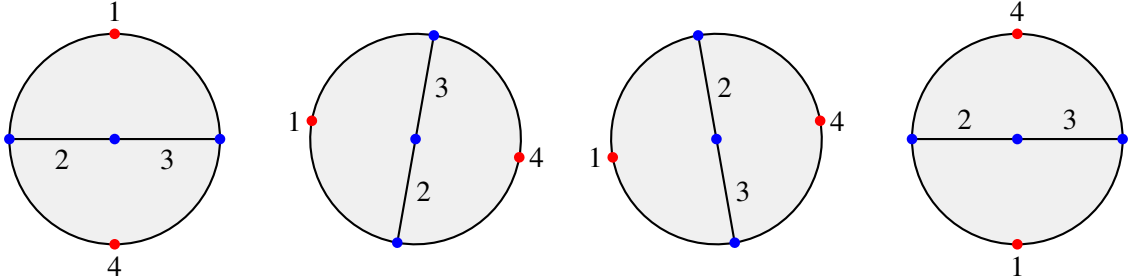
\begin{figure}[h]
\centering
\begin{tikzpicture}
  \begin{scope}[shift={(-4,0)}]
    \fill[disccolor] (0,0) circle (1.4cm); 
    \draw[thick] (0,0) circle (1.4cm); 

    \draw[thick] (0,0) -- (1.4,0); 
    \node[below] at (0.7,0) {3};

    \draw[thick] (0,0) -- (-1.4,0);
    \node[below] at (-0.7,0) {2};

    \fill[red] (0,1.4) circle (2pt); 
    \node[above, yshift=0.03cm] at (0,1.4) {1}; 

    \fill[red] (0,-1.4) circle (2pt); 
    \node[below, yshift=-0.03cm] at (0,-1.4) {4}; 

    \fill[blue] (0,0) circle (2pt); 
    \fill[blue] (1.4,0) circle (2pt); 
    \fill[blue] (-1.4,0) circle (2pt); 
  \end{scope}

  \begin{scope}[rotate=80]
    \fill[disccolor] (0,0) circle (1.4cm); 
    \draw[thick] (0,0) circle (1.4cm); 

    \draw[thick] (0,0) -- (1.4,0); 
    \node[right] at (0.7,0) {3};

    \draw[thick] (0,0) -- (-1.4,0);
    \node[right] at (-0.7,0) {2};

    \fill[red] (0,1.4) circle (2pt); 
    \node[left] at (0,1.4) {1}; 

    \fill[red] (0,-1.4) circle (2pt); 
    \node[right] at (0,-1.4) {4}; 

    \fill[blue] (0,0) circle (2pt); 
    \fill[blue] (1.4,0) circle (2pt); 
    \fill[blue] (-1.4,0) circle (2pt); 
  \end{scope}

  \begin{scope}[shift={(4,0)},rotate=100]
    \fill[disccolor] (0,0) circle (1.4cm); 
    \draw[thick] (0,0) circle (1.4cm); 

    \draw[thick] (0,0) -- (1.4,0); 
    \node[right] at (0.7,0) {2};

    \draw[thick] (0,0) -- (-1.4,0);
    \node[right] at (-0.7,0) {3};

    \fill[red] (0,1.4) circle (2pt); 
    \node[left] at (0,1.4) {1}; 

    \fill[red] (0,-1.4) circle (2pt); 
    \node[right] at (0,-1.4) {4}; 

    \fill[blue] (0,0) circle (2pt); 
    \fill[blue] (1.4,0) circle (2pt); 
    \fill[blue] (-1.4,0) circle (2pt); 
  \end{scope}

  \begin{scope}[shift={(8,0)}]
    \fill[disccolor] (0,0) circle (1.4cm); 
    \draw[thick] (0,0) circle (1.4cm); 

    \draw[thick] (0,0) -- (1.4,0); 
    \node[above] at (0.7,0) {3};

    \draw[thick] (0,0) -- (-1.4,0);
    \node[above] at (-0.7,0) {2};

    \fill[red] (0,1.4) circle (2pt); 
    \node[above, yshift=0.03cm] at (0,1.4) {4}; 

    \fill[red] (0,-1.4) circle (2pt); 
    \node[below, yshift=-0.03cm] at (0,-1.4) {1}; 

    \fill[blue] (0,0) circle (2pt); 
    \fill[blue] (1.4,0) circle (2pt); 
    \fill[blue] (-1.4,0) circle (2pt); 
  \end{scope}
\end{tikzpicture}
\caption{Triangulated discs $D(\vartheta)$ with $\vartheta = -\frac{\pi}{2}, -\frac{1}{10}, +\frac{1}{10}, +\frac{\pi}{2}$. The blue points are the vertices of the triangulation; the red points on the boundary are the marked points of $D(\vartheta)$ (corresponding to anti-Stokes lines). Each edge of the triangulation is labeled with an index $i = 1, \dots, 4$; each corresponds to a cluster coordinate $x_i$. Note that the labeling of the internal edges
jumps as $\vartheta$ crosses $0$.
These four triangulated discs are induced by the four spectral networks
in \autoref{fig:sn-gl2-2}.} \label{fig:triangulations-gl2}
\end{figure}
At $\vartheta = \vartheta_- \in (-\pi,0)$ we obtain the following
formula for the cluster coordinates $x_i$:
\begin{equation}
x_1 = X^{\vartheta_-}_{{a_{12}}-{d_2}}, \quad x_2 = X^{\vartheta_-}_{{h_{12}}}, \quad x_3 = X^{\vartheta_-}_{{\tilde h_{12}}}, \quad x_4 = -X^{\vartheta_-}_{{a_{21}}-{c_1}} \, .
\end{equation}
In short, the cluster coordinates are equal to the spectral coordinates.

At $\vartheta = \vartheta_+ \in (0,\pi)$, we get similar formulas
for the cluster coordinates $x'_i$,
\begin{equation}
x'_1 = - X^{\vartheta_+}_{-{a_{21}}+{c_2}}, \quad x'_2 = X^{\vartheta_+}_{-{h_{12}}}, \quad x'_3 = X^{\vartheta_+}_{-{\tilde h_{12}}}, \quad x'_4 = X^{\vartheta_+}_{-{a_{12}}+{d_1}} \, ,
\end{equation}
so again the cluster coordinates are equal to the spectral coordinates.

Finally, the spectral transformation law
\eqref{eq:spectral-mutation-gl2} 
translates to the relations
\begin{equation} \label{eq:cluster-mutations-gl2}
    x'_1 = x_1 (1 - x_3) (1 - x_2^{-1})^{-1}, \quad x'_2 = x_2^{-1}, \quad x'_3 = x_3^{-1}, \quad x'_4 = x_4(1-x_2)(1-x_3^{-1})^{-1} \, .
\end{equation}
The relations \eqref{eq:cluster-mutations-gl2} can be described as the
action of two cluster $\fX$-mutations, at the variables $x_2$ and $x_3$.

\subsubsection{The case $n=3$}

In the $n=2$ case above, we connected the spectral coordinates 
with the cluster structure of \cite{GoSh} on ${\mathrm G}^*$ 
by considering spectral networks for phases other than $\vartheta = 0$.

In the $n=3$ case, we can try to do the same thing, but we will 
meet a much more
intricate story, because the networks for different $\vartheta \in (-\pi,0)$
are not all isotopic: there are closed cycles $\gamma$ with $Z(\gamma) \notin \Real$,
and these cycles can be the charges of 
finite webs. The network $\cW(u,A,\vartheta)$ jumps
when $\vartheta$ crosses the phase of any finite web.

We can still make some predictions.
For $u$ near the caterpillar line, there may be 
countably many distinct finite webs whose phases approach $0$ from either direction, 
so that $\cW(u,A,\vartheta)$ jumps countably many
times as $\vartheta \to 0^\pm$. Nevertheless, we 
expect that each spectral coordinate $X^\vartheta_\gamma$ has a
limit as $\vartheta \to 0^\pm$; we call these limits $X^{0^\pm}_\gamma$.
These two limits should differ by the usual 
transformation law \eqref{eq:spectral-jump-formula}.
That transformation law depends on some BPS invariants $\Omega(\gamma)$; in 
principle they can be computed from the spectral network following the algorithm
of \cite{GMN}, but here we guess them instead
using the physical picture sketched in \autoref{sec:qft} below.
That leads to the following prediction:
\begin{equation}
      X^{0_+}_{\gamma} = X^{0_-}_{\gamma} \cdot \left( \prod_{{h}} (1 - X_{{h}})^{\IP{\gamma,{h}}} \right) \cdot (1 + X_{{v}})^{-2 \IP{\gamma,{v}}}
\end{equation}
Here $v = V^{(2)}_2 - V^{(2)}_1$, and 
$h$ runs over $8$ values, 
the differences $V^{(2)}_i - V^{(3)}_j$, 
$V^{(1)}_i - V^{(2)}_j$ for $i \ge j$ and $V^{(3)}_j - V^{(2)}_i$, $V^{(2)}_j - V^{(1)}_i$
for $i < j$.
This is the $n=3$ analogue of the formula \eqref{eq:spectral-mutation-gl2} 
which holds in the $n=2$ case.

\subsection{Relation to quantum field theory} \label{sec:qft}

The story we have been discussing is related to an ${\mathcal N}=2$ supersymmetric
quantum field theory in four dimensions, 
as recently discussed in \cite{GMNY}.

On the one hand, we may think of it as a class $S$ theory of type
$\fgl_n$ built using a twice-punctured sphere, with a regular singularity at $z = 0$
and an irregular one at $z = \infty$.
This theory has a ${\rm U}(n)$ flavor symmetry associated to the singularity at $z = 0$. In the pure class $S$ theory there
is also a $\U(1)^{n-1}$ 
flavor symmetry at the irregular puncture; the precise theory we want to consider
is obtained by gauging that symmetry.

We can also describe the same theory concretely by 
a Lagrangian, determined by the quiver diagram shown in
\autoref{fig:quiver}.

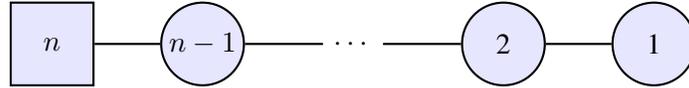
\begin{figure}[h]
\centering
\begin{tikzpicture}[node distance=2cm,thick,
  main node/.style={circle,fill=blue!10,draw,minimum size=1.1cm,inner sep=0pt},
  square node/.style={rectangle,fill=blue!10,draw,minimum size=1.1cm,inner sep=0pt}]
    
    \node[square node] (1) {$n$};
    \node[main node] (2) [right of=1] {$n-1$};
    \node (3) [right of=2] {$\cdots$};
    \node[main node] (4) [right of=3] {2};
    \node[main node] (5) [right of=4] {1};

    \path
    (1) edge (2)
    (2) edge (3)
    (3) edge (4)
    (4) edge (5);
\end{tikzpicture}

\caption{A linear quiver, describing the field content of a Lagrangian ${\mathcal N}=2$
supersymmetric quantum field theory in four dimensions. 
The circular (gauge) nodes correspond to ${\rm U}(k)$ gauge groups
with $1 \le k \le n-1$.
The leftmost node corresponds to a ${\rm U}(n)$ flavor symmetry. Each link corresponds
to a bifundamental hypermultiplet.}
\label{fig:quiver}
\end{figure}
The ${\rm SU}(n)$ parts of the gauge symmetry are conformal, but
the ${\rm U}(1)$ parts are not conformal or asymptotically free. 
Nevertheless, we can consider this theory as an effective theory.
For $n=2$, it is a ${\rm U}(1)$ gauge theory coupled to
$2$ charged hypermultiplets;
the fact that the gauge group is abelian is reflected in the fact that the Stokes
matrices can be computed exactly in this case, as we already
noted above.
For $n>2$, on the other hand, it is a nonabelian gauge theory.

The decomposition of the spectral curve shown in 
\autoref{fig:annuli} reflects the quiver
description of the theory: the loop $V^{(k)}$ corresponds to the
quiver node labeled $k$, and the annulus between $V^{(k)}$ and $V^{(k+1)}$
corresponds to a link in the quiver.
Going to the caterpillar line 
$u_i - u_{i-1} \ll u_{i+1} - u_i$ corresponds to the
weak-coupling limit in which the gauge couplings are sent to zero.

From this description of the theory we see that at weak coupling 
the BPS spectrum should include $k(k-1)$ hypermultiplets in the bifundamental
of ${\rm U}(k) \times {\rm U}(k-1)$, for $k = 2, \dots, n$,
and $W$-bosons in the adjoint of $\mathfrak{su}(k)$ for
$k = 2, \dots, n-1$.
This spectrum should then appear in the Stokes behavior of the $\varepsilon \to 0$
asymptotics at weak coupling. Concretely, this is reflected in the arguments of 
the gamma functions appearing in \autoref{app:Stokes_3}:
indeed the $\lambda_l^{(k)} - \lambda_i^{(k)}$ which appear with multiplicity 
$-2$ are the masses of the $W$-bosons,
while the $\lambda_l^{(k+1)} - \lambda_i^{(k)}$
and $\lambda_l^{(k-1)} - \lambda_i^{(k)}$ which appear with multiplicity $1$
are the masses of the hypermultiplets.

We saw part of the BPS spectrum explicitly above for $n=2$ and $n=3$. 
For $n=2$ we just expect
the bifundamental of $\U(2) \times \U(1)$,
and indeed in \autoref{fig:sn-gl2-1} 
we found $2$ saddle connections with charges 
${h_{12}}$ and ${\tilde h_{12}}$.
For $n=3$ we expect bifundamentals of 
$\U(3) \times \U(2)$ and $\U(2) \times \U(1)$.
Looking at \autoref{fig:sn-gl3-3},
we see that the 
saddle connections with charges ${h_{12}}$ and ${\tilde h_{12}}$ give the
expected bifundamental of $\U(2) \times \U(1)$,
while the webs with charges 
${{h}_{13}}, {{h}_{23}}, {{\tilde h}_{13}}, {{\tilde h}_{23}}$ give $4$ of the $6$ expected in the 
bifundamental of $\U(3) \times \U(2)$.
The remaining two charges in the bifundamental are
$h_{13}+{\tilde h}_{13}+h_{23}$
and ${\tilde h}_{13} + h_{23} + {\tilde h}_{23}$, which do not arise
as indecomposable webs in \autoref{fig:sn-gl3-3}; rather they should be
understood as composites.

\appendix

\section{The De Concini-Procesi space} \label{subsect-dCP}

Recall that we denote by $\h_{\rm reg}(\mathbb{R})$ the space of $n\times n$ diagonal matrices with distinct real eigenvalues, and in this paper we consider the meromorphic linear system
\begin{eqnarray*}
\frac{dF}{dz}=\left(\I u-\frac{1}{2\pi\I}\frac{A}{z}\right)\cdot F,
\end{eqnarray*}
with $u\in\h_{\rm reg}(\mathbb{R})$ and $A\in{\Herm}(n)$. By definition, the Stokes matrices $S_\pm(u,A)$ are invariant under the translation action of $\mathbb{R}$ on $\h_{\rm reg}(\mathbb{R})$ given by $u \mapsto u+ c \cdot \mathrm{Id}_n (c \in \mathbb{R})$. Then, for any fixed $A$, $S_\pm(u,A)$ are parameterized by $\frkt_{\rm reg}(\mathbb{R})\cong \h_{\rm reg}(\mathbb{R})/\mathbb{R}$. Here, $\frkt_{\rm reg}(\mathbb{R})$ is the space of $n\times n$ diagonal matrices $u={\rm diag}(u_1,...,u_n)$ with distinct real eigenvalues and such that $\sum_{k=1}^nu_k=0$. 

Before we introduce the De Concini-Procesi space $\widetilde{\frkt_{\rm reg}(\mathbb{R})}$ of $\frkt_{\rm reg}(\mathbb{R})$, we first work over the field of complex numbers and introduce $\widetilde{\frkt_{\rm reg}(\mathbb{C})}$.
Let  $\mathfrak{g}$ be the simple Lie algebra $\mathfrak{sl}_n$, $\frkt \subset \g $ the Cartan subalgebra, $\Pi \subset \frkt^*$ the set of roots, $\Pi_+ $ the set of positive roots and $ \{ \alpha_i : i \in I \}$  the set of simple roots where $I$ is the index set of vertices of the Dynkin diagram of $\mathfrak{g}$. 

Denote by $ \CG $  the \emph{minimal building set} associated with the set of roots.  To define $ \CG $, let $ \CG' $ denote the set of all non-zero subspaces of $\frkt^* $ which are spanned by a subset of $ \Pi $. Let $ V \in \CG' $.  We say that $ V  = V_1 \oplus \cdots \oplus V_k $ is a \emph{decomposition} of $ V $ if  $ V_1, \dots, V_k \in \CG'$, and if whenever $ \alpha  \in \Pi $ and $ \alpha \in V $, then $ \alpha \in V_i $ for some $ i$.  From Section 2.1 of \cite{dCP}, every element of $ \CG' $ admits a unique decomposition. Then we define $ \CG $ to be the set of indecomposable elements of $ \CG' $. This set can be described as follows. There is an action of the Weyl group $ W $ on $\frkt $. This action preserves $\Pi $. Thus, we get actions of $ W $ on $ \CG$ and $\CG' $. If $ J \subseteq I $ is a non-empty, connected subset of $ I $, we can form $ V_J = {\rm span}(\alpha_j : j \in J) $.  Then $ V_J \in \CG $.  Every $ V \in \CG $ is of the form $ w(V_J) $ for some $ w \in W $ and $ J $ as above.  

Note that for any $ V \in \CG$, we have a map $ \frkt_{\rm reg} \rightarrow \mathbb{P}(\g/V^\perp) $.

\begin{defi}
The De Concini-Procesi space $\widetilde{\frkt_{\rm reg}}\subset \frkt\times \prod_{V \in \CG} \mathbb{P}(\frkt/V^\perp) $ is the closure of the image of the map $ \frkt_{\rm reg} \rightarrow  \frkt\times \prod_{V \in \CG} \mathbb{P}(\frkt/V^\perp)$.  
\end{defi}
Let us consider its projective analog. Note that the multiplicative group $\mathbb{C}^\times$ acts on $\frkt_{\rm reg}$, and the regular morphism $\frkt_{\rm reg}\hookrightarrow \prod_{V \in \CG} \mathbb{P}(\frkt/V^\perp) $ is constant on the $\mathbb{C}^\times$ orbits. Thus, we get a map $ \frkt_{\rm reg}/\mathbb{C}^\times \rightarrow  \prod_{V \in \CG} \mathbb{P}(\frkt/V^\perp)$.
\begin{defi}
The De Concini-Procesi space $\overline{\frkt_{\rm reg}}\subset \prod_{V \in \CG} \mathbb{P}(\frkt/V^\perp) $ is the closure of the image of the map $ \frkt_{\rm reg}/\mathbb{C}^\times \rightarrow  \prod_{V \in \CG} \mathbb{P}(\frkt/V^\perp)$.  
\end{defi}

Thus a point of $\overline{\frkt_{\rm reg}}$ consists of a collection $ (L_V)_{V \in \CG} $ where $ L_V \in \frkt/V^\perp $.  We will think of $ L_V $ as a subspace of $ \frkt $ containing $ V^\perp $ as a hyperplane.  Note that if $ \chi \in \frkt_{\rm reg}$, then $ \chi \notin V^\perp $ for all $ V \in \CG $ and the image of $ \chi $ in $\overline{\frkt_{\rm reg}}$ is the collection $  L_V = V^\perp + \mathbb{C}\chi $.

By \cite[Theorem 4.1]{dCP}, $\widetilde{\frkt_{\rm reg}}$ is the total space of a tautological line 
bundle on $\overline{\frkt_{\rm reg}}$. Furthermore, following \cite[Theorem 3.1 and 3.2]{dCP}, the boundary $D$ of $\frkt_{\rm reg}$ in $\widetilde{\frkt_{\rm reg}}$ is a divisor with normal crossings, and is the union of smooth irreducible divisors $D_{\omega(V_J)}$ indexed by the elements ${\omega(V_J)}$ in $\CG$. Then, by \cite[Theorem 4.1]{dCP}, $\overline{\frkt_{\rm reg}}$ is isomorphic to the boundary divisor $D_{\omega(V_I)}$.

Since the root system $\Pi$ is defined over $\mathbb{R}$, the variety $\widetilde{\frkt_{\rm reg}}$ is defined over $\mathbb{R}$ and we get the De Concini-Procesi space $\widetilde{\frkt_{\rm reg}(\mathbb{R})}$ of $\frkt_{\rm reg}(\mathbb{R})$ by taking the set of real points.

\subsection{The connection to moduli spaces of stable rational curves}
The space $\overline{\frkt_{\rm reg}}$ can be identified with the Deligne-Mumford space $\overline{\mathcal{M}_{0,n+1}}$ of stable rational curves with $n+1$ marked
points. Here, a rational curve with $n+1$ marked points is a finite union $\Sigma$ of projective lines $\Sigma_1,..., \Sigma_m$ together with marked distinct points $z_1,...,z_n\in \Sigma$ such that 
\begin{itemize}
    \item Each marked point belongs to a unique $\Sigma_j$;
    \item The intersection of projective lines $\Sigma_i\cap \Sigma_j$ is either empty or consists of one point, and in the latter
case the intersection is transversal;
\item The graph of components (whose vertices are the lines $C_i$ and whose edges correspond to pairs of
intersecting lines) is a tree;
\item The total number of special points (i.e. marked points or intersection points) that belong to a
given component $\Sigma_i$
is at least 3.
\end{itemize}

A point of $\overline{\mathcal{M}_{0,n+1}}$ is then an isomorphism class of such stable curves.
The space $\overline{\mathcal{M}_{0,n+1}}$ contains a dense open subset corresponding to curves with one
component. This open subset is isomorphic to $$\left( \left(\mathbb{P}^1 \right)^{n+1}\setminus\Delta\right)/{\rm PSL}_2(\mathbb{C}$$  where  $\Delta$ is the fat diagonal

\[\Delta=\{(u_1,...,u_n)\in \mathbb{C}^n~|~u_i = u_j, \,\text{for some } i\ne j \}\] 
and ${\rm PSL}_2(\mathbb{C})$ is the automorphism group of $\mathbb{P}^1$. Since the group ${\rm PSL}_2(\mathbb{C})$ acts
transitively on triples of distinct points, we can fix the $(n+ 1)$-th marked point to be $\infty\in \mathbb{P}^1$ and
fix the sum of coordinates of other points to be zero. Then the space $\mathcal{M}_{0,n+1}$ gets identified
with the quotient $\frkt_{\rm reg}/\mathbb{C}^\times$. That is we have
$$\left( \left(\mathbb{P}^1 \right) ^{n+1}\setminus\Delta\right)/{\rm PSL}_2(\mathbb{C})\cong (\mathbb{C}^n\setminus\Delta)/B\cong \h_{\rm reg}/B\cong \frkt_{\rm reg}/\mathbb{C}^\times$$
where  $B \subset {\rm PSL}_2$ is the Borel subgroup, and $B\cong\mathbb{C}^\times\ltimes \mathbb{C}$ acting on $\mathbb{C}^n$ by scaling
and translation. 
It follows from \cite[Section 4.3]{dCP} that the above isomorphism can be extended to an identification $\overline{\mathcal{M}_{0,n+1}}\cong \overline{\frkt_{\rm reg}}$. On the other hand, the space $\mathcal{M}_{0,n+1}$ comes with the tautological bundles $L_i$ whose fiber is the line
representing the point $z_i$. Then the total space $\widetilde{\mathcal{M}_{0,n+1}}$ of the tautological line bundle $L_{n+1}$ is isomorphic to the total space of the tautological line 
bundle $\widetilde{\frkt_{\rm reg}}$ on $\overline{\frkt_{\rm reg}}$.

 Taking the real points, we have \[\widetilde{\mathcal{M}_{0,n+1}}(\mathbb{R})\cong \widetilde{\frkt_{\rm reg}}(\mathbb{R}).\] 

\subsection{Stratification indexed by planar rooted trees} 

We denote by $T(u)$ a planar rooted tree $T$ with $n$ leaves colored by the components $u_1, \dots , u_n$ of $u$. We say that $T$ is compatible with $\sigma\in S_n$ if all internal vertices of the tree are in the lower half plane, all leaves are on the horizontal line $y = 0$ and are colored by $u_{\sigma(1)}, \dots ,u_{\sigma(n)}$ from left to right.

To any planar binary rooted tree $BT$ compatible with $\sigma$, one can assign a set of coordinates $z_I$, indexed by internal vertices $I$ of $BT$, in an appropriate neighborhood $U_{BT_\sigma}$ of the corresponding 0-dimensional stratum. The coordinate ring of the open
chart $U_{BT_\sigma}\subset \widetilde{\mathfrak{t}_{\rm reg}}(\mathbb{R})$ is generated by the following set of coordinate functions $z_I$ on $U_{BT_\sigma}$ indexed by inner vertices $I$
of the tree $BT$,
\begin{align}
z_{I}=\left\{
          \begin{array}{lr}
            u_{r(I)}- u_{l(I)},   & \text{ if $I$ is the root vertex},\\
            \frac{u_{r(I)}-u_{l(I)}}{u_{r(I')}-u_{l(I')}},   & \text{ if $I$ is any other vertex},
             \end{array}
\right.
\end{align}
where $I'$ is the preceding vertex of $I$ in $BT$, i.e. $I':= {\rm max}\{J \in BT~|~ J < I\}$ in the partial ordering $<$ of the vertices of $BT_\sigma$ with the root being the minimal element, and for any vertex $I$, $l(I)\in [1, \dots , n]$ is such that 
$\sigma(l(I))$ is the maximal index of the $u_i$'s in the left branch at $I$, and analogously, $r(I)\in [1, \dots ,n]$ is such that $\sigma(r(I))$ the minimal 
index of the $u_i's$ in the right branch at $I$.

The space $\widetilde{\mathfrak{t}_{\rm reg}}(\mathbb{R})$ has a stratification, with the strata indexed by rooted trees with $n$ colored leaves.
Let $T$ be such a tree, then the corresponding
stratum $\mathcal{M}_T$ is the product of $\mathcal{M}_{0,d(I)}$ over all internal vertices $I$ of $T$ with $d(I)$ the index
of $I$. In particular, $0$-dimensional strata correspond to binary rooted trees with $n$ ordered leaves, while $1$-dimensional strata correspond to almost binary trees (with exactly
one $4$-valent internal vertex). The stratum corresponding to a tree $T$ lies in the closure of the one corresponding to another
tree $T'$ if and only if $T'$ is obtained from $T$ by contracting some edges. In the coordinate charts, the stratum $\mathcal{M}_{T}$ corresponding to a rooted tree $T$ in the local
coordinates determined by a binary rooted tree $T'$ can be described as follows.

\begin{pro}
The stratum $\mathcal{M}_{T}$ has a nonempty intersection with the coordinate chart $U_{{BT_\sigma}}$ if and only if $T$ is
obtained from $BT$ by contracting some edges. In the latter case, $\mathcal{M}_{T}$ is a subset of $U_{BT_\sigma}$ defined as follows: $z_I \ne 0$ if the (unique) edge of $BT$ which ends at $I$ is contracted in $T$, and $z_I = 0$ otherwise.
\end{pro}

Let us take the following planar binary tree $BT_{\sigma}$ with coloring
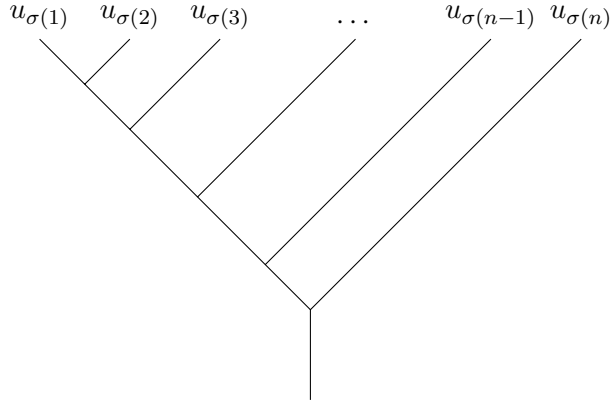
\begin{figure}[H]
\begin{center}
  \begin{tikzpicture}[scale=1.2]
  \draw
  (0,0)--(0,1)--(3,4) node[above]{$u_{\sigma(n)}$}
  (0,1)--(-3,4) node[above]{$u_{\sigma(1)}$}
  (-2.5,3.5)--(-2,4) node[above]{$u_{\sigma(2)}$}
  (-2,3)--(-1,4) node[above]{$u_{\sigma(3)}$}
  (-1.25,2.25)--(0.5,4) node[above]{$\cdots $}
  (-0.5,1.5)--(2,4) node[above]{$u_{\sigma(n-1)}$};
  
  \end{tikzpicture}
  \caption{A planar rooted tree with coloring}
 \end{center}
 \end{figure}
Let us denote the vertices of $BT_\sigma$ in the partial ordering by $I_1>I_2>\cdots >I_n$. And let 
$z_{I_1}, ..., z_{I_n}$ be the corresponding coordinates, then (following the definition in \cite[Page 16]{Sp})
\begin{defi}
The point $u^\sigma_{\rm cat}\in U_{BT_{\sigma}}$ with coordinates $z_{I_k}=0$ for all $k=1,...,n$ is called a caterpillar point. And the line $l^\sigma_{\rm cat}\in U_{BT_{\sigma}}$ consisting of the set of points in $U_{BT_{\sigma}}$ with coordinates $z_{I_k}=0$ for all $k=1,...,n-1$ is called a caterpillar line. For $\sigma={\rm id}\in S_n$, we simply denote $u^{\rm id}_{\rm cat}$ and $l^{\rm id}_{\rm cat}$ by $u_{\rm cat}$ and $l_{\rm cat}$ respectively.
\end{defi}
Note that a caterpillar point $u_{\rm cat}^{\sigma}$ is in the $0$-dimensional stratum of $\widehat{\h_{\rm reg}}(\mathbb{R})$, and $l_{\rm cat}^{\sigma}$ is the tautological line through $u_{\rm cat}^{\sigma}$.
In particular, the limit of elements $u={\rm diag}(u_1,...,u_n)$ in $U_{\rm id}$ such that $u_2 - u_1$ is equal to a fixed real number $t>0$ and  $\frac{u_{j+1}-u_{j}}{u_{j}-u_{j-1}}\rightarrow+\infty$ for all $j=2,...,n-1$, is a point, denoted by $u_{\rm cat}(t)$ in the caterpillar line $l_{\rm cat}$.  

\section{The WKB approximation of the \texorpdfstring{$3\times 3$}{3x3} Stokes matrices on the caterpillar line}     \label{app:Stokes_3}
In this appendix, we show some explicit computation of the WKB approximation of the Stokes matrices on the caterpillar line based on the formula in \autoref{reglimitcat}. Recall that for $A=(A_{ij})\in\Herm(n)$, we denote by $\lambda^{(k)}_1\le \cdots \le \lambda^{(k)}_k$ the ordered eigenvalues of its upper left $k\times k$ submatrix, and we use $t_i$ for the diagonal entry $A_{ii}$.

Let us take the case $n=3$. Then the explicit formula for the entries of the Stokes matrix $S_+^{\rm reg}(u_{\rm cat}(t),A)$ is
\begin{align*}
    (S_+^{\rm reg})_{11} &= e^{\frac{t_1}{2\eps}}, \hspace{3mm} (S_+^{\rm reg})_{22}=e^{\frac{t_2}{2\eps}}, \hspace{3mm} (S_+^{\rm reg})_{33}=e^{\frac{t_3}{2\eps}},\\
(S_+^{\rm reg})_{12} &=  \frac{A_{12}}{\eps} \cdot \frac{ \left(\frac{u_2-u_1}{\eps}\right)^{\frac{t_2-t_1}{2\pi \I\eps }} e^{\frac{t_1+t_2}{4\eps}}}{\Gamma\left(1+\frac{\lambda_1^{(2)}-t_1}{2\pi \I\eps}\right)\Gamma\left(1+\frac{\lambda_2^{(2)}-t_1}{2\pi \I\eps}\right)},\\
(S_+^{\rm reg})_{23} &= (S_+^{\rm reg})^1_{23}+(S_+^{\rm reg})^2_{23}, \\
(S_+^{\rm reg})_{13} &= (S_+^{\rm reg})^1_{13}+(S_+^{\rm reg})^2_{13},
\end{align*}
where the components
\begin{align*}
     (S_+^{\rm reg})^1_{23} &=  2\pi\I \cdot \left(\frac{u_2-u_1}{\eps}\right)^{\frac{t_3-t_2}{2\pi \I\eps}} 
   \cdot\frac{e^{\frac{{t_2+t_3}}{4\eps}}\Gamma\left(1+\frac{\lambda^{(2)}_2-\lambda^{(2)}_1}{2\pi \I\eps}\right)}{\prod_{j=1}^{3}\Gamma\left(1+\frac{\lambda^{(3)}_j-\lambda^{(2)}_1}{2\pi \I\eps}\right)}\frac{\Gamma\left(\frac{\lambda^{(2)}_2-\lambda^{(2)}_1}{2\pi \I\eps}\right)}{\Gamma\left(1+\frac{\lambda^{(1)}_1-\lambda^{(2)}_1}{2\pi \I\eps}\right)}\cdot \Delta^{1,2}_{1,3}\left(\frac{A-\lambda^{(2)}_1}{2\pi\I\eps}\right)\\
(S_+^{\rm reg})^2_{23} &= 2\pi\I \cdot \left(\frac{u_2-u_1}{\eps}\right)^{\frac{t_3-t_2}{2\pi \I\eps}} 
   \cdot\frac{e^{\frac{{t_2+t_3}}{4\eps}}\Gamma\left(1+\frac{\lambda^{(2)}_1-\lambda^{(2)}_2}{2\pi \I\eps}\right)}{\prod_{j=1}^{3}\Gamma\left(1+\frac{\lambda^{(3)}_j-\lambda^{(2)}_2}{2\pi \I\eps}\right)}\frac{\Gamma\left(\frac{\lambda^{(2)}_1-\lambda^{(2)}_2}{2\pi \I\eps}\right)}{\Gamma\left(1+\frac{\lambda^{(1)}_1-\lambda^{(2)}_2}{2\pi \I\eps}\right)}\cdot \Delta^{1,2}_{1,3}\left(\frac{A-\lambda^{(2)}_2}{2\pi\I\eps}\right),   \end{align*}
and
\begin{align*}
    (S^{\rm reg}_+)^1_{13}  &=   2\pi\I\left(\frac{u_2-u_1}{\eps}\right)^{\frac{t_3-t_1}{2\pi \I\eps}}   \frac{e^{\frac{-\lambda_1^{(1)}+2\lambda_1^{(2)}+t_3}{4\eps}}\Gamma\left(1+\frac{\lambda_2^{(2)}-\lambda_1^{(2)}}{2\pi\I\eps}\right)\Gamma\left(\frac{\lambda_2^{(2)}-\lambda_1^{(2)}}{2\pi\I\eps}\right)}
    {\prod_{j=1}^3\Gamma\left(1+\frac{\lambda_j^{(3)}-\lambda_1^{(2)}}{2\pi\I\eps}\right)\cdot \Gamma\left(1+\frac{\lambda_2^{(2)}-\lambda_1^{(1)}}{2\pi\I\eps}\right)}\cdot \frac{A_{12}\Delta_{1,3}^{1,2}\left(\frac{A-\lambda_1^{(2)}}{2\pi\I\eps}\right)}{(\lambda_1^{(1)}-\lambda_1^{(2)})} \\ \nonumber
   (S^{\rm reg}_+)^2_{13}  &= 2\pi\I\left(\frac{u_2-u_1}{\eps}\right)^{\frac{t_3-t_1}{2\pi \I\eps}}\frac{ e^{\frac{-\lambda_1^{(1)}+2\lambda_2^{(2)}+t_3}{4\eps}}\Gamma\left(1+\frac{\lambda_1^{(2)}-\lambda_2^{(2)}}{2\pi\I\eps}\right)\Gamma\left(\frac{\lambda_1^{(2)}-\lambda_2^{(2)}}{2\pi\I\eps}\right)}
    {\prod_{j=1}^3\Gamma\left(1+\frac{\lambda_j^{(3)}-\lambda_2^{(2)}}{2\pi\I\eps}\right)\cdot \Gamma\left(1+\frac{\lambda_1^{(2)}-\lambda_1^{(1)}}{2\pi\I\eps}\right)}\cdot \frac{A_{12} \Delta_{1,3}^{1,2}\left(\frac{A-\lambda_2^{(2)}}{2\pi\I\eps}\right)}{(\lambda_1^{(1)}-\lambda_2^{(2)})}
\end{align*}

It follows from a direct computation of asymptotics of the gamma functions that
\begin{pro}\label{pro:3by3entryexample}
The entries of the Stokes matrices at $u_{\rm cat}(t)$ have the following asymptotics as $\eps\rightarrow 0^+$:
\begin{align*}
    (S_+^{\rm reg})_{11} &= e^{\frac{t_1}{2\eps}}, \hspace{3mm} (S_+^{\rm reg})_{22}=e^{\frac{t_2}{2\eps}}, \hspace{3mm} (S_+^{\rm reg})_{33}=e^{\frac{t_3}{2\eps}},\\
(S_+^{\rm reg})_{12} & \sim e^{\frac{\lambda^{(2)}_2+\I\theta_{12}(A)}{2\eps}}(f_{12}(A)+O(\eps))\\
(S_+^{\rm reg})_{23} & \sim e^{\frac{\lambda^{(3)}_2+\lambda^{(3)}_3-\lambda^{(2)}_2+\I\theta_{23}(A)}{2\eps}}(f_{23}(A)+O(\eps))+e^{\frac{\lambda^{(3)}_2+\lambda^{(3)}_3-t_2+\I\psi_{23}(A)}{2\eps}}(g_{23}(A)+O(\eps)), \\
(S_+^{\rm reg})_{13} & \sim e^{\frac{\lambda^{(3)}_3+\I\theta_{13}(A)}{2\eps}}(f_{13}(A)+O(\eps))+e^{\frac{\lambda^{(3)}_3-\lambda^{(1)}_1-\lambda^{(2)}_2+\lambda^{(2)}_{1}+\lambda^{(3)}_{2}+\I\psi_{13}(A)}{2\eps}}(g_{13}(A)+O(\eps)),
\end{align*}
for some real valued functions $f_{ij}$, $g_{ij}$, $\theta_{ij}$ and $\psi_{ij}$ of $A$.
\end{pro}
By the strict interlacing inequalities, we have $\lambda^{(2)}_{1}+\lambda^{(3)}_{2}<\lambda^{(1)}_1+\lambda^{(2)}_2$, so the first term in the asymptotics of $(S_+^{\rm reg})_{13}$ dominates.
However, if some interlacing inequalities become non-strict, for example $\lambda^{(2)}_{1}+\lambda^{(3)}_{2}=\lambda^{(1)}_1+\lambda^{(2)}_2$, then, since in general $\theta_{13}(A)\ne \psi_{13}(A)$, the WKB asymptotics of $(S_+^{\rm reg})_{13}$ may not exist. 

We also remark that the interlacing inequalities do not impose any relation between $\lambda^{(3)}_2+\lambda^{(3)}_3-\lambda^{(2)}_2$ and $\lambda^{(3)}_2+\lambda^{(3)}_3-t_2$. Thus depending on different choices of $A$, either term in the asymptotics of $(S_+^{\rm reg})_{23}$ may dominate.

\begin{pro}\label{pro:3by3minorexample}
Under the assumption that the interlacing inequalities are strict, the minors of the Stokes matrices at $u_{\rm cat}(t)$ have the following asymptotics as $\eps\rightarrow 0^+$:
\begin{align*}
    \Delta^{(1)}_1(S_+^{\rm reg}) &= e^{\frac{t_1}{2\eps}}, \hspace{3mm}  \Delta^{(2)}_2(S_+^{\rm reg})_{22}=e^{\frac{t_1+t_2}{2\eps}}, \hspace{3mm}  \Delta^{(3)}_3(S_+^{\rm reg})_{33}=e^{\frac{t_1+t_2+t_3}{2\eps}},\\
 \Delta^{(2)}_1(S_+^{\rm reg}) & \sim e^{\frac{\lambda^{(2)}_2+\I\theta^{(2)}_{1}(A)}{2\eps}}(f^{(2)}_{1}(A)+O(\eps))\\
 \Delta^{(3)}_1(S_+^{\rm reg}) & \sim e^{\frac{\lambda^{(3)}_3+\I\theta^{(3)}_{1}(A)}{2\eps}}(f^{(3)}_{1}(A)+O(\eps)), \\
 \Delta^{(3)}_2(S_+^{\rm reg}) & \sim e^{\frac{\lambda^{(3)}_2+\lambda^{(3)}_3+\I\theta^{(3)}_{2}(A)}{2\eps}}(f^{(3)}_{2}(A)+O(\eps)),
\end{align*}
for some real valued functions $f^{(k)}_{i}$ and $\theta^{(k)}_{i}$ of $A$.
\end{pro}
\begin{proof}
The minors $\Delta^{(1)}_1$, $\Delta^{(2)}_1$, $\Delta^{(2)}_2$, $\Delta^{(3)}_1$ and $\Delta^{(3)}_3$ are either entries or monomials in entries. So we only need to compute $\Delta^{(3)}_2=(S_+^{\rm reg})_{12}(S_+^{\rm reg})_{23}-(S_+^{\rm reg})_{22}(S_+^{\rm reg})_{13}$. Let us compute $(S_+^{\rm reg})_{12}(S_+^{\rm reg})^1_{23}$ and $(S_+^{\rm reg})_{12}(S_+^{\rm reg})^2_{23}$ respectively. Along the way, we use the reflection formula of $\Gamma$-function to simplify the expressions. First, we have
\begin{align*}
&(S_+^{\rm reg})_{12}(S_+^{\rm reg})^1_{23}\\
 &= \frac{A_{12}}{\eps} 2\pi\I \cdot \left(\frac{u_2-u_1}{\eps}\right)^{\frac{t_3-t_2}{2\pi \I\eps}}  \left(\frac{u_2-u_1}{\eps}\right)^{\frac{t_2-t_1}{2\pi \I\eps }} e^{\frac{t_1+t_2}{4\eps}}e^{\frac{{t_2+t_3}}{4\eps}}\\
&\times \frac{\Gamma\left(1+\frac{\lambda^{(2)}_2-\lambda^{(2)}_1}{2\pi \I\eps}\right)\Gamma\left(\frac{\lambda^{(2)}_2-\lambda^{(2)}_1}{2\pi \I\eps}\right)}{\Gamma\left(1+\frac{\lambda^{(1)}_1-\lambda^{(2)}_1}{2\pi \I\eps}\right)\Gamma\left(1+\frac{\lambda_1^{(2)}-\lambda^{(1)}_1}{2\pi \I\eps}\right)\Gamma\left(1+\frac{\lambda_2^{(2)}-\lambda^{(1)}_1}{2\pi \I\eps}\right)\prod_{j=1}^{3}\Gamma\left(1+\frac{\lambda^{(3)}_j-\lambda^{(2)}_1}{2\pi \I\eps}\right)}\cdot \Delta^{1,2}_{1,3}\left(\frac{A-\lambda^{(2)}_1}{2\pi\I\eps}\right)\\
&=  2\pi\I \cdot \left(\frac{u_2-u_1}{\eps}\right)^{\frac{t_3-t_1}{2\pi \I\eps}} {e^{\frac{t_1+t_2}{4\eps}}e^{\frac{{t_2+t_3}}{4\eps}}}\left(e^{\frac{\lambda^{(1)}_1-\lambda^{(2)}_1}{2\eps}} - e^{-\frac{\lambda^{(1)}_1-\lambda^{(2)}_1}{2\eps}}\right)\\
&\times \frac{\Gamma\left(1+\frac{\lambda^{(2)}_2-\lambda^{(2)}_1}{2\pi \I\eps}\right)\Gamma\left(\frac{\lambda^{(2)}_2-\lambda^{(2)}_1}{2\pi \I\eps}\right)}{\Gamma\left(1+\frac{\lambda_2^{(2)}-\lambda^{(1)}_1}{2\pi \I\eps}\right)\prod_{j=1}^{3}\Gamma\left(1+\frac{\lambda^{(3)}_j-\lambda^{(2)}_1}{2\pi \I\eps}\right)}\frac{A_{12}\Delta^{1,2}_{1,3}\left(\frac{A-\lambda^{(2)}_1}{2\pi\I\eps}\right)}{\lambda^{(1)}_1-\lambda^{(2)}_1}\\
&= 2\pi\I \cdot \left(\frac{u_2-u_1}{\eps}\right)^{\frac{t_3-t_1}{2\pi \I\eps}} {e^{\frac{3t_1+2t_2+t_3-2\lambda^{(2)}_1}{4\eps}}}\frac{\Gamma\left(1+\frac{\lambda^{(2)}_2-\lambda^{(2)}_1}{2\pi \I\eps}\right)\Gamma\left(\frac{\lambda^{(2)}_2-\lambda^{(2)}_1}{2\pi \I\eps}\right)}{\Gamma\left(1+\frac{\lambda_2^{(2)}-\lambda^{(1)}_1}{2\pi \I\eps}\right)\prod_{j=1}^{3}\Gamma\left(1+\frac{\lambda^{(3)}_j-\lambda^{(2)}_1}{2\pi \I\eps}\right)}\frac{A_{12}\Delta^{1,2}_{1,3}\left(\frac{A-\lambda^{(2)}_1}{2\pi\I\eps}\right)}{\lambda^{(1)}_1-\lambda^{(2)}_1}\\
&+ 2\pi\I \cdot \left(\frac{u_2-u_1}{\eps}\right)^{\frac{t_3-t_1}{2\pi \I\eps}} {e^{\frac{-t_1+2t_2+t_3+2\lambda^{(2)}_1}{4\eps}}}\frac{\Gamma\left(1+\frac{\lambda^{(2)}_2-\lambda^{(2)}_1}{2\pi \I\eps}\right)\Gamma\left(\frac{\lambda^{(2)}_2-\lambda^{(2)}_1}{2\pi \I\eps}\right)}{\Gamma\left(1+\frac{\lambda_2^{(2)}-\lambda^{(1)}_1}{2\pi \I\eps}\right)\prod_{j=1}^{3}\Gamma\left(1+\frac{\lambda^{(3)}_j-\lambda^{(2)}_1}{2\pi \I\eps}\right)}\frac{A_{12}\Delta^{1,2}_{1,3}\left(\frac{A-\lambda^{(2)}_1}{2\pi\I\eps}\right)}{\lambda^{(1)}_1-\lambda^{(2)}_1}
\end{align*}
Here in the second identity, we use the reflection formula of the gamma function to reduce
\[\frac{1}{\Gamma\left(1+\frac{\lambda^{(1)}_1-\lambda^{(2)}_1}{2\pi \I\eps}\right)\Gamma\left(1+\frac{\lambda_1^{(2)}-\lambda^{(1)}_1}{2\pi \I\eps}\right)}=\frac{{2\I\eps}\sin\left(\frac{\lambda^{(1)}_1-\lambda^{(2)}_1}{2 \I\eps}\right)}{{\lambda^{(1)}_1-\lambda^{(2)}_1}}=\eps\cdot \frac{e^{\frac{\lambda^{(1)}_1-\lambda^{(2)}_1}{2\eps}} - e^{-\frac{\lambda^{(1)}_1-\lambda^{(2)}_1}{2\eps}}}{\lambda^{(1)}_1-\lambda^{(2)}_1}.\]
Similarly, we have
\begin{align*}
&(S_+^{\rm reg})_{12}(S_+^{\rm reg})^2_{23}\\
&= 2\pi\I \cdot \left(\frac{u_2-u_1}{\eps}\right)^{\frac{t_3-t_1}{2\pi \I\eps}} {e^{\frac{3t_1+2t_2+t_3-2\lambda^{(2)}_2}{4\eps}}}\frac{\Gamma\left(1+\frac{\lambda^{(2)}_1-\lambda^{(2)}_2}{2\pi \I\eps}\right)\Gamma\left(\frac{\lambda^{(2)}_1-\lambda^{(2)}_2}{2\pi \I\eps}\right)}{\Gamma\left(1+\frac{\lambda_1^{(2)}-\lambda^{(1)}_1}{2\pi \I\eps}\right)\prod_{j=1}^{3}\Gamma\left(1+\frac{\lambda^{(3)}_j-\lambda^{(2)}_2}{2\pi \I\eps}\right)}\frac{A_{12}\Delta^{1,2}_{1,3}\left(\frac{A-\lambda^{(2)}_2}{2\pi\I\eps}\right)}{\lambda^{(1)}_1-\lambda^{(2)}_2}\\
&+ 2\pi\I \cdot \left(\frac{u_2-u_1}{\eps}\right)^{\frac{t_3-t_1}{2\pi \I\eps}} {e^{\frac{-t_1+2t_2+t_3+2\lambda^{(2)}_2}{4\eps}}}\frac{\Gamma\left(1+\frac{\lambda^{(2)}_1-\lambda^{(2)}_2}{2\pi \I\eps}\right)\Gamma\left(\frac{\lambda^{(2)}_1-\lambda^{(2)}_2}{2\pi \I\eps}\right)}{\Gamma\left(1+\frac{\lambda_1^{(2)}-\lambda^{(1)}_1}{2\pi \I\eps}\right)\prod_{j=1}^{3}\Gamma\left(1+\frac{\lambda^{(3)}_j-\lambda^{(2)}_2}{2\pi \I\eps}\right)}\frac{A_{12}\Delta^{1,2}_{1,3}\left(\frac{A-\lambda^{(2)}_2}{2\pi\I\eps}\right)}{\lambda^{(1)}_1-\lambda^{(2)}_2}.
\end{align*}
Therefore, from the explicit expressions of the entries we get
\begin{align*}
&(S_+^{\rm reg})_{12}(S_+^{\rm reg})_{23}-(S_+^{\rm reg})_{22}(S_+^{\rm reg})_{13}\\
    &= 2\pi\I \cdot \left(\frac{u_2-u_1}{\eps}\right)^{\frac{t_3-t_1}{2\pi \I\eps}} {e^{\frac{3t_1+2t_2+t_3-2\lambda^{(2)}_1}{4\eps}}}\frac{\Gamma\left(1+\frac{\lambda^{(2)}_2-\lambda^{(2)}_1}{2\pi \I\eps}\right)\Gamma\left(\frac{\lambda^{(2)}_2-\lambda^{(2)}_1}{2\pi \I\eps}\right)}{\Gamma\left(1+\frac{\lambda_2^{(2)}-\lambda^{(1)}_1}{2\pi \I\eps}\right)\prod_{j=1}^{3}\Gamma\left(1+\frac{\lambda^{(3)}_j-\lambda^{(2)}_1}{2\pi \I\eps}\right)}\frac{A_{12}\Delta^{1,2}_{1,3}\left(\frac{A-\lambda^{(2)}_1}{2\pi\I\eps}\right)}{\lambda^{(1)}_1-\lambda^{(2)}_1}\\
    &+ 2\pi\I \cdot \left(\frac{u_2-u_1}{\eps}\right)^{\frac{t_3-t_1}{2\pi \I\eps}} {e^{\frac{3t_1+2t_2+t_3-2\lambda^{(2)}_2}{4\eps}}}\frac{\Gamma\left(1+\frac{\lambda^{(2)}_1-\lambda^{(2)}_2}{2\pi \I\eps}\right)\Gamma\left(\frac{\lambda^{(2)}_1-\lambda^{(2)}_2}{2\pi \I\eps}\right)}{\Gamma\left(1+\frac{\lambda_1^{(2)}-\lambda^{(1)}_1}{2\pi \I\eps}\right)\prod_{j=1}^{3}\Gamma\left(1+\frac{\lambda^{(3)}_j-\lambda^{(2)}_2}{2\pi \I\eps}\right)}\frac{A_{12}\Delta^{1,2}_{1,3}\left(\frac{A-\lambda^{(2)}_2}{2\pi\I\eps}\right)}{\lambda^{(1)}_1-\lambda^{(2)}_2}
\end{align*}
Using the ordering $\lambda^{(3)}_1<\lambda^{(3)}_2<\lambda^{(3)}_3$ and $\lambda^{(2)}_1<\lambda^{(2)}_2$, and the strict interlacing inequalities, we can compute the leading asymptotics of the two summands to get
\[\Delta^{(3)}_2\sim e^{\frac{\lambda^{(3)}_2+\lambda^{(3)}_3+\I\theta^{(3)}_2(A)}{2\eps}}(f^{(3)}_2(A)+O(\eps))+e^{\frac{\lambda^{(3)}_3+\lambda^{(1)}_1+\lambda^{(2)}_1-\lambda^{(2)}_{2}+\I\psi^{(3)}_2(A)}{2\eps}}(g^{(3)}_2(A)+O(\eps)) \]
for some real-valued functions $\theta^{(3)}_2$, $f^{(3)}_2$ and $\psi^{(3)}_2$, $g^{(3)}_2$ of $A$. 

By the strict interlacing inequalities, we have $\lambda^{(3)}_2+\lambda^{(3)}_3>\lambda^{(3)}_3+\lambda^{(1)}_1+\lambda^{(2)}_1-\lambda^{(2)}_{2}$. Therefore the first term dominates. This completes the proof.
\end{proof}

\begin{rmk}
In the computation of the asymptotics of $(S_+^{\rm reg})_{12}(S_+^{\rm reg})_{23}-(S_+^{\rm reg})_{22}(S_+^{\rm reg})_{13}$, by using the reflection formula of the gamma function we decompose $(S_+^{\rm reg})_{12}(S_+^{\rm reg})_{23}$ into a summation of four terms. Two summands then cancel with $(S_+^{\rm reg})_{22}(S_+^{\rm reg})_{13}$. This cancellation exhibits the leading asymptotics of the minor as a linear combination of the $\lambda^{(k)}_i$.
\end{rmk}

\section{A quick review of spectral networks} \label{app:sn-review}

Here we recall the main facts and conjectures about spectral networks
which we use in the main text.
Most of this material can be found in \cite{GMN2,HN}.

\subsection*{Basic definitions}

\begin{defi}[Oriented foliation for a holomorphic 1-form]
For a Riemann surface $S$ equipped with a
nowhere vanishing holomorphic 1-form $\rho$, we have the distribution
$\ker \im \rho$, which integrates to give a foliation $F_{\rho}$ of $S$.
This foliation is oriented by $\re \rho$.
For example, if $S = \Comp$ and $\rho = dz$, then $F_\rho$ is the foliation
by horizontal lines oriented to the right.
More generally, for $\vartheta \in \Real / 2 \pi \bbZ$, we 
also define $$F^\vartheta_\rho = F_{e^{\I \vartheta} \rho} \, .$$

If $\rho$ has isolated zeroes, we make the same definitions, 
now getting a singular foliation $F_\rho$ of $S$.
\end{defi}

Fix a Riemann surface $C$ and a complex curve 
$\Gamma \subset T^* C$, such that the projection $\pi: T^* C \to C$
makes $\Gamma$ a degree $n$ branched covering of $C$.
In our application, we will take $C = \Comp^\times$
and $\Gamma = \Gamma(u,A)$.\footnote{We note a notational clash: in the 
literature on spectral networks $\Gamma$ is usually the name of the
charge lattice, but in this paper, we use $\Gamma$ to denote the spectral curve.}
We always assume that $\Gamma$ has simple ramification, so in particular
it is smooth.

\begin{defi}[Root curve]
The \emph{root curve} $\Gamma^r$
is the closure of the fiber product with the diagonal removed, 
\begin{equation}
\Gamma^r = \overline {\{ (y,y') \in \Gamma \times \Gamma: \pi(y) = \pi(y'), y \neq y \}} \subset T^* C \times T^* C.
\end{equation}
\end{defi}

$\Gamma^r$ is a smooth $n(n-1)$-sheeted branched covering of 
$C$.
We often think of a point of $\Gamma^r$ as a point $z \in C$ plus a choice of an
ordered pair $ij$ of distinct sheets of $\Gamma$ over $z$. 
$\Gamma^r$ carries a holomorphic $1$-form 
\begin{equation}
\rho = p_1^* \omega - p_2^* \omega    
\end{equation}
where $\omega$ is the Liouville $1$-form and $p_1$, $p_2$
are the two projections $T^* C \times T^* C \to T^* C$; $\rho$ vanishes
only at the ramification points.
Thus we have the foliations $F^\vartheta_\rho$ on $\Gamma^r$, with isolated singularities
at the branch points.

If $c$ is a 1-chain on $\Gamma^r$, let $p(c)$ be a 1-chain on $\Gamma$
given by
\begin{equation}
p(c) = (p_1)_* c - (p_2)_* c \, .
\end{equation}
(So if $c$ lies on the sheet of $\Gamma^r$ corresponding to the pair $ij$, 
then $p(c)$ consists of two components on $\Gamma$, one on sheet $i$ and one on sheet
$j$, with the second one oppositely oriented.)
Note it has 
\begin{equation}
  \int_{p(c)} \omega = \int_c \rho.
\end{equation}

\begin{defi}[Topological solitons]
Fix $z \in C$ and $y, y' \in \pi^{-1}(z)$. A \textit{topological soliton 
from $y$ to $y'$} is a 1-chain $c$ on $\Gamma^r$,
such that $\partial(p(c)) = y' - y$.
\end{defi}

The projection of a topological soliton from $\Gamma^r$ 
to $C$ looks like a graph (generically trivalent),
with one leaf at $z$ and all other leaves at branch points.
Each edge of the graph carries a label $ij$ keeping track of which sheet of
$\Gamma^r$ it came from. The labels are constrained:
there is a balancing condition at each internal vertex, and
also a condition at each branch point.

\begin{defi}[WKB solitons]
Fix $z \in C$ and $y, y' \in \pi^{-1}(z) \subset \Gamma$. 
A \textit{WKB soliton from $y$ to $y'$ with phase $\vartheta$} 
is a topological soliton from $y$ to $y'$
made up of positively oriented paths lying in leaves of the foliation 
$F^\vartheta_\rho$ (permitting
turns at the singularities).
\end{defi}

\begin{defi}[Finite web]
A \textit{finite web on $\Gamma$ with phase $\vartheta$} 
is a 1-chain $c$ on $\Gamma^r$,
made up of positively oriented paths in leaves of the foliation $F^\vartheta_\rho$ (permitting
turns at the singularities),
with $\partial(p(c)) = 0$.
The \textit{charge} of the finite web $c$ is the class $[p(c)] \in H_1(\Gamma)$.
\end{defi}

\begin{pro}[Phases of periods of finite webs] \label{prop:finite-web-period-phase}
If $\gamma$ is the charge of a finite web with phase $\vartheta$,
then $\arg (-Z(\gamma)) = \vartheta$.
\end{pro}

\begin{defi}[WKB spectral network] The \textit{WKB spectral network} 
$\cW(\Gamma,\vartheta)$
is the set of all $(y,y') \in \Gamma^r$ such that there exists a WKB soliton 
from $y$ to $y'$ with phase $\vartheta$.
\end{defi}

Each point of $\cW(\Gamma,\vartheta)$ thus carries the additional data
of the set of such WKB solitons.
There is an algorithm for computing $\cW(\Gamma,\vartheta)$ described
in \cite{GMN2} and implemented in a Mathematica notebook included 
with the arXiv version of this paper.
It realizes $\cW(\Gamma,\vartheta)$ as a collection of walls in $\Gamma^r$.
The projection of $\cW(\Gamma,\vartheta)$ to $C$ is a collection of walls 
as well, with each wall carrying a label $ij$ keeping track of a sheet in $\Gamma^r$.
We sometimes write this label as $i < j$. This notation keeps track of two facts:
first, the arrow pointing from $j$ to $i$ reminds us that the WKB solitons are paths beginning on sheet $j$ and ending on sheet $i$ \cite{BoalchTop}; second,
in the application of spectral networks to exact WKB, we have an ODE on the surface $C$, and a
a basis of solutions $\psi^{(k)}$ indexed by the sheets; 
then along the wall we have $\psi^{(i)} \ll \psi^{(j)}$ as $\eps \to 0$
with $\arg \eps = \vartheta$.

\subsection*{Path lifting}

\begin{defi}[Path categories] For $Y \subset X$:
\begin{itemize}
\item Let $\Path(X,Y)$ be the category of 
paths in $X$ with endpoints in $Y$, enriched over abelian groups 
(so the objects are points $y \in Y$, and $\Hom(y,y')$ is the abelian group of formal $\bbQ$-linear combinations
of paths from $y$ to $y'$ in $X$, with composition extended linearly).
\item 
Given a map $\pi: Z \to X$, let $\Path^Z(X,Y)$ be the category of 
paths in $Z$ with endpoints in $\pi^{-1}(Y)$
(so the objects are points $y \in Y$, and $\Hom(y,y')$
is the abelian group of formal $\bbQ$-linear combinations of paths
from any $z \in \pi^{-1}(y)$ to any $z' \in \pi^{-1}(y')$ in $Z$, 
with composition extended linearly, and with the 
composition of paths that do not concatenate taken to be zero).
\end{itemize}
\end{defi}

\begin{defi}[Path-lifting functors]
For $\cW \subset \Gamma^r$, a \emph{path-lifting functor off $\cW$} is a map
\begin{equation}
\Lift: \Path(C,C \setminus \pi(\cW)) \to \Path^\Gamma(C,C \setminus \pi(\cW))
\end{equation}
which is almost-homotopy-invariant, i.e. if $\cP \sim \cP'$ then $\Lift(\cP) \sim^{al} \Lift(\cP')$, where $\sim^{al}$ means ordinary homotopy except that if we move a path across a branch point we multiply by $-1$.
\end{defi}

\begin{pro}[WKB path-lifting functor] \label{prop:path-lifting}
There exists a path-lifting functor off $\cW(\Gamma,\vartheta)$, $\Lift_{\cW(\Gamma,\vartheta)}$, with the following property:
all terms in $\Lift_{\cW(\Gamma,\vartheta)}(\cP)$ are 
obtained by splicing paths $p(c)$, where $c$ is a WKB soliton with phase $\vartheta$, into lifts of $\cP$.
\end{pro}

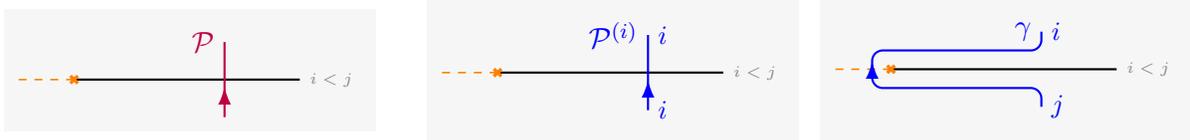
\begin{figure}[h]
\centering
\begin{tikzpicture}[baseline=(current bounding box.center),withbackgroundrectangle]
    \draw[wall] (0,0) -- (3,0);
    \draw[branchcut] (0,0) -- (-0.75,0);
    \draw[path,purple,witharrow=0.4] (2,-0.5) -- (2,0.5);
    \node[left,purple] at (2,0.5) {$\cP$};
    \draw[branchpoint] plot coordinates {(0,0)};
    \node[right,stokeslabel] at (3,0) {$i<j$};
\end{tikzpicture}
\hspace{0.5cm}
\begin{tikzpicture}[baseline=(current bounding box.center),withbackgroundrectangle]
    \draw[wall] (0,0) -- (3,0);
    \draw[branchcut] (0,0) -- (-0.75,0);
    \draw[path,blue,witharrow=0.4] (2,-0.5) -- (2,0.5);
    \node[left,blue] at (2,0.5) {$\cP^{(i)}$};
    \node[right,blue,sheetlabel] at (2,-0.5) {$i$};
    \node[right,blue,sheetlabel] at (2,0.5) {$i$};
    \draw[branchpoint] plot coordinates {(0,0)};
    \node[right,stokeslabel] at (3,0) {$i<j$};
\end{tikzpicture}
\hspace{0.1cm}
\begin{tikzpicture}[baseline=(current bounding box.center),withbackgroundrectangle]
    \draw[wall] (0,0) -- (3,0);
    \draw[branchcut] (0,0) -- (-0.75,0);
    \draw[path,blue,witharrow=0.52] (2,-0.5) -- (2,-0.25) -- (-0.25,-0.25) -- (-0.25,0.25) -- (2,0.25) -- (2,0.5);
    \node[left,blue] at (2,0.5) {$\gamma$};
    \node[right,blue,sheetlabel] at (2,-0.5) {$j$};
    \node[right,blue,sheetlabel] at (2,0.5) {$i$};
    \draw[branchpoint] plot coordinates {(0,0)};
    \node[right,stokeslabel] at (3,0) {$i<j$};
\end{tikzpicture}

\caption{Left: a path $\cP$ crossing one primary wall on $C$. Middle, right: paths on $\Gamma$ which arise as terms in $\Lift(\cP)$, as indicated in \autoref{prop:path-lifting-unique} and \autoref{prop:path-lifting-canonical}.}
\label{fig:lifts-crossing-one-wall}
\end{figure}

\begin{pro} \label{prop:path-lifting-unique}
Suppose that there are no finite webs on $\Gamma$ 
with phase $\vartheta$. Then the path-lifting functor $\Lift_{\cW(\Gamma,\vartheta)}$ 
in \autoref{prop:path-lifting} is unique,
and determined by an algorithm given in \cite{GMN}.
It has the following additional property. 
Call a wall of $\cW(\Gamma,\vartheta)$ 
\textit{primary} if 
it originates from a branch point.
Suppose that $\cP$ is a path crossing a single primary wall of $\cW(\Gamma,\vartheta)$. 
Then
\begin{equation}
 \Lift_{\cW(\Gamma,\vartheta)}(\cP) = \sum_{i=1}^n \cP^{(i)} + \gamma
\end{equation}
where $\cP^{(i)}$ and $\gamma$ are shown in \autoref{fig:lifts-crossing-one-wall}.
\end{pro}

The extra term $\gamma$ in $\Lift(\cP)$ is sometimes referred to as a ``detour'' path: 
it travels along $\cP$ to the wall,
then takes a detour along the wall to the branch point and back, then resumes along $\cP$ again.

\begin{pro} \label{prop:path-lifting-canonical}
If there are finite webs on $\Gamma$ with phase $\vartheta$, then the
path-lifting functor in \autoref{prop:path-lifting} is not unique, but there are two canonical choices $\Lift_{\cW(\Gamma,\vartheta^\pm)}$, again 
described in \cite{GMN}. They have 
the following additional property. Call a wall of $\cW(\Gamma,\vartheta)$ 
\textit{primary one-way} if 
it originates from a branch point and ends at a singularity.
Suppose that $\wp$ is a path crossing only one wall, and
that wall is a primary one-way wall. Then
\begin{equation} \label{eq:lift-pm}
 \Lift_{\cW(\Gamma,\vartheta^\pm)}(\wp) = \sum_{i=1}^n \cP^{(i)} + \gamma + \cdots
\end{equation}
where $\cP^{(i)}$ and $\gamma$ are shown in \autoref{fig:lifts-crossing-one-wall}, 
and all paths $\mu$ in the $\cdots$ terms have
$\arg(Z(\gamma) - Z({\mu})) = \vartheta$.
\end{pro}

In the situation of this paper, we will be particularly interested in the phase $\vartheta = 0$, where we do have finite webs,
and we will use either of the path-lifting functors $\Lift_{\cW(\Gamma,\vartheta^\pm)}$ from \autoref{prop:path-lifting-canonical}.\footnote{There is also a more canonical choice, which is a kind of geometric mean of $\Lift_{\cW(\Gamma,\vartheta^\pm)}$, used e.g. in \cite{HN}. That choice has better symmetry properties; in the WKB context, it is related to 
using median summation instead of lateral summation. However, nothing we do in this paper is sensitive to whether we take $\Lift_{\cW(\Gamma,\vartheta^+)}$, $\Lift_{\cW(\Gamma,\vartheta^-)}$ 
or the geometric mean; that choice only affects the $\cdots$ terms in
\eqref{eq:lift-pm}.}
Because $\Lift_{\cW(\Gamma,\vartheta^\pm)}$ 
is compatible with composition of paths, we can 
use \autoref{prop:path-lifting-canonical} to compute
$\Lift_{\cW(\Gamma,\vartheta^\pm)}(\cP)$ for any $\cP$ which crosses only primary one-way walls.

\subsection*{Nonabelianization of local systems}

\begin{defi}[Almost-local systems] Given a branched covering map $\pi: Z \to X$, an almost-local system $L^\ab$ over $Z$ is a local system over the complement of the branch locus in $Z$, which has monodromy $-1$ around each branch point.
\end{defi}

\begin{defi}[Nonabelianization map]
Fix a path-lifting functor $\Lift$ off $\cW$.
Given an almost-local system $L^\ab$ over $\Gamma$, we define
a local system
\begin{equation}
L = \Nab_{\Lift}(L^\ab) 
\end{equation}
over $C \setminus \pi(\cW)$
as follows: 
for an open set $U$, $L(U) = L^\ab(\pi^{-1}(U))$;
for a path $\cP$ from $U$ to $U'$, the map
$L(\cP): L(U) \to L(U')$ is
\begin{equation}
    L(\cP) = L^\ab(\Lift(\cP)) \, .
\end{equation}
The local system $L$ so defined extends over $C$.
\end{defi}

\subsection*{The case $n=2$}

The case $n=2$ is particularly simple.
In this case $\Gamma^r = \Gamma$, and the oriented foliation $F_\vartheta$ of $\Gamma^r$ descends 
to an unoriented foliation of $C$.
Then the WKB spectral network $\cW(\Gamma,\vartheta)$ is the
critical graph of a quadratic differential on $C$,
as described in \cite{GMN}.
The relevant quadratic differential is given in terms of the sheets
$y_1$, $y_2$ of $\Gamma$ by
\begin{equation}
\varphi_2 = e^{-\I \vartheta} (y_1 - y_2)^2 \, .    
\end{equation}
The zeroes of $\varphi_2$ are the branch points of $\Gamma$;
$\Gamma$ is smooth just if all these are simple.
Then the critical graph consists of three trajectories 
emerging from each branch point, characterized by the condition
that $\int \sqrt{\varphi_2}$ is real along each trajectory.
Finite webs occur only when there is a trajectory which runs between
branch points: such a trajectory is called
a saddle connection (if the two branch points are distinct)
or a closed loop (if they are the same).

\subsection*{The WKB conjecture}

Now suppose given a family of flat $\GL(n)$-connections $\nabla(\eps)$ in bundles $\cE(\eps)$ over $C$,
such that as $\eps \to 0$ the $(\cE(\eps), \eps \nabla(\eps))$ 
limit to a Higgs bundle $(\cE,\varphi)$ 
over $C$. Let $\Gamma$ be the spectral curve of $(\cE,\varphi)$.
Then the exact WKB method as employed in \cite{GMN2,HN} predicts:
\begin{conj}[Exact WKB conjecture for closed paths] \label{conj:gmn-closed-only}
There is a family $\nabla^{\ab,\vartheta}(\eps)$ of $\GL(1)$-connections in bundles $\cE^{\ab,\vartheta}(\eps)$
over $\Gamma$, flat except for monodromy $-1$ around branch points.
If there exist no finite webs on $\Gamma$ of phase $\vartheta$, then:
\begin{enumerate}
\item  
the local systems associated to
connections $(\cE(\eps),\nabla(\eps))$ and $(\cE^\ab(\eps),\nabla^{\ab,\vartheta}(\eps))$ are
related by the nonabelianization map $\Nab_{\Lift_{\cW(\Gamma,\vartheta)}}$.
\item For any cycle $\gamma \in H_1(\Gamma)$,
the holonomy $X^\vartheta_\gamma(\eps)$ of $\nabla^{\ab,\vartheta}(\eps)$
obeys $X^\vartheta_\gamma(\eps) \sim \exp(Z(\gamma) / \eps + \I \phi_\gamma)$
as $\eps \to 0$ in the half-plane centered on $\arg \eps = \vartheta$.
\end{enumerate}
If there exist finite webs on $\Gamma$ of phase $\vartheta$, then we 
similarly define two different
holonomies $X_\gamma^{\vartheta^\pm}(\eps)$
using parallel transports of $\nabla^{\ab,\vartheta^\pm}(\eps)$,
and either of them obeys
$X^{\vartheta^\pm}_\gamma(\eps) \sim \exp(Z(\gamma) / \eps + \I \phi_\gamma)$.
\end{conj}

\subsection*{Open paths}

In this paper, we need a small extension of the above story to include open paths.
This extension was not written in \cite{GMN2,HN} (though one special case
appeared in \cite{GMN}), so we formulate it here.
We consider the situation where the connections 
$\nabla(\eps)$ have an irregular singularity at a point
$z^* \in \overline{C}$, such that the leading term in the expansion of
$\varphi$ is regular semisimple. 
We take the special case where the leading term is real.
(This condition is satisfied in our example, where $z^* = \infty$, and the leading term is $u$.)
As we discussed in \autoref{beginsection}, there are
functions $P_i$, $Q_i$ and $\nabla(\eps)$-flat sections $\psi_i$
in neighborhoods of anti-Stokes rays, such that $\exp(-P_i(z) / \eps + Q_i(z)) \psi_i(z)$ is finite as $z \to z^*$
along an anti-Stokes ray.
The Stokes matrices can be understood as parallel transport matrices, relative to the bases $\psi_i$, 
along arcs which stay in a small neighborhood of $z^*$ and run from one anti-Stokes ray to another.

In this case, the spectral curve $\Gamma$ is unramified around $z^*$.
Computing asymptotics of 
the Stokes data around $z^*$ requires consideration
of open paths, with endpoints among the $n$ preimages $z^*_i$, and coming
into $z^*_i$ along anti-Stokes rays.
To define their periods requires regularization,
since $\omega$ has a pole at $z^*_i$.
Thus suppose $\gamma$ is an open path from $z^*_{i}$ to $z^*_{i'}$,
coming in along anti-Stokes rays.
We define a regularization $\gamma_\reg$ of $\gamma$, by perturbing the endpoints $z^*_i$, $z^*_{i'}$ to nearby points $z_i$, $z'_{i'}$
on the anti-Stokes rays, and then let
\begin{equation} \label{eq:open-period-def-appendix}
  Z(\gamma) = \int^\reg_\gamma \omega =  \lim_{z_i \to z^*_i} \lim_{z'_{i'} \to z^*_{i'}} \left( - P(z'_{i'}) + P(z_i) + \int_{\gamma_\reg} \omega \right) \, .
\end{equation}

Then we have an extension of \autoref{conj:gmn-closed-only},
as follows:

\begin{defi}[Extended homology]
Consider the real oriented blow-up $\widehat\Gamma$ of $\Gamma$ at 
the $n$ preimages of $z^*$.
Then $\widehat\Gamma$ has boundary consisting of $n$ circles.
On each of these $n$ circles, we mark points corresponding to the
anti-Stokes rays; let $L(\vartheta)$ 
be the set of marked points. Then define
\begin{equation}
  H(\Gamma,\vartheta) = H_1(\widehat\Gamma ; L(\vartheta)) \, .
\end{equation}
\end{defi}

\begin{defi}[Spectral coordinates]
If $\Gamma$ has no finite webs of phase $\vartheta$, the \textit{spectral coordinates} of the connection $\nabla(\eps)$
are the parallel transports $X_\gamma^\vartheta(\eps) \in \Comp^\times$ of $\nabla^{\ab,\vartheta}(\eps)$ on paths $\gamma \in H(\Gamma,\vartheta)$. When $\gamma$ 
is an open path, the parallel transport is taken relative to 
the basis vectors $\psi_i$ at the two ends of the path.
If $\Gamma$ has finite webs of phase $\vartheta$, then we define two different
spectral coordinates $X_\gamma^{\vartheta^\pm}(\eps)$
using parallel transports of $\nabla^{\ab,\vartheta^\pm}(\eps)$.
\end{defi}

\begin{conj}[Exact WKB conjecture for open and closed paths] \label{conj:gmn-open-closed} 
The spectral coordinates $X_\gamma^\vartheta(\eps)$ obey 
$X^\vartheta_\gamma(\eps) \sim \exp(Z(\gamma) / \eps + \I \phi_\gamma)$
as $\eps \to 0$ in the half-plane centered on $\arg \eps = \vartheta$.
\end{conj}

\subsection*{DT-type invariants}

As we have explained, when there are finite webs with phase $\vartheta_0$,
the path-lifting functor at phase $\vartheta_0$ is not unique: there
are two canonical choices $\Lift_{\cW(\Gamma,\vartheta_0^\pm)}$.
These two can be understood as the limits of $\Lift_{\cW(\Gamma,\vartheta)}$ as $\vartheta \to \vartheta_0^\pm$. The limits $X_\gamma^{\vartheta_0^\pm}$ of the 
spectral coordinates need not agree: in general they
differ by some coordinate transformation.
This coordinate transformation is one of the most
important structural aspects of the spectral coordinates. 
In the generic situation, it
can be written in terms of ``Donaldson-Thomas-type
invariants'' $\Omega(\gamma)$, as follows:
\begin{defi}[Mutually local phase]
We say $\vartheta_0$ is a \textit{mutually local} phase if,
whenever $\gamma$ and $\gamma'$ are charges of finite webs
with phase $\vartheta_0$, we have $\IP{\gamma,\gamma'} = 0$.
\end{defi}

\begin{pro}[Transformation law for spectral coordinates]
When $\vartheta$ is a mutually local phase, let $B_\vartheta$ be the set of charges of finite webs with phase $\vartheta$; 
then there exists $\Omega: B_\vartheta \to \bbQ$ such that for all
$\mu$
\begin{equation} \label{eq:spectral-jump-formula}
    X_\mu^{\vartheta^+} = X_\mu^{\vartheta^-} \prod_{\gamma \in B_\vartheta} (1 \pm X_\gamma)^{\IP{\mu,\gamma} \Omega(\gamma)} \, .
\end{equation}
\end{pro}

An algorithm for computing the $\Omega(\gamma)$ is given
in \cite{GMN,HN}. Two important special cases, which we encounter
in this paper: first, if there are no finite webs with charge $\mu$
then $\Omega(\mu) = 0$; second, an isolated tree
with charge $\mu$ contributes $1$ to $\Omega(\mu)$.

\subsection*{Equivalences}

Path-lifting maps admit a notion of equivalence:

\begin{defi}
An \emph{equivalence} between path-lifting maps $\Lift_0$, $\Lift_1$ to
curves $\Gamma_0$, $\Gamma_1$ is a family of 
curves $\Gamma_t$ and path-lifting maps $\Lift_t$,
such that for every $p \in C$ and pair $t_1, t_2 \in [0,1]$
there exists $R(p,t_1,t_2) \in \Path^\Gamma(C,C \setminus \pi(\cW))$,
such that for paths $\cP$ from $p$ to $p'$ we have
\begin{equation}
 \Lift_{t_2}(\cP) = R(p,t_1,t_2) \Lift_{t_1}(\cP) R(p',t_1,t_2)^{-1} \, .
\end{equation}
\end{defi}
If $\Lift$ and $\Lift'$ are equivalent, the corresponding nonabelianization
maps $\Nab_{\Lift}$ and $\Nab_{\Lift'}$ are also equivalent, in the sense
that for any $L^\ab$, $\Nab_{\Lift}(L^\ab)$ and $\Nab_{\Lift'}(L^\ab)$ are equivalent local systems.

Here is one important source of equivalences:

\begin{pro} \label{pro:equivalences-from-variations}
Fix some $\vartheta$ and
consider a variation $\Gamma_t$ such that
there is no finite web whose phase crosses $\vartheta$.
Then the path-lifting maps $\Lift_t$ give an
equivalence between $\Lift_0$ and $\Lift_1$.
\end{pro}

Another way to get an equivalence is to 
start with the spectral network $\cW(\Gamma,\vartheta)$ and modify it
by certain topological moves, which include and generalize isotopies, as described in \cite{GMN} section 10.6.
After these moves, we obtain a 
new path-lifting map $\Lift'$ off a ``topological spectral network''
$\cW'$. $\Lift'$ is equivalent to $\Lift$ but sometimes
more convenient.

If no walls of the spectral network move across $p$ or $p'$ between $t_1$ and $t_2$, then $R(p,t_1,t_2) = R(p',t_1,t_2) = 1$, and so for paths $\cP$
from $p$ to $p'$ we have just
$\Lift(\cP) = \Lift'(\cP)$.
We use one such equivalence in the $n=3$ example in the main text, where 
$p$ and $p'$ are points near infinity on the anti-Stokes lines.

\Addresses

\begin{thebibliography}{}
\addtolength{\itemsep}{-1.5 em}
\setlength{\itemsep}{-5pt}


\bibitem{AD}
A. Alekseev and I. Davydenkova, {\em Inequalities from Poisson brackets}. Indag. Math. (N.S.) 25, 846–871 (2014).


\bibitem{ABHL}
A. Alekseev, A. Berenstein, B. Hoffman and Y. Li, {\em Langlands duality and Poisson-Lie duality via cluster theory and tropicalization}, Selecta Math. (N.S.) 27 (2021), no. 4, paper no. 69.

\bibitem{ALL}
A. Alekseev, J. Lane and Y. Li, {\em The $U(n)$ Gelfand-Tsetlin system as a tropical limit of Ginzburg-Weinstein diffeomorphisms}, Phil. Trans. R. Soc. A 376: 20170428 (2018). 

\bibitem{APS}
A. Alekseev, M. Podkopaeva, A. Szenes,
{\em A symplectic proof of the Horn inequalities},
Adv. Math. 318 (2017) 711-736.


\bibitem{AHKKNSST}
T. Aoki, N. Honda, T. Kawai, T. Koike, Y. Nishikawa, S. Sasaki, A. Shudo and
Y. Takei, {\em Virtual turning points—a gift of microlocal analysis to the exact WKB
analysis,} in Algebraic analysis of differential equations from microlocal analysis to
exponential asymptotics, pp. 29–43. Springer, Tokyo, 2008.

\bibitem{AKT2}
T. Aoki, T. Kawai, and Y. Takei, {\em New turning points in the exact WKB analysis for higher order ordinary differential equations}, Analyse alg\'ebrique des perturbations singulieres, I, M\'ethodes r\'esurgentes, Hermann, 1994, pp. 69-84.


\bibitem{Balser}
W. Balser, {\it Formal power series and linear systems of meromorphic ordinary differential equations},
Springer-Verlag, New York, 2000.


\bibitem{BJL}
W. Balser, W.B. Jurkat, and D.A. Lutz, {\em Birkhoff invariants and Stokes' multipliers for meromorphic linear differential equations}, J. Math. Anal. Appl. 71 (1979), 48-94.

\bibitem{BNR}
H. L. Berk, W. M. Nevins and K. V. Roberts, {\em New Stokes' line in WKB theory}, J. Math. Phys., 23 (1982), pp. 988-1002.


\bibitem{BB}
O. Biquard and P. Boalch,  {\em Wild non-abelian Hodge theory on curves}, Compositio Mathematica, 140(1), 179-204 (2004). 

\bibitem{Boalch} P. Boalch, {\em Stokes matrices, Poisson Lie groups and Frobenius manifolds}, Invent. Math. 146 (2001), 479–506.
no. 3, 479-506.


\bibitem{BoalchG} P. Boalch, {\em G-bundles, isomonodromy and quantum Weyl groups}, Int. Math. Res. Not. (2002), no. 22, 1129–1166.

\bibitem{BoalchTop} P. Boalch, {\em Topology of the Stokes phenomenon.} In Proc. Sympos. Pure Math. 103.1 (2021), American Mathematical Society, 55-100.

\bibitem{Bridgeland}
T. Bridgeland, {\em Riemann-Hilbert problems from Donaldson-Thomas theory}, Invent. Math. 216 (2019), 69-124. 

\bibitem{BTL}
T. Bridgeland and V. Toledano Laredo, {\em Stability conditions and Stokes factors}, Invent. Math. 187 (2012), 61-98.

\bibitem{dCP}
C. De Concini and C. Procesi, {\em Wonderful models of subspace arrangements}, Selecta
Math. (N.S.) 1 (1995), no. 3, 459–494.

\bibitem{DDP}
E. Delabaere, H. Dillinger, and F. Pham, {\em \'{R}esurgence de Voros et \'{p}eriodes des courbes hyperelliptiques}, Ann. Inst. Fourier (Grenoble) 43 (1993), 163–199.

\bibitem{Dubrovin}
B. Dubrovin, {\it Geometry of 2D topological field theories}, Lecture Notes in Math, 1620 (1995).

\bibitem{FR} H. Flaschka, T. Ratiu, 
{\em A convexity theorem for Poisson actions of compact Lie groups}, IHES,
Preprint 1995. Available at http://preprints.cern.ch.

\bibitem{Gaiotto}
D. Gaiotto, {\em Opers and TBA}, arXiv:1403.6137.

\bibitem{GMN}
D. Gaiotto, G. W. Moore, and A. Neitzke, {\em Wall-crossing, Hitchin systems, and the WKB
approximation}, Adv. in Math. 234 (2013), 239–403.

\bibitem{GMN2}
D. Gaiotto, G. W. Moore, and A. Neitzke, {\em Spectral networks}, Ann. Henri
Poincar\'e, 14(7):1643–1731, 2013.

\bibitem{GMNY}
D.~Gaiotto, G.~W.~Moore, A.~Neitzke and F.~Yan, {\em Commuting Line Defects At $q^N=1$},
arXiv:2307.14429.

\bibitem{GGI}
S.Galkin, V. Golyshev, H. Iritani. {\em Gamma classes and quantum cohomology of Fano manifolds: Gamma conjectures}, Duke Math. J. 165 (11) 2005 - 2077 (2016).

\bibitem{Gamayun}
O.~Gamayun, N.~Iorgov and O.~Lisovyy. {\em Conformal field theory of Painlev\'e VI},
JHEP \textbf{10}, 038 (2012).

\bibitem{GGM}
A. Grassi, J. Gu, and M. Mari˜no, {\em Non-perturbative approaches to the quantum
Seiberg-Witten curve}, JHEP 07 (2020) 106.

\bibitem{GoSh}
A. Goncharov and L. Shen. {\em Donaldson-Thomas transformations of moduli spaces of $G$-local systems}, Adv. Math. 327, 225-348 (2018).

\bibitem{GS}
V. Guillemin and S. Sternberg, {\em The Gelfand-Cetlin system and quantization of the complex flag
manifolds}, J. Funct. Anal 52 (1983), 106–128.

\bibitem{GW} V. Ginzburg and A. Weinstein, {\em Lie-Poisson structure on some Poisson-Lie groups}, J. Amer. Math. Soc. 5:2
(1992), 445-53.


\bibitem{HKRW}
I. Halacheva, J. Kamnitzer, L. Rybnikov and A. Weekes, {\em Crystals and monodromy of Bethe vectors}, Duke Math. J. 169 (12) 2337-2419.


\bibitem{HN}
L. Hollands and A. Neitzke, {\em Exact WKB and Abelianization for the $T_3$ Equation,} Commun. Math. Phys. volume 380, 131–186 (2020).


\bibitem{ILT}
N.~Iorgov, O.~Lisovyy and J.~Teschner,
{\em Isomonodromic tau-functions from Liouville conformal blocks},
Commun. Math. Phys. 336, no.2, 671-694 (2015).

\bibitem{IKKS}
K. Ito, T. Kondo, K. Kuroda, and H. Shu, {\em WKB periods for higher order ODE and TBA
equations}, arXiv:2104.13680.

\bibitem{IN}
K. Iwaki and T. Nakanishi, {\em Exact WKB analysis and cluster algebras}, J. Phys. A 47, 474009 (2014).


\bibitem{KT}
A.-K. Kashani-Poor and J. Troost, {\em Pure ${\mathcal N}=2$ super Yang–Mills and exact WKB}, JHEP
08 (2015) 160.


\bibitem{Kato} 
T. Kato, {\em Perturbation theory for linear operators}, Classics in Mathematics volume 132, Springer, 1995.

\bibitem{LR}
M. Loday-Richaud, {\em Divergent series, summability and resurgence II. Simple and multiple summability}, vol. 2154 of
Lecture Notes in Mathematics, Springer, 2016.

\bibitem{LW}
J.-H. Lu and A. Weinstein, {\em Poisson Lie groups, dressing transformations, and Bruhat decompositions}, Journal of Differential Geometry 31 (1990), no. 2, 501-526.

\bibitem{Nekrasov}
N. Nekrasov, {\em Seiberg-Witten Prepotential from Instanton Counting}, Adv. Theor. Math. Phys. 7 (5) 831-864, September 2003.

\bibitem{NRS}
N. Nekrasov, A. Rosly, and S. Shatashvili, {\em Darboux coordinates, Yang-Yang functional, and
gauge theory}, Theor. Math. Phys. 181 (2014), no. 1, 1206–1234.

\bibitem{OV}
H. Ooguri and C. Vafa, {\em Summing up D-instantons}, Phys. Rev. Lett. 77 (1996) 3296–3298.

\bibitem{Silverstone}
H. J. Silverstone, {\em JWKB connection-formula problem revisited via Borel summation}, Phys. Rev. Lett., 55(1985), 2523-2526.

\bibitem{STS} M. Semenov-Tian-Shansky,
{\em Dressing transformations and Poisson group actions},
Publ. Res. Inst. Math. Sci. 21 (1985) no.6, 1237-1260.

\bibitem{Sp}
D. Speyer, {\em Schubert problems with respect to oscillating flags of stable rational curves}, Alg. Geom. 1 (2014), no. 1, 14–45.

\bibitem{Tulli}
I. Tulli, {\em The Ooguri–Vafa space as a moduli space of framed wild harmonic bundles},
arXiv:1912.00261.

\bibitem{voisin II}
C. Voisin, {\em Hodge Theory and Complex Algebraic Geometry II}, trans. Leila Schneps, Cambridge: Cambridge University Press, 2003, print, Cambridge Studies in Advanced Mathematics.

\bibitem{Vo}
A. Voros, {\em The return of the quartic oscillator. the complex WKB method}, Ann. Inst. Henri Poincar\'{e} 39 (1983), 211–338.

\bibitem{Wasow}
W. Wasow, {\em Asymptotic expansions for ordinary differential equations}, Wiley Interscience, New York, 1976.

\bibitem{Xu}
X. Xu, {\em Regularized limits of Stokes matrices, isomonodromy deformation and crystal basis}, arXiv:1912.07196v5.


\bibitem{Xu2}
X. Xu, {\em Representations of quantum groups arising from Stokes phenomenon}, arXiv: 2012.15673.


\end{thebibliography}
\end{document}